\documentclass[12pt]{article}
\usepackage [english]{babel}
\usepackage[utf8x]{inputenc}

\usepackage{amssymb} 
\usepackage{amsmath} 
\usepackage{amsthm} 

\usepackage{bigints}
\usepackage{float}
\usepackage{bbm}
\usepackage{geometry}
\usepackage{natbib}
\usepackage{newtxmath}

\usepackage{enumerate}
\usepackage{setspace} 

\usepackage{pgfplots}

\usepackage{tikz}
\usetikzlibrary{arrows,automata, fit, shapes, patterns, calc, intersections, decorations.pathreplacing, pgfplots.fillbetween}

\usepackage[breaklinks,backref,pagebackref]{hyperref}

\definecolor{solarizedBase03}{HTML}{002B36}
\definecolor{solarizedBase02}{HTML}{073642}
\definecolor{solarizedBase01}{HTML}{586e75}
\definecolor{solarizedBase00}{HTML}{657b83}
\definecolor{solarizedBase0}{HTML}{839496}
\definecolor{solarizedBase1}{HTML}{93a1a1}
\definecolor{solarizedBase2}{HTML}{EEE8D5}
\definecolor{solarizedBase3}{HTML}{FDF6E3}
\definecolor{solarizedYellow}{HTML}{B58900}
\definecolor{solarizedOrange}{HTML}{CB4B16}
\definecolor{solarizedRed}{HTML}{DC322F}
\definecolor{solarizedMagenta}{HTML}{D33682}
\definecolor{solarizedViolet}{HTML}{6C71C4}
\definecolor{solarizedBlue}{HTML}{268BD2}
\definecolor{solarizedCyan}{HTML}{2AA198}
\definecolor{solarizedGreen}{HTML}{859900}

\theoremstyle{definition}
\newtheorem{definition}{Definition}[section]

\theoremstyle{definition}

\theoremstyle{plain}
\newtheorem{lemma}{Lemma}[section]

\theoremstyle{plain}
\newtheorem{proposition}{Proposition}[section]

\theoremstyle{plain}

\theoremstyle{plain}

\theoremstyle{plain}
\newtheorem{theorem}{Theorem}[section]

\theoremstyle{definition}

\theoremstyle{definition}
\newtheorem{assumption}{Assumption}[section]

\theoremstyle{remark}
\newtheorem{remark}{Remark}[section]

\theoremstyle{plain}
\newtheorem{corollary}{Corollary}[section]

\newcommand{\unt}{\underline{\theta}}
\newcommand{\ovt}{\overline{\theta}}

\title{Coexistence of Centralized and Decentralized Markets}

\author{Berk Idem\footnote{Penn State University. Email address: berkidem@psu.edu. I am grateful to my advisor Vijay Krishna and my committee members Kalyan Chatterjee, Henrique De Oliveira, Nima Haghpanah, and Ran Shorrer for their invaluable guidance. I would also like to thank Nageeb Ali, Nino Doghonadze, Miaomiao Dong, Marc Henry, Rohit Lamba, Ece Teoman, and Neil Wallace for many insightful conversations, and the participants of EC'21, Midwest Theory Conference '21, Pennsylvania Theory Conference '21, Stony Brook GT Conference '21 for their comments and questions.}\\Penn State University}

\begin{document}
	
	\maketitle

\begin{abstract}

In this paper, I introduce a profit-maximizing centralized marketplace into a decentralized market with search frictions. Agents choose between the centralized marketplace and the decentralized bilateral trade. I characterize the optimal marketplace in this market choice game using a mechanism design approach. In the unique equilibrium, the centralized marketplace and the decentralized trade coexist. The profit of the marketplace decreases as the search frictions in the decentralized market are reduced. However, it is always higher than the half of the profit when the frictions are prohibitively high for decentralized trade. I also show that the ratio of the reduction in the profit depends only on the degree of search frictions and not on the distribution of valuations. The thickness of the centralized marketplace does not depend on the search frictions. I derive conditions under which, this equilibrium results in higher welfare than either institution on its own.

\end{abstract}

\newpage
\section{Introduction}

Centralized marketplaces that bring buyers and sellers together have seen massive growth in the last decade. For instance, it has been estimated that Amazon generates half of all e-commerce sales in the US.\footnote{\cite{amazon}.} Another analysis estimates that usage of ride-share apps surpassed taxis in NYC as early as 2017.\footnote{\cite{uber}.} Airbnb and Vrbo's market shares in vacation rentals reached half of the market.\footnote{\cite{airbnb}.} All of these evidence suggest that the centralized marketplaces have significant positions in their respective markets.

These centralized marketplaces are successful because they reduce search and information frictions present in decentralized markets. In a decentralized market, an agent may not meet with a trading partner, and even when he meets with a partner, it may not be the right one. Centralized marketplaces reduce these frictions by attracting agents, collecting information from them, and making sure the realized matches are efficient enough to allow the marketplace to make some profit. In turn, the near certainty of trade on these platforms attracts many agents and gives the platforms significant market shares.

But the reduction in frictions is not without cost. Centralized marketplaces take commissions from the participants --buyers or sellers, or both. It is up to the individual traders to decide whether it is worthwhile to buy/sell in the centralized marketplace or to do so in a decentralized manner. The two institutions compete. And, the payoff from trading in the decentralized market is a function of the self-selected agents who also choose to trade there. Given the rising market shares of platforms like Amazon, Uber, etc., there are concerns about the possible monopolization of trade by such platforms \citep{khan}. Indeed, a congressional panel investigating competition in the digital markets has asserted that Amazon has monopoly power as an intermediary in the US e-commerce market \citep{amazon}.

In this paper I ask: Would a centralized marketplace monopolize all trade or is there room for some forms of decentralized trade to coexist? Moreover, if coexistence is possible, what are the welfare and profit consequences of multiple trading modes? How is the centralized marketplace affected by search frictions?

The literature on competing platforms sheds some light on the competition when there are multiple profit-maximizing platforms. Search theory provides an understanding of the decentralized markets on their own. However, we know less about the impact of a centralized marketplace that competes with numerous other trading venues. Recent contributions to matching theory \citep{ashlagiroth,ekmekciyenmez,rothshorrer} have focused on how to make a marketplace attractive to agents, when there are other options. They show that even a benevolent marketplace can have difficulty in recruiting agents. It is natural to expect it to be even harder when the marketplace is a profit maximizer. To answer the questions above, I study a centralized marketplace introduced into a decentralized market with some frictions.

To illustrate the point, consider the example of someone who wants to buy or sell a used car. She can check the prices offered by Carvana, a platform specialized in the used car market. If the price offered by Carvana is acceptable, she can simply take this deal. However, if she thinks she can get a better deal by searching privately via newspaper ads, she can choose to do so. Even though Carvana may have a large market share in this sector, there are still endless possibilities for trading privately. When Carvana chooses what price to offer for each car, ignoring these possibilities would harm its profit. Moreover, it is hard to guess the impact of Carvana's response to these other options.

Financial markets provide another example. Many assets can be traded at the stock exchanges as well as over-the-counter. In the stock exchanges, there is essentially no uncertainty; agents can buy or sell at the posted prices. However, in the over-the-counter markets, trade is not as transparent; dealers often do not post prices in a public manner. Instead, they provide quotes when someone is interested in trading with them. The agents who trade with these dealers only observe the prices offered by a limited number of dealers before they trade. This paper provides a framework to think about the problems faced by a stock exchange that competes with over-the-counter trades.

In this paper, I develop a model of a centralized marketplace that competes for agents who also have the option to trade in a decentralized manner. I consider a setup with a single, indivisible good where each of a continuum of agents can buy or sell one unit of the good. The endowments are common knowledge whereas the valuations are the private information of the agents. A designer chooses an individually rational, incentive compatible mechanism to maximize revenue of the marketplace -- say, the commissions charged for intermediation.

The decentralized market is modeled as in Diamond-Mortensen-Pissarides \citep{diamond,mortensen,pissarides}; agents are randomly matched among those who choose to participate in the decentralized trade and then they engage in Nash Bargaining in each realized match.\footnote{Later I establish robustness by showing that the results extend to trading with a double auction, instead of Nash Bargaining under some distributions.}

I consider a market choice game where (i) the marketplace designer announces a mechanism, (ii) agents choose whether to join the mechanism or to search for a trading partner, (iii) outcomes are realized in both markets. I first establish that in the unique equilibrium, centralized marketplace and decentralized trade coexist. Thus, when the agents can choose between these two modes of trade, it is never an equilibrium for all agents to join the same market.\footnote{In solving the model, I focus on equilibria with no profitable bilateral deviations, i.e., in equilibria where there is not a pair of agents who would rather deviate to the other market together and trade there. Without this strengthening of the equilibrium notion, it would also be an equilibrium for all agents to join the same market.} In the equilibrium, the agents with low and high values join the marketplace while the agents with intermediate values choose to search. To summarize, high surplus trades take place in the marketplace while low surplus trades happen privately.

One might expect that competition from the decentralized market will significantly decrease the profits of the marketplace. This is not the case. I demonstrate that the profit of the marketplace in the coexistence equilibrium is at least half of the profit that the marketplace would make if there had been no decentralized trade. Moreover, I show that the ratio of the reduction in profit of the marketplace as result of competition from the decentralized market is independent of the distribution of agents valuations. In fact, this ratio is \textit{only} a function of the search friction -- the probability of finding a trading partner -- in the decentralized market. A decrease in the search frictions in the decentralized market decreases the profits of the marketplace. However, even if these frictions were absent, the profit of the marketplace is half of its profit when it operates on its own.

The thickness of the centralized marketplace is independent of the search frictions. Even at the extremes, where the search frictions are absent or prohibitively high, the centralized marketplace targets and successfully attracts exactly the same agents to trade there. As the frictions decrease, the centralized marketplace has to ``sweeten the deal'' for traders to join there. Thus, with lower frictions in the decentralized market, the profit of the centralized marketplace from each trade is lower. It would be reasonable to expect the centralized marketplace to become more exclusive as a response. However, it is in fact optimal for the centralized marketplace to attract exactly the same agents.

Next, I provide two types of welfare comparison. First, I focus on the traders' welfare in a Pareto sense. I show that decreasing the frictions in the decentralized market increases the payoff of each trader, no matter which market they choose to trade in. Then, I compare the total welfare -measured as gains from trade- created in this equilibrium to the welfare from either modes of trade operating on their own. Coexistence always improves the welfare over the centralized marketplace alone. Furthermore, I provide conditions under which the coexistence generates higher total welfare than the search market alone. Essentially, the decentralized market extends the extensive margin of trade (more agents trades in the coexistence equilibrium than in the baseline, single marketplace) while the marketplace extends the intensive margin (some agents trade with a higher probability in the coexistence equilibrium than when the decentralized market operates alone). Thus, the combination of these leads to increased efficiency.

Finally, I discuss an alternative setup where there are multiple profit-maximizing marketplace designers competing with each other. First, I study the case where they are restricted to choose direct mechanisms. This environment is akin to Bertrand Competition and the insight from the Bertrand Equilibrium carries over: In equilibrium, the designers make zero profit and the outcome is equivalent to a Walrasian Equilibrium. Next, I allow the designers to choose more complex mechanisms and construct one particular equilibrium using strategies that include ``price-matching guarantees.'' In standard Bertrand Competition, price-matching guarantees are known to allow monopoly pricing to be a Nash equilibrium. Here, each marketplace posts prices that are equivalent to the baseline marketplace where there is no decentralized market and the agents uniformly randomize over the marketplaces. In this case, the marketplaces share the baseline profit. This discontinuity between zero profit with direct mechanism and the collective baseline profit (similar to a cartel's monopoly profit) reflects the inadequacy of direct mechanisms when there are multiple designers. If agents could also join the decentralized market, then the marketplaces would each post prices equivalent to the coexistence equilibrium mechanism of the main model with price-matching guarantees, and the agents again randomize.

After the literature review below, the rest of the paper is organized as follows. In Section 2, I begin by studying a baseline model in which all trade takes place through the centralized marketplace. In Section 3, I consider the full model with both forms of trade and study the equilibria of the market choice game. I first show that there is an equilibrium and that it is unique. I then study the profit and the welfare properties of the equilibrium. In Section 4, I consider competition among multiple marketplace designers.

\subsection{Literature Review}

There are many studies that consider the problem of incentivizing participation to centralized marketplaces. However, many of them do not have a decentralized trade option, as in the literature on competing platforms. In the papers where there is a decentralized market as well, the centralized marketplace is often a benevolent one. I review these literatures separately.

In the literature on competing platforms \citep{rochet, armstrong}, the questions mainly focus on the competition among platforms under numerous configurations of fee and price structures that could be employed by the platforms. The important distinction between this literature and my study is that in the context of competing platforms, the competition is between two profit-maximizing entities; they each react (or best respond) to the other's actions. However, here, the centralized marketplace is in competition with a fixed set of rules that cannot react to the marketplace's actions. This paper complements these studies by providing a insights into the nature of the competition in a decentralized market with a platform. More recently \cite{hartline} also bridges the gap between mechanism design and two-sided markets with a model where sellers can choose to join a platform that sets a menu of selling procedures or to develop their own selling venue.

Following the financial crisis of 2007-2008, a series of papers initiated by \cite{philipponskreta, tirole} focused on the ability of the public interventions to increase the efficiency of the investments in the financial markets. In these papers, the designers of the centralized markets are concerned with social welfare. Specifically, they focus on markets with adverse selection (and moral hazard in the case of \cite{tirole}) in terms of the quality of the investments. Without the public intervention, the level of investments is below the socially optimal level. To reduce the adverse selection and increase the level of investment, the government introduces a program by overpaying for some assets and removing the weakest assets from the market. \cite{fuchsskrzypacz} studies a similar problem but considers the effects of the dynamic nature of the markets. By contrast, I focus on the profit-maximization problem of a centralized marketplace. Moreover, the nature of the frictions in the decentralized markets are different from the ones considered in this literature.

In another strand, some papers in matching theory \citep{ashlagiroth,ekmekciyenmez,rothshorrer} study the problem faced by a benevolent marketplace designer when the agents can choose between multiple venues. In kidney-exchange and school-choice settings, they show that it might be infeasible or undesirable to make sure everyone joins the centralized market, even if it aims to maximize the social welfare. This paper provides a natural counterpart where the centralized market is only concerned about is own profit. 

\cite{vohra} study a model where agents are allowed to deviate from a market mechanism to trade among themselves according to any feasible trading protocol. Their main result states that almost every market mechanism is inherently unstable in the sense that there is always a positive measure of agents who would like to deviate from it. My findings provide a partial counterpart to their result: By restricting the possible deviations from the market mechanism, I am able to find a stable market structure where both the centralized marketplace and the outside trade are active.

Literature on the efficient dissolution of partnerships has important parallels with this study. Starting with \cite{gck} and with contributions by many others \citep{mylovanovtroger, kitt,fieselerkittsteinermoldovanu, loertscherwasser, figueroaskreta}, the setup in this literature includes a divisible good that is owned by many agents and a designer who wants to allocate the whole supply to one agent (or in some cases, to at least reduce the number of owners), thus dissolving the partnership. All of these papers have endogenous roles as buyers and sellers; each agent has some endowment which is less than the total endowment in the economy. Thus, like here, they obtain intermediate types who that are excluded from the trade and U-shaped utilities as a functions of agents types. Moreover, in extending the Nash bargaining, I used the efficient double auction which was introduced by \cite{gck} and was shown to have a unique equilibrium by \cite{kitt}.

\cite{bilateral} studied the problem of choosing a trade mechanism to maximize the total welfare in the economy. Their main result shows that it is generically impossible to have an efficient trade mechanism -that allocates the good always to the agent who values it the most- without outside resources to finance it. In our model without a decentralized market, unsurprisingly, the welfare achieved is even less than what \cite{bilateral} provides. However, introducing the option to search improves the efficiency of the market as a whole.

\cite{miao} also studies a similar environment with centralized and decentralized markets. When the search technology in the decentralized market is improved so that it can support the Walrasian Equilibrium, he shows that the centralized market serves a vanishingly small part of the population.

	\section{Single Market}
	
	In this section, I study the centralized marketplace in isolation and restrict all trade to the centralized marketplace; hence the agents either trade on this marketplace or do not trade at all.

	\subsection{Setup}
	
	I consider a market with a single, indivisible good. There is a continuum of agents on $[0,1]$. Each agent has $1$ unit of endowment of the good and has a demand for $2$ units of it. As the good is indivisible, each agent can sell 1 unit, buy 1 unit, or neither buy nor sell any. Each agent has some valuation $\theta\in[0,1]$ for a unit of the good. The valuations are drawn from some continuous distribution $F$ with support $[0,1]$, and they are agents' private information. A mechanism designer knows the distribution of valuations, $F$, and wants to design a Bayesian Individually Rational mechanism to maximize its profit.

	By Revelation Principle, I focus on direct, Bayesian Incentive Compatible, Bayesian Individually Rational mechanisms. Since there is a continuum of agents, there is no aggregate uncertainty. Thus, it is without loss to focus on direct, Ex-Post Incentive Compatible, Dominant-Strategy Individually Rational mechanisms. Moreover, as agents are symmetric other than their valuations, I restrict attention to anonymous mechanisms; that is, the designer does not condition the mechanism on agents' 'names'.
	
	The designer will choose a mechanism described by the quantities and the transfers as functions of agents' reported valuations. Specifically, the allocations and transfers are given by functions $q:[0,1] \rightarrow \mathbb{R}$ and $t:[0,1] \rightarrow \mathbb{R}$, respectively. Thus, an agent who reports his valuation is $\theta$ will get $q(\theta)$ units of the good and will pay $t(\theta)$. Note that both the allocation and the transfer can be either positive or negative, depending on whether the agent is buying or selling. Hence, the net utility of the agent with the valuation $\theta$ from this mechanism is $$ u(\theta) = \theta \min\{ 1, q(\theta) \}  - t(\theta). $$
	
	As agents have demands for two units, having more than 2 units of the good is the same as having 2 unit for the agent. Therefore, the utility from the traded quantity is capped at 1 unit.

	\subsection{Analysis}

	The profit of the marketplace is the expected net payments. Thus, the designer seeks to maximize total payments, given incentive compatibility, individual rationality, and feasibility constraints.

	\begin{equation*}
		\begin{array}{lllll}
			
			\displaystyle\max\limits_{(q, t)} & \mathbb{E}_{\theta}\left[ t(\theta) \right]  \\

			\text{s. t. }\\
			
			\text{ (IC) } &\theta  \min\{ 1, q(\theta)\}  - t(\theta)  &\geq \theta \min\{ 1, q(\theta', \theta) \}   - t(\theta')\\
			
			\text{ (IR) }& \theta  \min\{ 1, q(\theta)  \}  - t(\theta) &\geq 0\\ 
			
			\text{ (Individual Feasibility) }& q(\theta) &\geq -1\\

			\text{ (Aggregate Feasability) }& \mathbb{E}_{\theta}\left[ q(\theta) \right] &\leq 0\\
			
		\end{array}
	\end{equation*}

	In an Online Appendix, I develop a series of lemmata to simplify this problem.\footnote{Online Appendix is available at \href{https://berkidem.com/documents/coexistence_online_appendix.pdf}{this link}.} They allow the problem to be restated in terms of the virtual values and virtual costs, defined as follows.
	
	\begin{definition}
		An agent with reported valuation $\theta$ has virtual value, $\mathcal{V}(\theta)$, and virtual cost, $\mathcal{C}(\theta)$, given by:
		
		\begin{equation*}
			\mathcal{V}(\theta) = \theta  - \dfrac{(1-F(\theta))}{f(\theta)} \text{ and } \mathcal{C}(\theta) = \theta  + \dfrac{(F(\theta)}{f(\theta)}.
		\end{equation*}
	\end{definition}

	Virtual value and virtual cost can be thought of as the marginal revenue and marginal cost. When an agent is a seller, that is, when an agent has a negative allocation, $q(\theta)<0$, his deduction from the profit of the marketplace is the virtual cost. Similarly, when an agent is a buyer, $q(\theta)>0$, his contribution to the profit is the virtual value. Then, the problem can be restated in these terms as follows:

	\begin{equation*}
		\begin{array}{llllll}
			
			\displaystyle\max_{q(\cdot)}  &\mathbb{E}    \left[  q(\theta) \left(    \mathbbm{1} \{ q(\theta)<0 \}\mathcal{C}(\theta)  + \mathbbm{1} \{  q(\theta)>0  \}\mathcal{V}(\theta) \right) \right] \\

			\text{s. t. }&\\
			
			& q(\theta) \text{ is increasing}\\
			
			&  q(\theta) \geq -1\\
			
			&\mathbb{E}_{\theta}\left[ q(\theta) \right] = 0
			
		\end{array}
	\end{equation*}

	\begin{definition}\label{defn:regular}
		The distribution of agents' valuations, $F$ is \textbf{regular} if both $\mathcal{V}$ and $\mathcal{C}$ are increasing.
	\end{definition}
	
	The regularity condition guarantees that the marketplace has a decreasing marginal revenue from having additional buyers and increasing marginal cost from having additional sellers. With this definition, we are ready to state the main result of this section, which characterizes the baseline optimal marketplace.

	\begin{theorem}\label{thm:monagora}
		Suppose the distribution $F$ is regular. Then, the optimal mechanism has the allocation rule 
		
		\begin{align*}
			q(\theta)=\begin{cases}
				-1 &\text{ if } \theta\leq \unt\\
				0 &\text{ if } \unt < \theta < \ovt\\
				1 &\text{ if } \theta\geq \ovt
			\end{cases}
		\end{align*}
		
		and the transfer rule
		
		\begin{align*}
			t(\theta)=\begin{cases}
				-\unt &\text{ if } \theta\leq \unt\\
				0 &\text{ if } \unt < \theta < \ovt\\
				\ovt &\text{ if } \theta\geq \ovt
			\end{cases}
		\end{align*}
		
		where $\unt$ and $\ovt$ satisfies $\mathcal{C}(\unt) = \mathcal{V}(\ovt)$ and solves the problem
		
		\begin{equation*}
			\begin{array}{llllll}
				
				\displaystyle\max_{\unt,\ovt}  &\left[ -\unt F(\unt)  + \ovt (1-F(\ovt )) \right] \\

				\text{s. t. }&\\
				
				&F(\unt) = 1- F(\ovt)\\
				
				& 0\leq \unt\leq \ovt \leq 1.
				
			\end{array}
		\end{equation*}
		
	\end{theorem}
	
	The proof can be found in Appendix \ref{proof:monagora}.
	
	Notice that each agent below $\unt$ sells 1 unit and gets paid $\unt$, and each agent above $\ovt$ buys 1 unit and pays $\ovt$. If the designer posts $\unt$ as the price for selling and $\ovt$ as the price for buying, and let agents choose what to do, the allocation above represents exactly what the agents would do. Thus, the designer can implement this mechanism in a very straightforward way by posting bid-ask prices.

	\subsection{The Simple Economics of Optimal Marketplaces}
	
	Here, I provide an analysis of the optimal profit-maximizing marketplaces in the same spirit as \cite{bulow}. They have shown that the optimal auction design problem can be understood as a monopoly pricing problem. In my setting, we will see that the optimal marketplace design problem can be understood as solving two simultaneous monopoly and monopsony problems.
	
	The objective function of the designer is the expected payments of the agents. In the restatement of the problem, we have seen that an equivalent way of thinking about the expected payments is the expected difference between the virtual values and costs, or expected \textit{virtual surplus}. Thus, I will now use virtual values and costs to find the optimal level of trade graphically.
	
	For each quantity level, $q$, if the marketplace wants measure $q$ of sellers, the price for selling should be $F^{-1}(q)$. Similarly, to have measure $q$ of buyers, the price for buying should be $F^{-1}(1-q)$. Then, the inverse supply and demand in the economy are given by $\mathcal{S}=F^{-1}(q)$ and $\mathcal{D}=F^{-1}(1-q)$
	
	Also, for each quantity level, $q$, let $\mathcal{MC}(q)= \mathcal{C}(F^{-1}(q))$ and $\mathcal{MR}(q)= \mathcal{V}(F^{-1}(1-F(\theta)))$. Here, the virtual cost is the marginal cost of the marketplace so it represents the `effective supply' in the marketplace. Similarly, the virtual value gives us the marginal revenue curve so it represents the `effective demand' for the marketplace. Since the virtual values and costs are the marginal revenue and cost curves for the profit-maximization problem, instead of agents' willingness to pays (or willingness to get paid), I use their virtual values and costs to find the optimal quantity of trade. This is very similar to using the marginal revenue instead of the demand in the monopoly problem.	If the designer knew a buyer's valuation, she would charge the buyer exactly his valuation. However, since the designer is uninformed about the valuations, she has to charge a buyer something less than his valuation. Similar reasoning holds for the sellers as well. So, each agent who trades, gets some information rent for his private information. This can be seen in the Figure \ref{monagoraprofit}.

	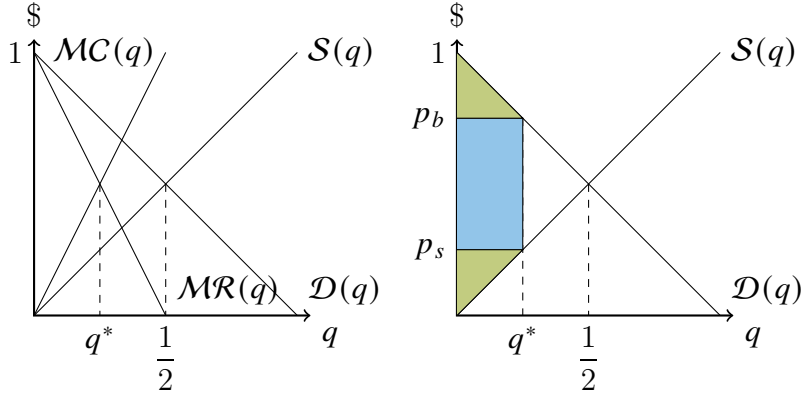
\begin{figure}
		\centering
		\begin{tikzpicture}[scale=0.35]
			
			\draw[thick,<->] (0,10.5) node[above]{$\$$}--(0,0)--(10.5,0) node[below right]{$q$};
			
			\draw(0,0)--(10,10) node[right]{$\mathcal{S}(q)$};

			\draw(0,10)--(10,0) node[above right]{$\mathcal{D}(q)$};
			
			\draw[dashed](5,5)--(5,0) node [below] {$\dfrac{1}{2}$};

			\node[] at (0,10) [left] {$1$};
			
			\draw(0,10)--(5,0) node[above right]{$\mathcal{MR}(q)$};
			
			\draw(0,0)--(5,10) node[left]{$\mathcal{MC}(q)$};
			
			\draw[dashed](2.5,5)--(2.5,0) node [below] {$q^*$};

		\end{tikzpicture}
		\begin{tikzpicture}[scale=0.35]
			
			\draw[thick,<->] (0,10.5) node[above]{$\$$}--(0,0)--(10.5,0) node[below right]{$q$};
			
			\draw(0,0)--(10,10) node[right]{$\mathcal{S}(q)$};

			\draw(0,10)--(10,0) node[above right]{$\mathcal{D}(q)$};
			
			\draw[dashed](5,5)--(5,0) node [below] {$\dfrac{1}{2}$};

			\node[] at (0,10) [left] {$1$};

			\draw[dashed](2.5,7.5)--(2.5,0) node [below] {$q^*$};
			
			\draw[dashed](2.5,7.5)--(0,7.5) node [left] {$p_b$};
			
			\draw[dashed](2.5,2.5)--(0,2.5) node [left] {$p_s$};
			
			\draw[fill=solarizedBlue!50] (0,2.5) -- (2.5,2.5) -- (2.5,7.5) --(0,7.5);	
			
			\draw[fill=solarizedGreen!50] (0,10) -- (2.5,7.5) --(0,7.5);
			
			\draw[fill=solarizedGreen!50] (0,0) -- (2.5,2.5) --(0,2.5);
			
		\end{tikzpicture}
		\caption{\footnotesize \textcolor{solarizedBlue}{\textbf{Profit of the Marketplace}}, \textcolor{solarizedGreen}{\textbf{Information Rents}}.}
		\label{monagoraprofit}
	\end{figure}
	
	Figure \ref{monagoraprofit} shows $\mathcal{D}$, $\mathcal{S}$, $\mathcal{MR}$, and $\mathcal{MC}$ together.\footnote{The figure is drawn for the uniform distribution over $[0,1]$ for simplicity.} As the marketplace is maximizing the profit, the optimal level of trade is given by $q^*$ such that $\mathcal{MR}$ and $\mathcal{MC}$ are equal to each other. Moreover, the area of the rectangle given by $q^*$ and the difference $p_b-p_s$ is equal to the profit of the marketplace. The triangles above $p_b$ and below $p_s$ are the buyers' and sellers' information rents. Finally, the triangle between $p_b$ and $p_s$ to the right of the profit is the deadweight loss created by (i) incomplete information and (ii) profit-maximization.

	\section{Decentralized Market}\label{secdec}

	Now, I introduce a decentralized market and give the agents a choice between joining the mechanism or the decentralized market. The decentralized market is modeled as a search market with frictions as in the Diamond-Mortensen-Pissarides Model \citep{diamond,mortensen,pissarides}.
	
	Here is how the decentralized market works. Agents who join there are randomly matched to each other. This process is governed by a matching function, adapted from the search theory. Once matched, each pair of agents engage in Nash Bargaining; that is, the agents observe each other's valuations and then split the surplus equally. 
	
	Timeline is as follows. A mechanism designer announces a mechanism through which agents can trade the good and invite some of the agents to join it. Upon observing the mechanism, an agent joins either the mechanism or the search market, or does not join either market.
	
	Agents who join the mechanism get what the announced mechanism promises. The mechanism can depend on the set of agents who join it. The simplest such mechanism would be one that operates if agents who join it are exactly those who were invited and shuts down otherwise. Results I obtain here would be valid for an equilibrium under this extreme mechanism. However, in that case it would also be an equilibrium for all agents to join the decentralized market. With proper handling of the off-path payoffs, I obtain an obedience principle: In the unique equilibrium, exactly the agents whom the designer wants to serve join the marketplace. I expand on this in the Section \ref{obedience}.
	
	In this market choice game, I look at the subgame perfect equilibrium with no bilateral deviations, meaning no pair of agents would prefer to change markets together.\footnote{As there is a continuum of agents, measure zero deviations do not affect any payoff. Thus, I only prevent the deviations that have an impact, i.e., those with a positive measure of deviators. However, as I will show while proving Proposition \ref{prop:unique}, whenever there is one pair who can profitable deviate, there is always a positive measure of profitable deviations.} When the decentralized market is not empty, this strengthening does not have a bite in that direction, since deviating agents still go through the random search process and it is a measure 0 event for them to meet with their co-deviator. As before, I focus on deterministic mechanisms. Moreover, for technical reasons, I assume that the designer invites a closed set of agents to join the marketplace.

	I first study these markets under the following assumption that restricts the set of mechanisms the designer can choose. Then, I consider the general case without this assumption. There, I show that this assumption is without loss. Starting the analysis under the assumption makes the environment more accessible; which is why I present the results in this order.

	\begin{assumption}\label{ass:simple}
		Suppose a convex set of types, $(\unt,\ovt)$ joins the search market in the equilibrium.
	\end{assumption}
	
	So, the designer chooses $\unt$ and $\ovt $ optimally, anticipating the agents' best responses to the announced mechanism. I call the equilibria that satisfy this assumption \textit{simple equilibria} and the mechanisms that induce these equilibria \textit{simple mechanisms}.
	
	First, I show that the simple equilibria are equivalent to posting bid-ask prices, i.e. prices for buying and selling. This structure makes sure that the mechanism has the agents with high virtual values and low virtual costs. Then, in Theorem \ref{unrestrictedthm}, I prove that this assumption is without loss of profit: This is the structure we would have in the equilibrium without assuming it.

	Although most of the paper focuses on a decentralized market with Nash bargaining, I later extend the analysis to a double auction. In this double auction, each agent makes a bid and the agent with the higher bid buys the other agent's endowment and pays the mid-point of their bids. I show that under some restrictions, with uniform distribution, all results obtained for the Nash Bargaining also applies to the environment with the double auction.

	\subsection{Matching Functions}
	
	In this section, I am going to introduce the matching technology in the search market. As in the search theory literature, a matching function determines the efficiency of the search process. Before defining this formally, I introduce some notation.
	
	Suppose $\Theta^d$ is the set of agents who join the search market in a strategy profile and $\mu$ is the measure with respect to the distribution $F$. Then, the measure of meetings in the search market will be given by a matching function $M(\mu(\Theta^d))$ as a function of the measure of agents in the search market, $\mu(\Theta^d)$.
	
	In search theory, matching functions are commonly assumed to have constant returns to scale (CRS). This means that doubling the size of the market also doubles the number of meetings. Since I focus on a market where every agent has the same endowment and demand, this is a one-sided market. In this setup, CRS matching functions are simply linear in the size of the market: $M(\mu(\Theta^d)=m \times \mu(\Theta^d)$ where $m\in [0,0.5]$ is the efficiency parameter of the matchings. Then, probability that an agent finds a match in the search market, $p$ is equal to $2m$ since the total measure of meetings is $M(\mu(\Theta^d)$, each agent is equally likely to be in any meeting, and there will be two agents in each meeting.
	
	Notice that the probability of a match, $p$ is independent of the set of agents who join the search market as well as the measure of the set. Since there is a one-to-one relationship between $p$ and $m$, from now, an agents probability of finding a match is simply denoted by $p$. For the same reason, in this setup, $p$ itself can be thought of as the primitive of the search market and the efficiency parameter of the matching process.

	\subsection{Payoffs from Search}\label{payoffsearch}
	
	With the matching process reduced to the probability of a match, $p$, we can think about agents' expected payoffs from search. Of course, this will depend on the set of agents who join the decentralized market, $\Theta^d$, since your payoff changes depending on who you meet.
	
	When an agent joins the decentralized market, he gets a match with probability $p$. With some probability, the match has a lower valuation. In this case, the agent buys the good and pays the average of their valuations, since they engage in Nash Bargaining. On the other hand, with some probability, the match will have a higher valuation. Then, the agents sells his endowment and gets paid the average of their valuations.
	
	Under the Assumption \ref{ass:simple}, we can learn more about the payoffs from search. Suppose the agents who join the decentralized market is $(\unt,\ovt)$. Then, when an agent with valuation $\theta\leq \unt$ considers deviating to the decentralized market, he knows that he will be a seller in all matches; all agents he can meet there have higher valuations. Conversely, an agent with valuation $\theta\geq \ovt$ knows that he would be a buyer in the decentralized market, since everyone there has lower valuations.
	
	In the Appendix \ref{app:payoffs}, I go into further details about the expected payoffs from the decentralized market. In Figure \ref{outsideutilities}, I illustrate the shape of the expected utilities from the decentralized market, when the valuations are uniformly distributed over $[0,1]$ and agents in $(0.1,0.9)$ join the decentralized market.

	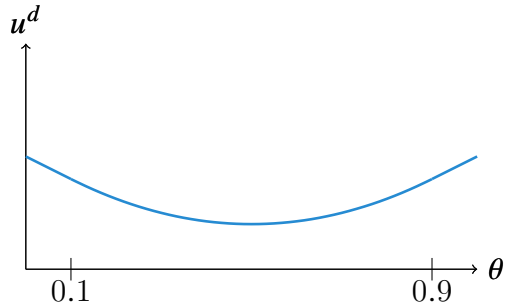
\begin{figure}[h]
		\centering 
			\begin{tikzpicture} [scale=6]
				\draw[->] (0, 0) -- (1, 0) node[right] {$\theta$};
				\draw[->] (0, 0) -- (0, 0.5) node[above] {$u^d$};

				\draw[->] (0, 0) -- (1, 0) node[right] {$\theta$};
				\draw[->] (0, 0) -- (0, 0.5) node[above] {$u^d$};
				
				\draw[domain=0:0.1, smooth, variable=\x, solarizedBlue, line width=1pt] plot ({\x}, {(1-2*\x)/4});
				\draw[domain=0.1:0.9, smooth, variable=\x, solarizedBlue, line width=1pt] plot ({\x}, {(\x^2-\x+0.41)/1.6});
				\draw[domain=0.9:1, smooth, variable=\x, solarizedBlue, line width=1pt] plot ({\x}, {(\x+\x-1)/4});

				\draw[] (0.1,0.25mm) -- (0.1,-0.25mm) node[left]{}; 
				\draw[] (0.9,0.25mm) -- (0.9,-0.25mm) node[left]{};	
				\node[] at (0.1,0) [below] {$0.1$};		
				\node[] at (0.9,0) [below] {$0.9$};
				
			\end{tikzpicture}
		\caption{Agents' utilities from their outside options as a function of their valuations, given $(0.1,0.9)$ join the decentralized market.}
		\label{outsideutilities}
	\end{figure}
	
\subsection{An Obedience Principle}\label{obedience}

	In this section, I show that the marketplace designer can implement the segmentation she prefers.

	Suppose the designer announces a direct mechanism and invites a set of agents $[0,\unt]\cap [\ovt,1]$. She commits to operating the announced mechanism as long as (i) the reported types in the mechanism is a subset of the invited agents, (ii) the reported types in the mechanism is union of two intervals $[0,a]\cap [b,1]$ where $a<\unt$ and $\ovt<b$ and finally (iii) $F(a)=1-F(b)$. Otherwise, the designer shuts down the marketplace. The following proposition shows that with this announcement, the designer can induce any segmentation of this form.
	
	\begin{proposition}\label{prop:obed}
		In a simple equilibrium, set of agents who join the marketplace are exactly those who were invited.
	\end{proposition}
	
	The proposition is proved in Appendix \ref{proof:obed}.
	
	The intuition is as follows: By committing to shut down the marketplace if there are unwanted type reports and designing an incentive compatible mechanism, the marketplace makes sure that there would not be any unwanted agents there. In other words, the unwanted agents are offered utilities that are less than their expected payoff from the decentralized market. Since the marketplace is incentive compatible, joining there and misreporting cannot be better than joining the decentralized market by transitivity.
	
	The more interesting part of the proof deals with showing that a strict subset of wanted agents joining the marketplace cannot be an equilibrium either. This is accomplished by showing that if there were such an equilibrium, then the realized cutoff types would have to be indifferent between two markets due to profit-maximization. The designer originally posts a mechanism that makes the cutoff types of the wanted agents indifferent between these two markets. Then, it is not possible that the cutoff types of this strict subset are also indifferent.

	\subsection{Designing the Mechanism}
	
	Under Assumption 1, the agents who join the mechanism will have valuations in $[0,\unt]$ and $[\ovt,1]$. We can write their utilities as follows, using payoff equivalence from the \href{https://berkidem.com/documents/coexistence_online_appendix.pdf}{Online Appendix}.
	
	\begin{equation*}
		u^m(\theta) =\begin{cases}
			u^m(\ovt ) +\displaystyle \int\limits_{\ovt }^{\theta}q(x)dx & \text{ if } \theta\geq \ovt\\
			u^m(0) + \displaystyle \int\limits_{0}^{\theta}q(x)dx & \text{ if } \theta\leq \unt.\\
		\end{cases}
	\end{equation*}

	In this environment, there will be different binding individual rationality constraints for agents below $\unt$ and for agents above $\ovt$. Thus, writing the utilities in this form is more convenient. Similarly, we can write the transfers:
	
	\begin{equation*}	
		t(\theta) =\begin{cases}
			\theta q(\theta) -u^m(\ovt) - \displaystyle \int\limits_{\ovt}^{\theta}q(x)dx & \text{ if } \theta\geq \ovt,\\
			\theta q(\theta) -u^m(0) - \displaystyle \int\limits_{0}^{\theta}q(x)dx. & \text{ if } \theta\leq \unt.\\
		\end{cases}
	\end{equation*}

	Now, we can study the profit from the optimal allocation given the cutoffs. The step-by-step derivation can be followed in the Appendix \ref{symprofit} but here is the end-result:
	
	\begin{align*}\hspace*{-0.5in}
		\Pi_{\unt,\ovt} & = \mathbb{P}[\theta\in[0,\unt]] \mathbb{E}[t(\theta)|\theta\in[0,\unt]] + \mathbb{P}[\theta\in[\ovt,1]] \mathbb{E}[t(\theta)|\theta\in[\ovt,1]] \\	
		&= \underbrace{ \vphantom{\int\limits_{\ovt }^1}  -F(\underline{\theta})u^m(\underline{\theta}) -(1-F(\overline{\theta}))u^m(\overline{\theta})}_{\text{Compensations for joining the mechanism}} + \underbrace{\int\limits_0^{\unt} \left[   \left( x + \dfrac{F(x)}{f(x)} \right) q(x) \right]f(x)dx + \int\limits_{\ovt }^1 \left[  \left(x- \frac{1-F(x)}{f(x)} \right) q(x) \right]f(x)dx.}_{\text{Profit if there were no search market}}
	\end{align*}
	
	Here, the agents will get some compensations for joining the centralized marketplace instead of the decentralized market. This is because the agents are giving up the opportunity to trade in the decentralized market when they join the centralized marketplace. Thus, the individual rationality constraints are endogenously determined by the equilibrium segmentation. The centralized marketplace has to pay the agents this lost `opportunity cost', on top of the standard information rents created by the informational asymmetry. The exact form of the compensation will become clear below when we consider the individual rationality constraints.

	I defined the virtual cost, $\mathcal{C}$ and the virtual value $\mathcal{V}$ as:
	
	\[ \mathcal{C}(x)= x + \dfrac{F(x)}{f(x)} \text{ and } \mathcal{V}(x)=x- \frac{1-F(x)}{f(x)}. \]
	
	I continue to assume both of them are increasing and call such distributions regular. 
	
	\paragraph{The Constraints}

	The support of the distribution of valuations is $[0,1]$. Thus, the first restriction is $0\leq \unt\leq \ovt \leq 1$.
	
	Second, we have a market-clearing constraint in the form of
	
	\[  \int_{0}^{\unt} q(x)f(x)dx +  \int_{\ovt }^1 q(x) f(x)dx\leq 0. \]
	
	This simply says that the designer cannot sell more than she buys.
	
	Third, since each agent has one unit of endowment, they cannot sell more than that. Thus, we need $-1 \leq q(\theta)$. In fact, since the good is indivisible, each agent must have $q(\theta)\in \{-1,0,1\}$ as their allocation.
	
	Fourth, we know that for incentive compatibility in the mechanism, we need the allocation to be increasing. The rest of the requirements of the incentive compatibility are already embedded in the transfers. So, as long as the allocation is increasing, the mechanism will be incentive compatible.
	
	Finally, we need to consider the implications of the individual rationality.
	
	\paragraph{Individual Rationality}
	
	Notice that in any profit maximizing mechanism, individual rationality (IR) constraint for at least one type of agent in each segment who joins the mechanism ($[0,\unt]$ and $[\ovt,1]$) must bind. If not, then uniformly increasing the payment of all agents in the particular segment without a binding IR until there is a binding constraint increases the profit.
	
	Moreover, the binding IR constraints in the optimal mechanism must be $\unt$ and $\ovt$. To see this, notice that under the Assumption 1, we want to construct an equilibrium such that agents in $[0,\unt]$ and $[\ovt,1]$ join the mechanism while the rest of the agents join the search market. In the equilibrium, each agent will choose the market that offer him a higher utility. Then, in the equilibrium, agents with valuations in $[0,\unt]$ and $[\ovt,1]$ must have a higher utility in the centralized marketplace and agents with valuations in $(\unt,\ovt)$ must have a higher expected utility in the decentralized market. Since the utilities from both markets are continuous in valuations, this means the utilities from the decentralized market, $u^d$, and the utilities from the centralized marketplaces, $u^m$, must cross each other at $\unt$ and $\ovt$. Thus, agents with valuations $\unt$ and $\ovt$ must be indifferent between the markets and these are the binding individual rationality constraints: 
	
	\begin{equation*}
	u^d(\unt)=u^m(\unt)= u^m(0) +\int_{0}^{\unt} q(x)dx \text{ and } u^m(\ovt)=u^d(\ovt).
	\end{equation*}
	
	This allows us the make the following observation, proved in the Appendix \ref{app:lemma:ir}.

	\begin{lemma}\label{simplealloc}
		In a simple equilibrium with cutoffs $\unt$, $\ovt$, allocations must be such that:
		
		\begin{align*}
			q(\theta)=\begin{cases}
				-1 \text{ if } \theta\leq \unt,\\
				1 \text{ if } \theta\geq \ovt.\\
			\end{cases}
		\end{align*}
	\end{lemma}

	This is a very useful observation. On intervals where the allocation is constant, transfer is constant as well. Thus, with this lemma, we know that the marketplace will be equivalent to posting some bid-ask prices. From here, we also know that the agents with valuations below $\unt$ can only be sellers in any market and agents with valuations above $\ovt$ can only be buyers in any market.

\subsection{Existence of Simple Equilibrium}

Here, I show how the designer can find the optimal pair of $\unt$ and $\ovt$. The Main Result of this section shows that under the Assumptions\ref{ass:simple}, a simple equilibrium exists. In this equilibrium, the marketplace designer attracts the agents with very low and very high valuations. The rest of the agents are left to search. Moreover, cutoffs are such that the marginal cost of the highest seller in the marketplace is equal to the marginal revenue of the lowest buyer in the marketplace. Finally, the measures of buyers and sellers in the marketplace are equal to each other.

	\begin{theorem}\label{2cutofftheorem}
		Suppose $F$ is regular. Then, there exists $\unt,\ovt$ such that for any $p\in [0,1]$, in the simple equilibrium,
		\begin{itemize}
			\item agents in $[0,\unt]$ and $[\ovt ,1]$ join the mechanism,
			\item agents in $(\unt,\ovt)$ join the search market,
			\item $\mathcal{C}(\unt) = \mathcal{V}(\ovt)$ and $F(\unt)=1-F(\ovt)$.\vspace{0.2in}
			
		\end{itemize}
		
	\end{theorem}
	The proof is in the Appendix \ref{app:simpleexistence}.

	To illustrate the utilities agents are offered in the mechanism and expect from the search market, suppose everyone gets a meeting in the decentralized market ($p=1$) and agents are uniformly distributed over the unit interval. The optimal cutoffs for the  are $\unt=\frac{1}{4}$ and $\ovt=\frac{3}{4}$. Figure \ref{simp_utilities} shows the utilities agents can expect from either market \textit{in the equilibrium}.

	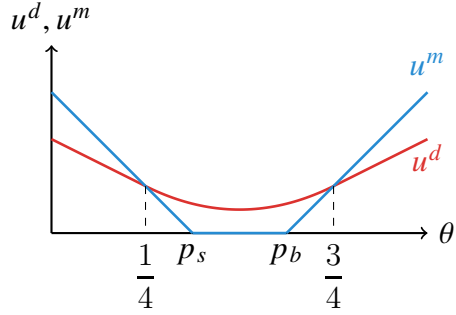
\begin{figure}[h!]
		\centering
		\begin{tikzpicture} [baseline=0, scale=5]
			\draw[thick, ->] (0, 0) -- (1, 0) node[right] {$\theta$};
			\draw[thick, ->] (0, 0) -- (0, 0.5) node[above] {$u^d, u^m$};		
			
			\draw[domain=0:0.25, smooth, variable=\x, solarizedRed, line width=1pt] plot ({\x}, {(1-2 *\x)/4});
			\draw[domain=0.25:0.75, smooth, variable=\x, solarizedRed, line width=1pt] plot ({\x}, {(\x^2-\x + 5/16)});
			\draw[domain=0.75:1, smooth, variable=\x, solarizedRed, line width=1pt] plot ({\x}, {(2*\x-1)/4});
			
			\node[solarizedRed] at (1, 0.2) [] {$u^d$};

			\draw[domain=0:0.25, smooth, variable=\x, solarizedBlue, line width=1pt] plot ({\x}, {0.375-\x});
			\draw[domain=0.25:0.375, smooth, variable=\x, solarizedBlue, line width=1pt] plot ({\x}, {0.375-\x});
			\draw[domain=0.375:0.625, smooth, variable=\x, solarizedBlue, line width=1pt] plot ({\x}, {0});
			\draw[domain=0.625:0.75, smooth, variable=\x, solarizedBlue, line width=1pt] plot ({\x}, {\x-0.625});
			\draw[domain=0.75:1, smooth, variable=\x, solarizedBlue, line width=1pt] plot ({\x}, {\x-0.625});			
			
			\node[solarizedBlue] at (1, 0.45) [] {$u^m$};						
			
			\draw[dashed] (0.25, 0.125) -- (0.25, 0);
			\draw[dashed] (0.75, 0.125) -- (0.75, 0);

			\node[] at (0.25,0) [below] {$\dfrac{1}{4}$};		
			\node[] at (0.75,0) [below] {$\dfrac{3}{4}$};			
			
			\node[] at (0.375,0) [below] {$p_s$};		
			\node[] at (0.625,0) [below] {$p_b$};
		
		\end{tikzpicture}
		\caption{The utilities from the search market and the optimal mechanism under the simple equilibrium.}
		\label{simp_utilities}
	\end{figure}

As it can be observed from the Figure \ref{simp_utilities}, agents with intermediate types receive a lower utility in this equilibrium in the centralized marketplace. So they are happy to join the decentralized market. Moreover, the agents with low and high values are offered higher utilities in the centralized marketplace than they expect from the decentralized market. Thus, no agent has any unilateral profitable deviation.

The preceding theorem describes the structure of the simple equilibrium. However, it does not completely describe the mechanism that induces this equilibrium. Lemma \ref{simplealloc} pins down the allocations and transfers for the agents who join the mechanism. What remains to be determined is what should be offered to agents in $[\unt,\ovt]$.

In fact there are many mechanisms that would induce the same equilibrium that only differ in the off-path payoffs. Proposition \ref{simplemech} in the Appendix \ref{simplemechproof} describes one such mechanism. As described before, it is equivalent to offering some bid-ask prices, i.e., prices for buying and selling. Essentially, I compute the transfers for the agents who join the mechanism and then extends the allocation and the transfer rules to the rest of the agents in a way that the allocations are increasing and we have $u^m(\theta)<u^d(\theta)$ for each $\theta\in(\unt,\ovt)$. The prices for buying and selling, $p_b$ and $p_s$, implied by this proposition are the lowest and the highest valuations such that $u^m$ is equal to 0, as can be seen in the Figure \ref{simp_utilities}.

\subsection{Profit of the Marketplace}
	
Next, I establish the relationship between the profit in this equilibrium and the profit of the marketplace when there is no decentralized market. As the decentralized trade makes it costly to recruit agents to the centralized marketplace, it is not hard to guess that the decentralized trade decreases the profit of the centralized marketplace. Interestingly, the decrease in the profit only depends on the efficiency of the matching process in the decentralized market; it is independent of the distribution of the valuations.

Let $\Pi^M$ denote the profit of the marketplace when there is no decentralized market.

\begin{theorem}\label{thm:profit}
	Suppose $F$ is regular. Then, the profit of the marketplace in coexistence is equal to $\left( 1-\dfrac{p}{2} \right) \Pi^M$.
	
	Moreover, the aggregate compensations agents receive in coexistence is equal to $\dfrac{p}{2} \Pi^M$.

\end{theorem}

It is surprising that the ratio of the profits is completely independent of the distribution of valuations. The ratio only depends on the probability that an agents finds a match in the search market, $p$. The intuition for this result is as follows. When there is no decentralized trade, the marketplace keeps all the surplus that remains after `paying' the information rents. When agents can join the decentralized market, they find a match with probability $p$, and they get half of the surplus created by their meetings. I show that the expected surplus of a buyer and a seller is equal to the profit of the marketplace from each sale. Thus, in total, marketplace pays $\frac{p}{2} \Pi^M$ as compensation and the rest is kept as the profit of the marketplace.\footnote{The result directly follows from the derivation of the coexistence profit in Appendix \ref{genprofit}.}

Since the ratio $1-\frac{p}{2}$ is decreasing in $p$, the profit is also decreasing in $p$: As the matching process becomes more efficient, the decentralized market becomes more attractive. Then, the opportunity cost of each agent increases, meaning the compensations increase. Thus, the profit of the marketplace decreases.

Even if every agent finds a match in the decentralized market, meaning $p=1$, the profit of the marketplace is half of the profit when it operates on its own. This provides a distribution-free lower bound for the profit as a function of $\Pi^M$. Here is why the marketplace can still make positive profit when $p=1$. Due to the random matching process, even when each agent in the decentralized market gets a match with certainty, the matches may have small gains from trade. An agent might meet someone whose valuation is very close to her own. However, the marketplace solves this problem and makes profit by creating efficient matches. At the other extreme, $p=0$, when there is no match in the decentralized market, agents cannot trade bilaterally. Then, the marketplace makes the profit $\Pi^M$. Thus, the baseline with the marketplace alone is obtained as a special case.

\begin{remark}
	The independence of compensation:profit ratio holds for off-equilibrium as well: Given any (potentially sub-optimal) cutoffs $\unt$ and $\ovt$ for segmentation, the total compensations is equal to $\frac{p}{2}$ times the virtual surplus generated by this segmentation.
\end{remark}

\subsection{Thickness of the Marketplace}

In this section, I compare the agents who trade on the centralized marketplace in the coexistence to those who get to trade in the marketplace when it operates on its own.

In the previous section, we have seen that the decrease in the profit of the marketplace is given by a ratio which is independent of the distribution. This means that the objective function of the marketplace, with or without the decentralized trade, is the same, up to a multiplication with a constant. Then, the solution -in terms of the allocation- is the same. The centralized marketplace wants to serve exactly the same types for each value of $p$ including $p=0$. Of course, for each value of $p$, the transfers are adjusted so that the cutoff types, $\unt$ and $\ovt$, are always made indifferent between trading in the marketplace and the decentralized market.

\begin{corollary}\label{cor:thick}
	The agents who trade on the marketplace are the same with or without the decentralized market. Thus, the thickness of the marketplace is unaffected by the decentralized trade.
\end{corollary}

Although we can understand the corollary in light of the previous section, it is a counter-intuitive result on its own. When the agents have the option to trade in the decentralized market, they need to be compensated to join the centralized market. Thus, decentralized trade makes each trade in the marketplace costlier. Since the marketplace maximizes profit, it would be reasonable to expect the marketplace to become more exclusive to account for the higher costs. However, this is not optimal. Since the total compensations agents receive is a constant fraction of their virtual surplus, the optimal strategy is to maximize the virtual surplus. But $\Pi^M$ is precisely the maximum attainable virtual surplus. Thus, marketplace keeps the same agents, no matter what happens in the decentralized market.

\subsection{The Simple Economics of Optimal Marketplaces (Continued)}

I started the analysis with the centralized marketplace on its own. There, I provided the parallel between the mechanism design problem in that section and the profit-maximization problem of a monopolist. Later, I obtained the case of `single market' as a special case of `coexistence' with $p=0$. Here I will extend the analogy with the monopolist's problem to the coexistence case as well.

With $p=0$ (the case of no decentralized trade), the marketplace can ensure participation as long as the agents expect nonnegative payoffs from trade. Thus, the highest seller and the lowest buyer are offered 0 utility. However, when $p>0$, the marketplace has to make sure that each agent it wants to attract receive a utility greater than what he expects from the decentralized market. This is accomplished by paying the agents the opportunity cost of joining the marketplace as compensations. Thus, with coexistence, the effective marginal revenue and cost curves will be different.
Let $\mathcal{MR}_2$ and $\mathcal{MC}_2$ denote the new curves.

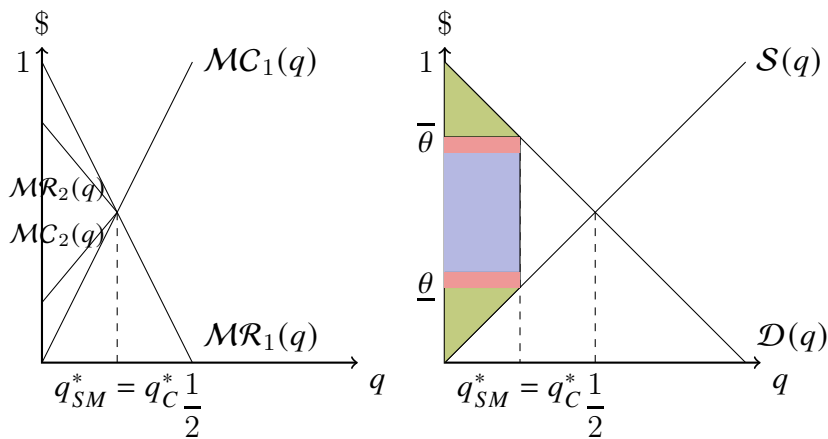
\begin{figure}[H]
	\begin{tikzpicture}[scale=0.4]
		
		\draw[thick,<->] (0,10.5) node[above]{$\$$}--(0,0)--(10.5,0) node[below right]{$q$};

		\node[] at (0,10) [left] {$1$};
		
		\draw(0,10)--(5,0) node[above right]{$\mathcal{MR}_1(q)$};
		
		\draw(0,0)--(5,10) node[right]{$\mathcal{MC}_1(q)$};
		
		\draw(0,8)--(2.5,5) node[above left]{\footnotesize $\mathcal{MR}_2(q)$};
		
		\draw(0,2)--(2.5,5) node[below left]{\footnotesize $\mathcal{MC}_2(q)$};

		\draw[dashed](2.5,5)--(2.5,0) node [below] {$q^*_{SM}=q^*_{C}$};
		
		\node[] at (5,0) [below] {$\dfrac{1}{2}$};			
		
	\end{tikzpicture}
	\centering
	\begin{tikzpicture}[scale=0.4]
		
		\draw[thick,<->] (0,10.5) node[above]{$\$$}--(0,0)--(10.5,0) node[below right]{$q$};
		
		\draw(0,0)--(10,10) node[right]{$\mathcal{S}(q)$};
		
		\draw(0,10)--(10,0) node[above right]{$\mathcal{D}(q)$};
		
		\draw[dashed](5,5)--(5,0) node [below] {$\dfrac{1}{2}$};
		
		\node[] at (0,10) [left] {$1$};

		\draw[dashed](2.5,7.5)--(2.5,0) node [below] {$q^*_{SM}=q^*_{C}$};
		
		\draw[dashed](2.5,7.5)--(0,7.5) node [left] {$\ovt$};
		
		\draw[dashed](2.5,2.5)--(0,2.5) node [left] {$\unt$};
		
		\draw[fill=solarizedViolet!50] (0,2.5) -- (2.5,2.5) -- (2.5,7.5) --(0,7.5);	
		
		\draw[fill=solarizedGreen!50] (0,10) -- (2.5,7.5) --(0,7.5);
		
		\draw[fill=solarizedGreen!50] (0,0) -- (2.5,2.5) --(0,2.5);
		
		\draw [name path=X] (0,7.5) to  (2.5,7.5) node[above left]{};
		
		\draw [name path=A] (0,7) to (2.5,7) node[above left]{};
		
		\tikzfillbetween[of=X and A]{solarizedRed!50};
		
		\draw [name path=Y] (0,2.5) to  (2.5,2.5) node[above left]{};
		
		\draw [name path=B] (0,3) to (2.5,3) node[above left]{};
		
		\tikzfillbetween[of=Y and B]{solarizedRed!50};
		
	\end{tikzpicture}
	\caption{\footnotesize \textcolor{solarizedViolet}{\textbf{Coexistence Profit of the Marketplace}}, \textcolor{solarizedGreen}{\textbf{Information Rents}}, ``\textcolor{solarizedRed}{\textbf{Compensations}}''. $q^*_{SM}$ is the optimal quantity when $p=0$ (single market) and $q^*_{C}$ is the optimal quantity when $p>0$ (coexistence).}
	\label{fig:coexistencesimpleecon}
\end{figure}

In Figure \ref{fig:coexistencesimpleecon} (left panel), we have the plots of the single market marginal revenue and cost of the marketplace, $\mathcal{MR}_1$ and $\mathcal{MC}_1$ together with the plots of the coexistence marginal revenue and cost of the marketplace, $\mathcal{MR}_2$ and $\mathcal{MC}_2$. This picture illustrates many of the results stated above.

$\mathcal{MR}_2$ and $\mathcal{MC}_2$ intersect at the same quantity level as $\mathcal{MR}_1$ and $\mathcal{MC}_1$ does. Thus, the optimal quantity under no decentralized trade, $q^*_{SM}$ is equal to the optimal quantity under coexistence for any $p>0$: $q^*_{C}$. Thus, the thickness of the centralized marketplace is unaffected by $p\in [0,1]$ as stated earlier.

As $q^*_C$ is below $\frac{1}{2}$ (that corresponds to the median of the distribution), some agents are excluded from the marketplace. But since these agents can always join the decentralized marketplace and get positive payoffs, it means there will be coexistence.

From Figure \ref{fig:coexistencesimpleecon} (right panel), compensation for each buyer (seller) is the same. Moreover, with a little more notation on the figures, it is possible to pictorially represent Theorem \ref{thm:profit} as well. \footnote{The more detailed picture is discussed in the Appendix \ref{app:simpecon}.}

\subsection{Welfare Comparisons}

In this section, I first compare the consumers' welfare in a Pareto sense for different levels of $p$. Next, I compare the welfare generated in the coexistence to (i) welfare generated by the centralized market alone and (ii) welfare generated by the decentralized trade alone. For this analysis the welfare is measured as gains from trade. Thus, it includes the profit of the marketplace, when it exists. Omitted proofs are in the Appendix \ref{app:efficiency}.

\subsubsection{Consumer Welfare}

From Corollary \ref{cor:thick}, we know that the agents who trade in the decentralized market are the same. Then, conditional on meeting someone, agents' expected payoffs from the decentralized market are the same for each value of $p$. Thus, as $p$ increases, agents' expected payoffs from the decentralized market increase. This directly leads to higher utilities for the intermediate types who join the decentralized market. It indirectly increases the payoff of each agent in the marketplace as well. When $p$ increases, the cutoff types have a higher expected payoff from the decentralized market. Since the cutoff types remain the same, to make them indifferent under a higher $p$, centralized marketplace has to increase the compensations for each agent. Thus, a more efficient decentralized market increases the payoffs for all agents, regardless of their market choice.

\begin{corollary}
	Equilibrium payoff of each type is increasing in $p$.
\end{corollary}

This provides a nice policy recommendation. Make the decentralized trade more efficient. It will make every agent strictly better off.

\subsubsection{Total Welfare}

Comparison between the coexistence and the centralized market alone is simple. Once again, by Corollary \ref{cor:thick}, the same agents trade on the marketplace, whether there is a decentralized market or not. Thus, the gains from trade generated on the marketplace is constant. However, in the coexistence, the agents with intermediate types also generate some trade. Thus, coexistence generates more gains under any regular distribution.

\begin{corollary}
	Coexistence leads to a higher welfare than the centralized market alone.
\end{corollary}

For ease of exposition, I first focus on a concrete distribution for valuations to compare the welfare from the coexistence to the welfare from the decentralized trade alone. Later, I show that the result holds much more generally.

\begin{proposition}\label{efficiencypropuni}
	Suppose agents valuations are drawn from the uniform distribution over $[0,1]$. Then, the total welfare under the coexistence is greater than when either market operates on its own.
\end{proposition}
	
Intuition behind this result is that the search market extends the extensive margin of trade by allowing more agents to trade while the centralized marketplace extends the intensive margin -by solving the search and matching frictions. Thus, when they operate together, everyone gets to trade with some probability and some agents trade for certain.
	
Although the proposition is stated for the uniform distribution, the result is true for many other distribution. Next, I provide a sufficient condition for distributions under which the coexistence is more efficient. This condition is satisfied by most commonly used distributions.

\begin{assumption}\label{expineq}
	Under $F$, the following inequality holds:
	\[ 2  \mathbb{E}[\theta F(\theta) | \theta\in [\unt,\ovt]]  + \mathbb{E}[\theta|\theta > \ovt] \geq   \mathbb{E}[\theta|\theta\in [\unt,\ovt]] + \mathbb{E}[\theta F(\theta) | \theta\leq \unt]   + \mathbb{E}[\theta F(\theta) | \theta\geq \ovt], \]
	
	where $\unt$ and $\ovt$ are the optimal cutoffs for a simple equilibrium, given the distribution $F$.
\end{assumption}

Standard Normal Distribution, Logistic Distribution, Exponential Distribution, Standard Beta Distribution are among the distributions that has this property\footnote{Mathematica codes for computations that show these distributions satisfy this assumption are available upon request.}.

\begin{proposition}\label{efficiencyprop}
	If $F$ satisfies Assumption \ref{expineq}, then the total welfare under the simple equilibrium is greater than when either market operates on its own.
\end{proposition}

\subsection{Simple Equilibrium is Without Loss}

In this section, I remove the restriction to the `simple' market structure that required an interval of agents to join the search market. Here, I allow the designer to invite any closed set of agents. As I show in Theorem \ref{unrestrictedthm}, the restriction to the simple equilibrium is without loss of generality. Thus, Assumption \ref{ass:simple} was without loss of generality. Therefore, everything we have studied so far holds for all equilibria where both markets are active.

\begin{theorem}\label{unrestrictedthm}
	If both markets are active in an equilibrium, then the sets of agents in the search market must be an interval.
\end{theorem}

This result is extremely helpful in reducing `the dimensionality' of the problem. Without the structure provided by this result, searching for an equilibrium would require considering the consequences of every possible segmentation of a continuum of agents across markets. However, when we know that the agents in the search market will be an interval, the search is essentially reduced to finding a pair of cutoffs.

The proof is covered in the Appendix \ref{unrestricted}. The intuition is as follows. The utilities from the decentralized market would be U-shaped, under any segmentation. Moreover, the utilities from the mechanism would be `bucket-shaped'; consisting of a decreasing line segment, a constant line segment, and an increasing line segment, in this order. Under any segmentation, each type would choose the market that offers higher utilities. Using this condition, I show that there are two possible segmentations that can happen with these utility shapes: Either one interval joins the decentralized market or the union of two invervals that satisfy some conditions join there. However, excluding the union of two intervals is not profit-maximizing: This would require having some agents with intermediate types in the marketplace who have small gains from trade. Thus, the unique equilibrium structure is excluding an interval of intermediate types.

\subsection{Uniqueness of the Equilibrium}

So far, we have seen that the simple equilibrium is the unique coexistence equilibrium. This leaves the question of whether there can be equilibria where only one of the markets is active.

It is relatively easy to argue that there cannot be any equilibrium where all agents join the centralized marketplace. This is because the marketplace will exclude some agents from trade even if they are in the marketplace. But then, these agents would have a profitable bilateral deviation to the decentralized market where they would get a positive payoff, rather than getting 0 utility in the marketplace. Thus, the only other possibility is an equilibrium where all agents join the decentralized market. In the following proposition, I also show that, all agents joining the decentralized market is not an equilibrium either. The proof is omitted; it can be found in Appendix \ref{app:unique}.

\begin{proposition}\label{prop:unique}\hfill
	\begin{enumerate}
		\item There is no equilibrium where all agents join the centralized marketplace.
		
		\item There is no equilibrium where all agents join the decentralized market.
		
		\item Thus, the simple equilibrium is the unique equilibrium.
	\end{enumerate}
\end{proposition}

The intuition for the result is as follows. In the equilibrium, there cannot be a pair of agents who would rather deviate to another market bilaterally. I show that even if all agents join the decentralized market, the payoffs the centralized market offers to the lowest and the highest types are greater than what they expect in the decentralized market. Thus, they would deviate to the centralized marketplace. Therefore, all agents joining the decentralized market cannot be an equilibrium. On the other hand, as I described above, it is not an equilibrium for all agents to join the centralized marketplace either. Hence, it is not an equilibrium for either one of the markets to be inactive. Finally, since the simple equilibrium exists and it is the unique equilibrium with coexistence, it is indeed the unique equilibrium.

\subsection{Simple Equilibrium with Double Auction}

So far, I assumed that when agents are matched to each other in the decentralized market, they engage in Nash bargaining. This provides a nice benchmark since Nash bargaining is efficient: Even with an efficient trading protocol, the marketplace can attract the agents it targets and make positive profit. However, in Nash bargaining, the assumption is that the agents observe each others' valuations upon meeting. This creates an information asymmetry between markets since the designer cannot observe agents' valuations. Thus, we might be concerned that the results I have presented may not be robust to other models of decentralized trade. In this section, I consider an alternative bargaining protocol and I show that under some conditions, all of the main results hold true more generally.

Suppose the agents who join the decentralized market are randomly matched to each other and then they participate in a double auction. The random matching process is unchanged from the main model. The double auction works as follows. Each agent submits a bid, $b_i$. Agent with the higher bid buys the other's endowment and pays $\frac{b_i+b_j}{2}$, when bids are given by $b_i$ and $b_j$; so they are trading a unit of the good at the midpoint of the bids.

This is a special case of the ``simple trading rule'' studied by \cite{gck}. Assuming the valuations are drawn from the same, smooth distribution for each agent, they show that this game has a symmetric equilibrium where the bids are increasing in agents' valuations. Thus, the equilibrium is ex-post efficient in the sense that, the agent who values the good more ends up with the whole quantity. Moreover, \cite{kitt} has shown that the symmetric equilibrium characterized by \cite{gck} is indeed the unique equilibrium. 

Suppose the distribution of agents who participates in this double auction is given by some CDF, $G$. Suppose $G$ is strictly increasing on its support, $[\unt, \ovt]$, and is differentiable. Then, agents' bids in the unique equilibrium are given by:

\begin{equation*}
	b(\theta)=\theta - \dfrac{\int\limits_{G^{-1}(\frac{1}{2})}^{\theta} [G(x)-\dfrac{1}{2}]^2 dx }{[G(\theta)-\dfrac{1}{2}]^2}
\end{equation*}

It follows from the Proposition 5 of \cite{gck} that following this bidding strategy constitutes an equilibrium. Theorem 1 in \cite{kitt} further shows that there is no other equilibrium.

Here, I focus on simple mechanisms that exclude an interval of agents: Suppose agents in $(\unt,\ovt)$ join the decentralized market and the rest join the centralized marketplace. Then, the endogenous distribution of agents in the decentralized market is simply $F$ truncated from both below and above. Thus, $G(\theta)=\dfrac{F(x)-F(\unt)}{F(\ovt)-F(\unt)}$ on its support $(\unt,\ovt)$.\footnote{The equilibrium of this double auction game when $G$ has gaps in its support is not known.}

In the Appendix \ref{app:doubleauctionmain}, I develop an analysis parallel to that on Nash bargaining for the double auction. I summarize my results in this section. Suppose the agents' valuations are drawn from the uniform distribution over $[0,1]$. Then, essentially every result I obtained under the Nash bargaining holds true for the decentralized market with double auction as well:

\begin{theorem}\label{2cutoffdoubleauction}
	There exists $\unt,\ovt $ such that in the simple equilibrium,
	\begin{itemize}
		\item agents in $[0,\unt]$ and $[\ovt ,1]$ join the mechanism,
		\item agents in $(\unt,\ovt)$ join the decentralized market with double auction,
		\item $\mathcal{C}(\unt) = \mathcal{V}(\ovt )$ and $F(\unt)=1-F(\ovt )$.\vspace{0.2in}
		
	\end{itemize}
	
\end{theorem}

\begin{proposition}
	The profit of the marketplace in coexistence is $1-\frac{5p}{6}=1-\frac{5m}{3}$ times the profit the marketplace would make if there were no decentralized market with double auction.
\end{proposition}

\begin{proposition}
	The total welfare under the simple equilibrium with the double auction is greater than when either market operates on its own.
\end{proposition}

\begin{proposition}\label{prop:da-thick}
	Same agents trade in the marketplace, with or without the decentralized market.
\end{proposition}

Thus, exactly same agents trade in the centralized marketplace (i) for each level of friction in the decentralized market (including the extreme case of no decentralized trade with $p=0$) and (ii) whether the decentralized trade happens according to Nash bargaining or a double auction.

\section{Multiple Designers}

A natural question to ask is how the equilibrium would change if instead of one centralized marketplace competing with a decentralized market, there were multiple centralized marketplaces competing with each other. This is closer to the models extensively studied by the competing platforms literature. However, it is still worthwhile to understand what the structure of the competition would look like within the framework of this paper.

Mechanism design problems with multiple designers are notoriously difficult to solve. The difficulty lies in the fact that when there are multiple designers, each one wants to condition her mechanism to the others' and none of them wants to post a mechanism that only depends on the agents' types. As a result, revelation principle does not hold in environments with multiple designers. Despite recent developments \citep{hartlineyiding}, mechanism design without revelation principle is itself a notorious problem, even with a single designer.

One way the literature has dealt with this problem of infinite regress between competing mechanism is to explicitly assume that the designers each choose a direct mechanism. Although this is not without loss anymore, it at least provides a tractable way to study the problem. I will first follow this literature and then illustrate what else can be achieved with a non-direct mechanism.

Suppose there are $n$ marketplace designers competing with each other for agents' participation. Each of them is a profit-maximizer. Suppose they are restricted to choose direct mechanisms. The timeline is as follows: First, designers announce their mechanisms simultaneously. Then, agents choose which marketplace to join simultaneously. Finally, trades in each marketplace realize simultaneously.

From the analysis of the Online Appendix\footnote{Here is a link to the Online Appendix: \url{https://berkidem.com/documents/coexistence_online_appendix.pdf}.}, we know that this is equivalent to saying that each designer $i$ chooses a pair of prices (bid-ask prices) $p_b^i$ and $p_s^i$ for buying and selling in the marketplace $i$, respectively.

This environment is essentially a model of Bertrand competition between the marketplaces where agents' roles as buyers and sellers are endogenously determined. The reasoning from the Bertrand Equilibrium still holds: If some designer $i$ makes a positive profit, then another designer $j$ can offer slightly smaller bid-ask spread, and increase her own profit discontinuously by recruiting the agents who were trading on the marketplace $i$. Thus, there cannot be an equilibrium where a designer makes a positive profit. Conversely, there cannot be an equilibrium where a designer makes a negative profit: It is always feasible to make 0 profit by setting a prohibitively high price to everyone.

Next, I argue that there are equilibria where all designers make 0 profit: Consider the strategy profiles where (i) at least one of the marketplaces set $p_b=p_s=m(F)$ where $m(F)$ is the median of the distribution of the valuations, (ii) the rest of the marketplaces either shut down or set prohibitively unfavorable prices, and (iii) all agents uniformly randomize over marketplaces where $p_b=p_s=m(F)$. Clearly, none of the designers or agents can improve their payoffs strictly and all such strategy profiles constitute equilibria.\footnote{Notice that there are many other payoff-equivalent equilibria. For instance, for $n=2$, suppose we have $p_b=p_s=m(F)$ on both marketplaces. Then, the agents are indifferent between the marketplaces. Any segmentation such that the market is cleared in each marketplace is also an equilibrium.} Finally, there is no equilibrium where all marketplaces are inactive: If this happened, then each marketplace would simply deviate to announcing the baseline mechanism and make a positive profit.

Notice that $m(F)$ is the price that would clear the market if we had a competitive market: Setting the price at $m(F)$ ensures that the supply is equal to demand. Thus, the equilibria of this game replicate the Walrasian equilibrium due to the nature of the competition here.

We have seen that introducing another profit-maximizing marketplace reduced the profits to zero. In the previous section, we have observed that when we instead considered a competing decentralized market, the profit only decreased by a constant ratio, always less than the half. This stark difference is a result of the restriction to direct mechanisms when they are not without loss of generality. Next, I will briefly discuss the simplest way in which mechanisms can depend on each other and how it allows the designers to recoup the profits.

Now suppose the designers can announce any kind of mechanism. Thus, they can make their prices depend on the other marketplaces. One of the simplest ways to utilize this interdependence is by posting prices together with price-matching guarantees.

In an equilibrium with price-matching guarantees, each designer $i$ announces prices $p_b^i$ and $p_s^i$, and also promises to honor the prices $\min_{j}\{p_b^j:j=1,\ldots,n\}$ and $\max_j\{p_s^j:j=1,\ldots,n\}$ as well; so each agent can get the lowest price for buying and the highest price for selling from each marketplace. This is in line with introducing price-matching guarantees in a Bertrand competition setting and obtaining an equilibrium with monopoly pricing \citep{pricematching}.

In the framework of this paper, it is see an equilibrium for all designers to announce the prices that are equivalent to the baseline mechanism and for all agents to uniformly randomize over all the marketplaces. In this case, the designers are splitting the baseline profit equally. Thus, they essentially operate as a cartel.

It is worth noting that this is not the unique equilibrium in this environment. Specifically, it is still an equilibrium for each designer to announce the median of the valuations as the posted price for both the buyers and sellers.\footnote{Requiring coalition-proofness for the designers as well would make sure that they are collectively achieving the baseline profit in the equilibrium.}

Finally, in addition to choosing among these marketplaces, suppose agents also have the option to join the decentralized market, after observing the posted mechanisms. Then, the same argument works with a minor modification: The designers would each post prices equivalent to the coexistence mechanism with price-matching guarantees and the agents with extreme values would randomize over the marketplaces uniformly. Since no marketplace has any incentives to deviate, they collectively act as the centralized market from the main model. Given that the centralized market is the same (on the equilibrium-path), no agent has any incentives to deviate either. Thus, the same segmentation would constitute an equilibrium with the addition that the agents with extreme values now randomize over the marketplaces.

	\clearpage
	\bibliographystyle{apacite}
	\bibliography{marketplace}

\begin{thebibliography}{}

\bibitem[Armstrong, 2006]{armstrong}
Armstrong, M. (2006).
\newblock Competition in two-sided markets.
\newblock {\em The RAND Journal of Economics}, 37(3):668--691.

\bibitem[Ashlagi and Roth, 2014]{ashlagiroth}
Ashlagi, I. and Roth, A.~E. (2014).
\newblock Free riding and participation in large scale, multi-hospital kidney
  exchange.
\newblock {\em Theoretical Economics}, 9(3):817--863.

\bibitem[Bulow and Roberts, 1989]{bulow}
Bulow, J. and Roberts, J. (1989).
\newblock The simple economics of optimal auctions.
\newblock {\em Journal of political economy}, 97(5):1060--1090.

\bibitem[{Congress Majority Staff}, 2020]{amazon}
{Congress Majority Staff} (2020).
\newblock Investigation of competition in digital markets.

\bibitem[Cramton et~al., 1987]{gck}
Cramton, P., Gibbons, R., and Klemperer, P. (1987).
\newblock Dissolving a partnership efficiently.
\newblock {\em Econometrica: Journal of the Econometric Society}, pages
  615--632.

\bibitem[Diamond, 1971]{diamond}
Diamond, P.~A. (1971).
\newblock A model of price adjustment.
\newblock {\em Journal of economic theory}, 3(2):156--168.

\bibitem[Ekmekci and Yenmez, 2019]{ekmekciyenmez}
Ekmekci, M. and Yenmez, M.~B. (2019).
\newblock Common enrollment in school choice.
\newblock {\em Theoretical Economics}, 14(4):1237--1270.

\bibitem[Feng and Hartline, 2018]{hartlineyiding}
Feng, Y. and Hartline, J.~D. (2018).
\newblock An end-to-end argument in mechanism design (prior-independent
  auctions for budgeted agents).
\newblock In {\em 2018 IEEE 59th Annual Symposium on Foundations of Computer
  Science (FOCS)}, pages 404--415. IEEE.

\bibitem[Fieseler et~al., 2003]{fieselerkittsteinermoldovanu}
Fieseler, K., Kittsteiner, T., and Moldovanu, B. (2003).
\newblock Partnerships, lemons, and efficient trade.
\newblock {\em Journal of Economic Theory}, 113(2):223--234.

\bibitem[Figueroa and Skreta, 2012]{figueroaskreta}
Figueroa, N. and Skreta, V. (2012).
\newblock Asymmetric partnerships.
\newblock {\em Economics Letters}, 115(2):268--271.

\bibitem[Fuchs and Skrzypacz, 2015]{fuchsskrzypacz}
Fuchs, W. and Skrzypacz, A. (2015).
\newblock Government interventions in a dynamic market with adverse selection.
\newblock {\em Journal of Economic Theory}, 158:371--406.

\bibitem[Hartline and Roughgarden, 2014]{hartline}
Hartline, J.~D. and Roughgarden, T. (2014).
\newblock Optimal platform design.
\newblock {\em arXiv preprint arXiv:1412.8518}.

\bibitem[Hess and Gerstner, 1991]{pricematching}
Hess, J.~D. and Gerstner, E. (1991).
\newblock Price-matching policies: An empirical case.
\newblock {\em Managerial and Decision Economics}, 12(4):305--315.

\bibitem[Hinote, 2021]{airbnb}
Hinote, A. (2021).
\newblock Taxi and ridehailing usage in new york city.

\bibitem[Khan, 2016]{khan}
Khan, L.~M. (2016).
\newblock Amazon's antitrust paradox.
\newblock {\em Yale lJ}, 126:710.

\bibitem[Kittsteiner, 2003]{kitt}
Kittsteiner, T. (2003).
\newblock Partnerships and double auctions with interdependent valuations.
\newblock {\em Games and Economic Behavior}, 44(1):54--76.

\bibitem[Loertscher and Wasser, 2019]{loertscherwasser}
Loertscher, S. and Wasser, C. (2019).
\newblock Optimal structure and dissolution of partnerships.
\newblock {\em Theoretical Economics}, 14(3):1063--1114.

\bibitem[Miao, 2006]{miao}
Miao, J. (2006).
\newblock A search model of centralized and decentralized trade.
\newblock {\em Review of Economic dynamics}, 9(1):68--92.

\bibitem[Mortensen, 1970]{mortensen}
Mortensen, D.~T. (1970).
\newblock Job search, the duration of unemployment, and the phillips curve.
\newblock {\em The American Economic Review}, 60(5):847--862.

\bibitem[Myerson and Satterthwaite, 1983]{bilateral}
Myerson, R.~B. and Satterthwaite, M.~A. (1983).
\newblock Efficient mechanisms for bilateral trading.
\newblock {\em Journal of economic theory}, 29(2):265--281.

\bibitem[Mylovanov and Tr{\"o}ger, 2014]{mylovanovtroger}
Mylovanov, T. and Tr{\"o}ger, T. (2014).
\newblock Mechanism design by an informed principal: Private values with
  transferable utility.
\newblock {\em The Review of Economic Studies}, 81(4):1668--1707.

\bibitem[Peivandi and Vohra, 2021]{vohra}
Peivandi, A. and Vohra, R.~V. (2021).
\newblock Instability of centralized markets.
\newblock {\em Econometrica}, 89(1):163--179.

\bibitem[Philippon and Skreta, 2012]{philipponskreta}
Philippon, T. and Skreta, V. (2012).
\newblock Optimal interventions in markets with adverse selection.
\newblock {\em American Economic Review}, 102(1):1--28.

\bibitem[Pissarides, 1979]{pissarides}
Pissarides, C.~A. (1979).
\newblock Job matchings with state employment agencies and random search.
\newblock {\em The Economic Journal}, 89(356):818--833.

\bibitem[Rochet and Tirole, 2003]{rochet}
Rochet, J.-C. and Tirole, J. (2003).
\newblock Platform competition in two-sided markets.
\newblock {\em Journal of the european economic association}, 1(4):990--1029.

\bibitem[Roth and Shorrer, 2021]{rothshorrer}
Roth, B.~N. and Shorrer, R.~I. (2021).
\newblock Making marketplaces safe: Dominant individual rationality and
  applications to market design.
\newblock {\em Management Science}, 67(6):3694--3713.

\bibitem[Schneider, 2021]{uber}
Schneider, T.~W. (2021).
\newblock Taxi and ridehailing usage in new york city.

\bibitem[Tirole, 2012]{tirole}
Tirole, J. (2012).
\newblock Overcoming adverse selection: How public intervention can restore
  market functioning.
\newblock {\em American economic review}, 102(1):29--59.

\end{thebibliography}


\begin{thebibliography}{}

\bibitem[Idem, 2021]{idem}
Idem, B. (2021).
\newblock Optimal marketplace design.
\newblock {\em Working Paper}.

\end{thebibliography}

\appendix

\section{Appendix}

\subsection{Integrating the Virtual Surplus}\label{virtualsurplus}

\begin{align*}
	&-\int\limits_0^{\unt}C(x)f(x)dx +\int\limits^1_{\ovt }V(y)f(y)dy \\
	&=-\int\limits_0^{\unt}C(x)f(x)dx +\int\limits^1_{\ovt }V(y)f(y)dy \\
	&= -\int\limits_0^{\unt}\left[ x +\dfrac{F(x)}{f(x)} \right]f(x)dx +\int\limits^1_{\ovt }\left[ y -\dfrac{1-F(y)}{f(y)} \right]f(y)dy \\
	&=	-\int\limits_0^{\unt}\left[x \right]f(x)dx -\int\limits_{0}^{\unt} \left[F(x)\right]dx +\int\limits^1_{\ovt }\left[ y  \right]f(y)dy -\int\limits^1_{\ovt }\left[1-F(y) \right]dy  \\
	&= 
	-\left( \left[x F(x) \right]_0^{\unt} -\int\limits_{0}^{\unt} \left[F(x)\right]dx  \right) -\int\limits_{0}^{\unt} \left[F(x)\right]dx \\
	&+ \left( \left[ y F(y) \right]_{\ovt }^1 -\int\limits^1_{\ovt }\left[F(y) \right]dy \right) -\int\limits^1_{\ovt }\left[1-F(y) \right]dy \\
	&= 	- \unt F(\unt)
	+  1-\ovt F(\ovt )  -\int\limits^1_{\ovt }\left[F(y) \right]dy -(1-\ovt ) +\int\limits^1_{\ovt }\left[F(y) \right]dy \\
	&= - \unt F(\unt)
	+  1-\ovt F(\ovt )   -1+\ovt    \\
	&= - \unt F(\unt) -\ovt F(\ovt ) +\ovt 
\end{align*}

\subsection{Optimal Baseline Marketplace}\label{proof:monagora}

	\begin{proof}[Proof of Theorem \ref{thm:monagora}]
		Since the allocation needs to be increasing, if $q(\theta)<0$ for some $\theta$, we would have $q(\theta')<0$ for each $\theta'\leq \theta$. Similarly, if $q(\theta)>0$ for some $\theta$, we would have $q(\theta')>0$ for each $\theta'\geq \theta$. So, let $0\leq \unt \leq \ovt \leq 1$ such that $\unt$ is the supremum of values with negative allocation and $\ovt$ is the infimum of the values with positive allocation. Then, we can write the objective function as follows:

		\begin{align*}
			\Pi^M&= \mathbb{P}[\theta\in[0,\unt]] \mathbb{E}[\mathcal{C}(\theta)q(\theta)|\theta\in[0,\unt]] + \mathbb{P}[\theta\in[\ovt,1]] \mathbb{E}[\mathcal{V}(\theta)q(\theta)|\theta\in[\ovt,1]]\\
			&=\int\limits_0^{\unt}    \mathcal{C}(x) q(x) f(x)dx + \int\limits_{\ovt}^1   \mathcal{V}(x) q(x) f(x)dx
		\end{align*}
		
		Note that for $\theta \leq \unt$, we must have $q(\theta) = -1$ as $-1$ is the only possible negative allocation with the indivisible good and similarly, $q(\theta) = 1$ for $\theta\geq \ovt$. Thus, the optimal allocation will have the following form and the next step is to choose the cutoffs, $\unt$ and $\ovt$ optimally.
		\begin{align*}
			q(\theta)=\begin{cases}
				-1 &\text{ if } \theta\leq \unt\\
				0 &\text{ if } \unt \leq \theta\leq \ovt\\
				1 &\text{ if } \theta\geq \ovt
			\end{cases}
		\end{align*}
		
		Then, the problem can be restated as follows:
		\begin{equation*}
			\begin{array}{llllll}
			
				\displaystyle\max_{\unt,\ovt}  &\left[ -\int\limits_0^{\unt}    \mathcal{C}(x) f(x)dx + \int\limits_{\ovt}^1   \mathcal{V}(x)  f(x)dx \right] \\

				\text{s. t. }&\\
				
				&F(\unt) = 1- F(\ovt)\\
				
				& 0\leq \unt\leq \ovt \leq 1
				
			\end{array}
		\end{equation*}
		
		Integrating out the total virtual value and the virtual cost using integration by parts\footnote{See Appendix \ref{virtualsurplus} for details.} shows that the objective function is equal to:
		\begin{equation*}
			-\unt F(\unt)  + \ovt (1-F(\ovt ))  =  F(\unt) (\ovt -\unt) 
		\end{equation*}
		
		where the equality is obtained by using the feasibility condition $F(\unt)=1-F(\ovt)$. Then, there exists a solution to this problem. Moreover, the solution is interior: If $\unt=0$ or $\unt=\ovt$, the profit is 0. However, positive profit is feasible by any feasible interior solution as is clear from the objective function.
		
		Moreover, by using the formula for the transfer rule and the optimal allocation from above, we can compute the transfers in the mechanism to be
		\begin{align*}
			t(\theta)=\begin{cases}
				-\unt &\text{ if } \theta\leq \unt\\
				0 &\text{ if } \unt \leq \theta\leq \ovt\\
				\ovt &\text{ if } \theta\geq \ovt
			\end{cases}
		\end{align*}
		
		Notice that each `seller' gets the same payment while each `buyer' pays the same amount. Thus, this mechanism is equivalent to offering bid-ask prices that the transfer rule above suggest, that is a price for buying and a price for selling, and letting agents choose whether they want to buy or sell or not trade.
		
		Finally, notice that the solution must have $\mathcal{C}(\unt) = \mathcal{V}(\ovt)$: If we had $\mathcal{C}(\unt) > \mathcal{V}(\ovt)$, decreasing $\unt$ and adjusting $\ovt$ accordingly for feasibility would increase the profit since buying from $\unt$ is costlier than what selling to $\ovt$ pays off. Similarly, if we had $\mathcal{C}(\unt) < \mathcal{V}(\ovt)$, then increasing $\unt$ and adjusting $\ovt$ so that the feasibility binds would again increase the profit since there are more agents whose trade is profitable.
	\end{proof}

	\subsubsection{Illustrative Example with Uniform Distribution}
	
	Suppose $\theta$ is distributed uniformly over $\left[ 0,1 \right]$ with c.d.f. $F(\theta) =\theta$. Then, the virtual values and costs are given by $$ \mathcal{V}(\theta)=2\theta-1 \text{ and } \mathcal{C}(\theta)=2\theta\,. $$
	
	It is easy to show that the objective function becomes $\unt(1-2\unt)$ after substituting for $\ovt=1-\unt$ (feasibility). So, the optimal cutoffs are $\unt=\frac{1}{4}$ and $\ovt=\frac{3}{4}$.
	
	Let us consider two agents with valuations $\theta_1$ and $\theta_2$. Then, 1 will buy from 2 in the optimal marketplace when it operates on its own if and only if $\theta_1\geq 0.75$ and $\theta_2\leq 0.75$, and 2 will buy from 1 if and only if $\theta_2\geq 0.75$ and $\theta_1\leq 0.75$. The Figure \ref{fig:graph2} depicts the space of $(\theta_1, \theta_2)$ where the shaded areas represent these trading regions.

	\begin{figure}[h!]
		\centering
		\begin{tikzpicture}[font=\sffamily\small]
			\draw[->] (0,0) -- (4.5,0) node[below] {$\theta_1$};
			\draw[->] (0,0) -- (0,4.5) node[left] {$\theta_2$};
			
			\foreach \i in {0.25,0.5,0.75,1} \draw (4*\i,1mm) -- (4*\i,-1mm) node[below] {\i}; 
			\foreach \i in {0, 0.25, 0.5, 0.75, 1} \draw (0.25mm,4*\i) -- (-0.25mm,4*\i) node[left] {\i};

			\path[draw, fill=solarizedRed!20] (0,3)--(1,3)--(1,4)--(0,4)--cycle;
			
			\path[draw, fill=solarizedGreen!20] (3,0)--(4,0)--(4,1)--(3,1)--cycle;
			
		\end{tikzpicture}
		\caption{$x$-axis represents $\theta_1\in [0,1]$ and $y$-axis represents $\theta_2\in [0,1]$. Green and red areas show the type profiles at which agent 1 and 2 is the buyer, respectively.}
		\label{fig:graph2}
	\end{figure}
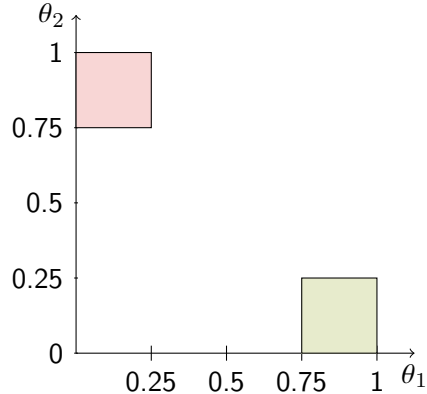

	Using the payment formula from the theorem above, we can compute the payments as below:
	
	$$t(\theta) =
	\begin{cases}
		-\dfrac{1}{4}, &\text{ if } \theta \leq \dfrac{1}{4}\,,\\
		0, &\text{ if }  \dfrac{1}{4} \leq  \theta \leq \dfrac{3}{4}\,,\\
		\dfrac{1}{4}, &\text{ if } \theta \geq \dfrac{3}{4}\,.
	\end{cases}$$

	The large area of the type space where there is no trade motivates our consideration for the decentralized market. The marketplace excludes these type profiles from trade because their is not profitable; this is akin to a monopolist or a monopsonist excluding some agents from trade. Indeed, the marketplace acts as both a monopolist and a monopsonist. However, unlike the buyers a monopolist excludes, the agents the marketplace excludes actually have a surplus that they can create, if they were allowed to trade. Thus, assumption that they will remain in the marketplace even though they are not trading is not very realistic. Next, I am going to allow them to choose between the marketplace and a decentralized market.
	
\subsection{Payoffs from Search}\label{app:payoffs}

From the Assumption \ref{ass:simple}, suppose $(\unt,\ovt)$ joins the decentralized market.

If an agent with valuation $\theta$ joins the search market, then, as explained in the Section \ref{payoffsearch}, her net expected payoff is

\begin{align*}
	u^d(\theta) = p\mathbb{P}[x>\theta | x \in [\unt, \ovt ]] \mathbb{E}\left[ \dfrac{x-\theta}{2} | x \in [\unt, \ovt ], x>\theta \right] + p\mathbb{P}[x<\theta | x \in [\unt, \ovt ]] \mathbb{E}\left[ \dfrac{\theta-x}{2} | x \in [\unt, \ovt ], x<\theta \right]
\end{align*}

where $p$ is the probability of matching someone, $\mathbb{P}[x>\theta | x \in [\unt, \ovt ]]$ and $\mathbb{P}[x<\theta | x \in [\unt, \ovt ]]$ are the probabilities that the match has a higher or lower values, respectively, and $\mathbb{E}\left[ \dfrac{x-\theta}{2} | x \in [\unt, \ovt ], x>\theta \right] $ and $\mathbb{E}\left[ \dfrac{\theta-x}{2} | x \in [\unt, \ovt ], x<\theta \right]$ are the expected shares of the surplus given the match has a higher or a lower value, respectively. The denominator 2 in the expression for the surplus comes from the fact that each agent gets half of the surplus from the trade in the Nash Bargaining.

The expected payoff will have a different shape for agents with valuations in $(\unt,\ovt)$, $[0,\unt]$ and $[\ovt,1]$. In the first case, the agent can either be a buyer or a seller since there are agents with valuations below and above his valuation. In the second case the agent can only be a seller and in the third case the agent can only be a buyer in the decentralized market. I consider the expected payoffs for these situations case-by-case.

Suppose $\theta\in (\unt, \ovt ) $. Then,
\begin{align*}
	u^d(\theta) &= p\dfrac{F(\ovt) - F(\theta) }{2(F(\ovt) - F(\unt)) } \left[ \int_{\theta}^{\ovt} \dfrac{x f(x)dx}{F(\ovt) - F(\theta)} - \theta \right]  + p\dfrac{F(\theta) - F(\unt) }{2(F(\ovt) - F(\unt)) } \left[ \theta - \int_{\unt}^{\theta} \dfrac{x f(x)dx}{F(\theta) - F(\unt)} \right] \\
	&= \dfrac{ p }{2(F(\ovt) - F(\unt)) } \left[ \int_{\theta}^{\ovt} x f(x)dx - \theta [F(\ovt) - F(\theta)]   +  \theta [F(\theta) - F(\unt)] - \int_{\unt}^{\theta} x f(x)dx \right] \\
	&= \dfrac{ p }{2(F(\ovt) - F(\unt)) } \left[ \int_{\theta}^{\ovt} x f(x)dx - \int_{\unt}^{\theta} x f(x)dx + \theta [ 2  F(\theta) - F(\ovt)  - F(\unt)] \right] 
\end{align*}

Suppose $\theta \leq \unt$. Then,
\begin{align*}
	u^d(\theta) &= \dfrac{p }{2} \left[ \int_{\unt}^{\ovt} \dfrac{x f(x)dx}{F(\ovt) - F(\unt)} - \theta \right] \\
	&= \dfrac{ p }{2 (F(\ovt) - F(\unt)) } \left[ \int_{\unt}^{\ovt} x f(x)dx - \theta [F(\ovt) - F(\unt)]  \right]  
\end{align*}

Similarly, if $\theta \geq \ovt $ then,
\begin{align*}
	u^d(\theta)&= \dfrac{p }{2} \left[ \theta - \int_{\unt}^{\ovt} \dfrac{x f(x)dx}{F(\ovt) - F(\unt)} \right] \\
	&= \dfrac{ p }{2 (F(\ovt) - F(\unt)) } \left[ \theta [F(\ovt) - F(\unt)] - \int_{\unt}^{\ovt} x f(x)dx   \right]  
\end{align*}

Notice that I write the net utilities from trade, not the total utilities. The total utility would also include the utility they get from consuming their endowment. To be consistent, I also compute the net utilities mechanism promises to the agents. Hence, I ignore the utility they get from consuming their own endowment and simply focus on net utilities from each market.

We differentiate the utilities to understand their shape:

\begin{align*}\label{slopes}
	\dfrac{ \partial u^d(\theta)}{\partial \theta} =  \begin{cases}
		- \dfrac{p }{2} &\text{ if } \theta\leq \unt,\\
		\dfrac{ p(2 F(\theta) - F(\unt ) - F(\ovt) )}{2(F(\ovt) - F(\unt)) } & \text{ if } \in [\unt, \ovt],\\
		\dfrac{p }{2} &\text{ if } \theta\geq \ovt  \end{cases} 
\end{align*}

For $\theta \in [\unt , \ovt]$, we have:

\begin{align*}
	& 2F(\ovt) \geq 2 F(\theta) \geq 2F(\unt) \\
	\iff & F(\ovt) - F(\unt ) \geq 2 F(\theta) - F(\unt ) - F(\ovt) \geq F(\unt) - F(\ovt) \\
	\iff & \dfrac{1}{2} \geq \dfrac{ 2 F(\theta) - F(\unt ) - F(\ovt) }{2(F(\ovt) - F(\unt)) } \geq - \dfrac{1}{2}
\end{align*}

Thus, the slope of the expected utility from the search is always in the interval $[-\frac{p}{2}, \frac{p}{2}]$. This will be useful when we want to show where the utilities the mechanism offers and the agents expect from search can cross each other.

Figure \ref{outsideutilitiesapp} shows the utilities from the search market for all agents given that the valuations are drawn from the uniform distribution over $[0,1]$ and agents with valuations in $(0.1, 0.9 )$ join the search market.

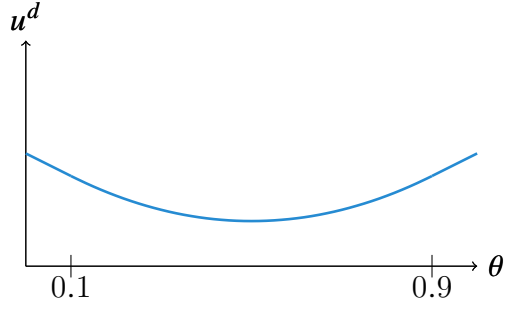
\begin{figure}[h]
	\centering 
	\begin{tikzpicture} [scale=6]
		\draw[->] (0, 0) -- (1, 0) node[right] {$\theta$};
		\draw[->] (0, 0) -- (0, 0.5) node[above] {$u^d$};
		
		\draw[->] (0, 0) -- (1, 0) node[right] {$\theta$};
		\draw[->] (0, 0) -- (0, 0.5) node[above] {$u^d$};
		
		\draw[domain=0:0.1, smooth, variable=\x, solarizedBlue, line width=1pt] plot ({\x}, {(1-2*\x)/4});
		\draw[domain=0.1:0.9, smooth, variable=\x, solarizedBlue, line width=1pt] plot ({\x}, {(\x^2-\x+0.41)/1.6});
		\draw[domain=0.9:1, smooth, variable=\x, solarizedBlue, line width=1pt] plot ({\x}, {(\x+\x-1)/4});

		\draw[] (0.1,0.25mm) -- (0.1,-0.25mm) node[left]{}; 
		\draw[] (0.9,0.25mm) -- (0.9,-0.25mm) node[left]{};	
		\node[] at (0.1,0) [below] {$0.1$};		
		\node[] at (0.9,0) [below] {$0.9$};
	\end{tikzpicture}
	\caption{Agents' utilities from their outside options as a function of their valuations, given $(0.1,0.9)$ join the decentralized market.}
	\label{outsideutilitiesapp}
\end{figure}

\subsection{Proof of the Obedience Principle}\label{proof:obed}

		\begin{proof}[Proof of Proposition \ref{prop:obed}]
		First, note that if some agents join the marketplace but the reported types in the marketplace does not satisfy the requirements for the operation of the announced mechanism, then, the marketplace shuts down and agents there gets 0 utility. Then, they would strictly prefer to be in the decentralized market since the expected payoff there is always positive. Thus, a different set of agents than the invited ones can only be in the marketplace if the marketplace is active. Moreover, notice that when the marketplace is active, the incentive compatibility of the mechanism ensures that if an agent joins there, the truthful revelation is better than misreporting one's type.
		
		Suppose the designer announced $\unt,\ovt$ (meaning she invited $[0,\unt]\cap [\ovt,1]$) and a strict subset of invited types, $[0,a]\cap [b,1]$ that satisfy the above requirements joined the mechanism (and reported truthfully, by incentive compatibility). For this to be an equilibrium, agents with valuations $a$ and $b$ must be indifferent between the two market. To see this, suppose the utilities $a$ or $b$ get from the mechanism is strictly higher than their expected payoffs from search. Then, the designer can increase the transfers to strictly increase the profit. Now suppose the expected payoff from search is strictly lower than the utility from the mechanism for $a$ ($b$). Note that both the expected payoff from search and the utilities the mechanism offer are continuous in agents' own valuations. Then, a positive measure of agents with valuations just above $a$ (below $b$) must also have been offered a higher utility in the mechanism than what they expect from search, which contradicts their joining the search market in an equilibrium. Thus, we conclude that in the equilibrium, the realized cutoffs in this potential equilibrium, $a$ and $b$ must be indifferent between two markets. We now write the utilities these agents expect from both markets.
		
		\begin{align*}
			u^d(a)&=\dfrac{p}{2}\left[\mathbb{E}[\theta| \theta\in [a,b]] - a\right] & = \dfrac{p}{2}\left[\mathbb{E}[\theta| \theta\in [\unt,\ovt]] - \unt\right] + \unt - a=u^m(a)\\
			u^d(b)&=\dfrac{p}{2}\left[b -\mathbb{E}[\theta| \theta\in [a,b]]\right] & = \dfrac{p}{2}\left[\ovt -\mathbb{E}[\theta| \theta\in [\unt,\ovt]]\right] + b - \ovt=u^m(b)\\
		\end{align*}
		
		where we calculate the expected payoff from the mechanism using the payoff equivalence: The mechanism would make $\unt$ and $\ovt$ indifferent between two markets and then have a slope of $-0.5$ for sellers and of $0.5$ for buyers.
		
		By adding these equations up, we end up with
		
		\begin{align*}
			\dfrac{p}{2}\left[b - a\right]  = & \dfrac{p}{2}\left[\ovt - \unt\right] + b+\unt - a-\ovt\\
			\dfrac{p}{2}\left[b - a-\ovt +\unt\right] = & b+\unt - a-\ovt\\
			p=&2,
		\end{align*}
		
		which is of course impossible since $p$ is a probability. Thus, we conclude that there is no equilibrium where the set of agents in the mechanism is nonempty and different from the set of invited agents.
	\end{proof}

\subsection{Simplifying the Profit Function}\label{symprofit}

Here we simplify the profit function.

\begin{equation*}
	\Pi_{\unt,\ovt } = \int\limits_0^{\unt} t(\theta) f(\theta) d\theta + \int\limits_{\ovt}^1 t(\theta)f(\theta)d\theta.
\end{equation*}

We will study each integral separately. We start with the first one.

\begin{align}
	&\int\limits_0^{\unt} t(\theta) f(\theta) d\theta\\
	&= \int\limits_0^{\unt} \left[ \theta q(\theta) -u^m(0) - \int\limits_{0}^{\theta}q(x)dx \right]f(\theta)d\theta \\
	&= \int\limits_0^{\unt} \left[ \theta q(\theta) -u^m(0)\right]f(\theta)d\theta - \int\limits_0^{\unt}\int\limits_{0}^{\theta}q(x)f(\theta)dx d\theta\\
	&= \int\limits_0^{\unt} \left[ \theta q(\theta) -u^m(0)\right]f(\theta)d\theta - \int\limits_0^{\unt}\int\limits_{x}^{\unt}q(x)f(\theta)d\theta dx\\
	&= \int\limits_0^{\unt} \left[ \theta q(\theta) -u^m(0)\right]f(\theta)d\theta - \int\limits_0^{\unt} q(x) \int\limits_{x}^{\unt}f(\theta) d\theta dx \\
	&= \int\limits_0^{\unt} \left[ \theta q(\theta) -u^m(0)\right]f(\theta)d\theta - \int\limits_0^{\unt} q(x) (F(\unt)-F(x)) dx\\
	&= \int\limits_0^{\unt} \left[  -u^m(0) + \left( x - \dfrac{F(\unt)-F(x)}{f(x)} \right) q(x) \right]f(x)dx\\	
	&= \int\limits_0^{\unt} \left[  -u^m( \unt) + \int\limits_{0}^{\unt} q(y)dy + \left( x - \dfrac{F(\unt)-F(x)}{f(x)} \right) q(x) \right]f(x)dx \\
	&= \int\limits_0^{\unt} F(\unt) q(y)\dfrac{f(y)}{f(y)}dy + \int\limits_0^{\unt}\left[  -u^m( \unt)  + \left( x - \dfrac{F(\unt)-F(x)}{f(x)} \right) q(x) \right]f(x)dx\\
	&=\int\limits_0^{\unt}\left[  -u^m( \unt)  + \left( x + \dfrac{F(x)}{f(x)} \right) q(x) \right]f(x)dx \\
	&= -F(\unt)u^m(\unt) +  \int\limits_0^{\unt} \left[   \left( x + \dfrac{F(x)}{f(x)} \right) q(x) \right]f(x)dx
\end{align}

In line 4, we change the order of integration; in line 5, we isolate the inner integral by extracting the allocations out; in line 6, we replace the value of the inner integral; in line 7, we merge the sum back; in line 8, we replace the value of the utility of the lowest type; in line 9, we integrate out the information rent for these types; in line 10, we cancel the new double integral with the $-\unt q(x)$ as that integral turns out to be just the integral of $\unt q(x)$ by changing the order of integration as above. We follow the similar steps for the transfers from $[\ovt,1]$.

\begin{align*}
	&\int\limits_{\ovt }^1 t(\theta)f(\theta)d\theta\\
	&=  \int\limits_{\ovt }^1 \left[ \theta q(\theta) -u^m(\ovt ) - \int\limits_{\ovt }^{\theta}q(x)dx \right]f(\theta)d\theta\\
	&=  \int\limits_{\ovt }^1 \left[ \theta q(\theta) -u^m(\ovt )\right]f(\theta)d\theta - \int\limits_{\ovt }^1 \int\limits_{\ovt }^{\theta}q(x) f(\theta)dx d\theta\\
	&=  \int\limits_{\ovt }^1 \left[ \theta q(\theta) -u^m(\ovt )\right]f(\theta)d\theta - \int\limits_{\ovt }^1 \int\limits_{x}^{1}q(x) f(\theta) d\theta dx\\
	&=  \int\limits_{\ovt }^1 \left[ \theta q(\theta) -u^m(\ovt )\right]f(\theta)d\theta - \int\limits_{\ovt }^1q(x) \int\limits_{x}^{1}f(\theta) d\theta dx\\
	&=  \int\limits_{\ovt }^1 \left[ \theta q(\theta) -u^m(\ovt )\right]f(\theta)d\theta - \int\limits_{\ovt }^1 q(x) (1-F(x)) dx\\
	&=  \int\limits_{\ovt }^1 \left[ -u^m(\ovt ) + \left( x - \dfrac{1-F(x)}{f(x)} \right) q(x) \right]f(x)dx\\	
	&= -(1-F(\ovt ))u^m(\ovt ) + \int\limits_{\ovt }^1 \left[  \left(x- \frac{1-F(x)}{f(x)} \right) q(x) \right]f(x)dx
\end{align*}

\subsection{Individual Rationality and the Allocations}\label{app:lemma:ir}

\begin{proof}[Proof of Lemma \ref{simplealloc}]
Notice that just below $\unt$, the utility from the mechanism should have a left derivative below $-\frac{p}{2}$:

\begin{align*}
	u^m(\theta) = u^m(0) + \int_{0}^{\theta} q(x)dx &\geq u^d(\theta) \iff \\
	u^d(\unt) -\int_{0}^{\unt} q(x)dx + \int_{0}^{\theta} q(x)dx &\geq u^d(\theta) \iff \\
	u^d(\unt) -\int_{\theta}^{\unt} q(x)dx &\geq u^d(\theta) \iff \\
	p \left[ \dfrac{\int\limits_{\unt}^{\ovt} x f(x) dx }{2\left( F(\ovt) - F(\unt) \right)} -\dfrac{\unt}{2} \right] -\int_{\theta}^{\unt} q(x)dx &\geq p \left[ \dfrac{\int\limits_{\unt}^{\ovt} x f(x) dx }{2\left( F(\ovt) - F(\unt) \right)} -\dfrac{\theta}{2} \right] \iff \\
	p\dfrac{\theta-\unt}{2}  &\geq \int_{\theta}^{\unt} q(x)dx \iff \\	
	p\dfrac{\theta-\unt}{2}  &\geq u^m(\unt) - u^m(\theta).
\end{align*}

Since only possible allocations are $-1$, $0$, and $1$, this means we must have $q(\unt)=-1$. Moreover, due to monotonicity of the allocation, for each $\theta\in[0,\unt]$, this implies $q(\theta)=-1$. Following the same steps around $\ovt$ also shows that $q(\theta) = 1$ for each $\theta\in [\ovt,1]$. We note the observations we have made here in the following lemma.	
\end{proof}

\subsection{Profit from the simple equilibrium}\label{genprofit}

Using the Lemma \ref{simplealloc},

\begin{equation*}
	\int\limits_0^{\unt}C(x)q(x)f(x)dx +\int\limits^1_{\ovt}V(y)q(y)f(y)dy = \int\limits^1_{\ovt }V(y)f(y)dy -\int\limits_0^{\unt}C(x)f(x)dx
\end{equation*}

From the Appendix \ref{virtualsurplus}, we have

\begin{equation*}
	-\int\limits_0^{\unt}C(x)f(x)dx +\int\limits^1_{\ovt }V(y)f(y)dy = - \unt F(\unt) -\ovt F(\ovt ) +\ovt
\end{equation*}

Moreover, under a CRS matching function $M$ for the decentralized market, $p$ is independent of the segmentation of the market. Then, using the binding IR constraints, the compensations paid to the agents will be given by:

\begin{align*}
	& F(\unt) u^d(\unt) + (1-F(\ovt )) u^d(\ovt ) \\
	&= p\dfrac{F(\unt)}{2} \left[\dfrac{\int\limits_{\unt}^{\ovt} \theta f(\theta) d\theta }{ F(\ovt) - F(\unt)} -\unt \right] + p\dfrac{1-F(\ovt )}{2} \left[ \ovt  - \dfrac{  \int\limits_{\unt}^{\ovt } \theta f(\theta) d\theta }{ F(\ovt ) - F(\unt)} \right] \\
	&= p \dfrac{1}{2} \left( F(\unt) \left[ \dfrac{\int\limits_{\unt}^{\ovt } \theta f(\theta) d\theta }{ F(\ovt) - F(\unt)} -\unt \right] + (1-F(\ovt )) \left[ \ovt  -  \dfrac{  \int\limits_{\unt}^{\ovt } \theta f(\theta) d\theta}{ F(\ovt) - F(\unt)} \right] \right) \\
	&=  \dfrac{p}{2} \left( F(\unt)  E[\theta|\unt\leq \theta \leq \ovt ] -\unt F(\unt)  + \ovt (1-F(\ovt))  -  (1-F(\ovt)) E[\theta|\unt\leq \theta \leq \ovt ]  \right) \\
	&=  \dfrac{p}{2} \left( \left( F(\unt) -  (1-F(\ovt)) \right)  E[\theta|\unt\leq \theta \leq \ovt ] -\unt F(\unt)  + \ovt (1-F(\ovt))  \right)
\end{align*}

Thus, 

\begin{align*}
	& \Pi \\
	&=- \dfrac{p}{2} \left( \left( F(\unt) -  (1-F(\ovt )) \right)  E[\theta|\unt\leq \theta \leq \ovt ] -\unt F(\unt)  + \ovt (1-F(\ovt ))  \right)  +  \left[- \unt F(\unt) -\ovt F(\ovt ) +\ovt \right] \\
	&= \dfrac{1}{2} \left( (2-p) \left[ -\unt F(\unt)  + \ovt (1-F(\ovt )) \right]  - p\left[ F(\unt) -  (1-F(\ovt )) \right] E[\theta|\unt\leq \theta \leq \ovt ] \right)
\end{align*}

Notice that

\begin{align*}
	\mathbb{E}[\theta|\theta\in [\unt, \ovt]]&= \dfrac{\int_{\unt}^{\ovt} x f(x)dx}{F(\ovt)-F(\unt)}\\
	\dfrac{\partial\mathbb{E}[\theta|\theta\in [\unt, \ovt]]}{\partial \unt} &= \dfrac{-\unt f(\unt) [F(\ovt)-F(\unt)] + f(\unt)\left[ \int_{\unt}^{\ovt}xf(x)dx \right] }{[F(\ovt)-F(\unt)]^2}\\
	&=f(\unt)\dfrac{\int_{\unt}^{\ovt}xf(x)dx - \unt [F(\ovt)-F(\unt)]}{[F(\ovt)-F(\unt)]^2}\\
	&=f(\unt)\dfrac{\mathbb{E}[\theta|\theta\in [\unt, \ovt]] - \unt }{[F(\ovt)-F(\unt)]}>0.
\end{align*}

Then, using this,

\begin{align*}
	&2 \dfrac{\partial\Pi}{\partial \unt} \\
	&= - (2-p)F(\unt) - (2-p)\unt f(\unt) -p [f(\unt)] E[\theta|\unt\leq \theta \leq \ovt ] -p \left[ F(\unt) -  (1-F(\ovt )) \right] \dfrac{\partial\mathbb{E}[\theta|\theta\in [\unt, \ovt]]}{\partial \unt}
\end{align*}

For a feasible mechanism, we need $F(\unt) -  (1-F(\ovt ))\geq 0$. Thus, each term above is negative, and $\Pi$ is decreasing in $\unt$. Given that $\Pi$ is decreasing in $\unt$, $F(\unt) -  (1-F(\ovt ))\geq 0$ will bind in any equilibrium. Therefore,

\begin{align*}
	\Pi &= \dfrac{1}{2}  (2-p) \left[ -\unt F(\unt)  + \ovt (1-F(\ovt )) \right] \\
	&= \dfrac{2-p}{2}  \left[ -\unt F(\unt)  + \ovt (F(\unt)) \right]\\
	&= \dfrac{2-p}{2}  \left[ F(\unt) [\ovt -\unt]  \right]
\end{align*}

Notice that for interior values with $\ovt>\unt$, $\Pi>0$. Thus, positive profit is feasible and will be achieved in the equilibrium.

\subsection{Existence of Simple Equilibrium}\label{app:simpleexistence}

\begin{proof}[Proof of Theorem \ref{2cutofftheorem}]
	
	In the Appendix \ref{genprofit}, I show that the profit from any pair of thresholds, $\unt$ and $\ovt$, can be written as follows.
	\begin{equation*}
		\Pi= \dfrac{1}{2} \left( [2-p] \left[ -\unt F(\unt)  + \ovt (1-F(\ovt )) \right]  - p\left[ F(\unt) -  (1-F(\ovt )) \right] E[\theta|\unt\leq \theta \leq \ovt ] \right)
	\end{equation*}
	Our constraints are $0\leq \unt \leq \ovt  \leq 1 $ and $F(\unt) \geq 1-F(\ovt )$. First, we cannot have $\ovt  < m(F)$  where $m(F)$ is the median of $F$ since that would require $ 0.5 >  F(\ovt ) \geq  F(\unt) \geq 1-F(\ovt )\geq 0.5$. 
	Second, Appendix \ref{genprofit} also shows that $\Pi$ is strictly decreasing in $\unt$. Then, $F(\unt) \geq 1-F(\ovt )$ should bind. Third, $\unt\leq m(F) \leq \ovt $ as a result of previous two observations. Thus, for a strictly increasing $F$, this is essentially a single parameter problem with a continuous objective and a compact domain. Hence, it has a solution by Weierstrass Theorem. Moreover, the solution is interior in the sense that the constraints $\unt\leq m(F) \leq \ovt $ do not bind. If they did, then the profit would be 0 while it is possible to achieve a positive profit when they do not bind. (Appendix \ref{genprofit} shows this in more detail as well.) When the feasibility condition binds, the expectation term disappears from the profit. Then, the profit in this equilibrium is equal to the profit when there was no search market, times a constant, $\frac{2-p}{2}$. Thus, the solution must still have $\mathcal{C}(\unt) = \mathcal{V}(\ovt)$, as shown in the Theorem \ref{thm:monagora}.
	
	The mechanism is constructed so that no agent has any unilateral profitable deviation. However, there are no profitable bilateral deviations either. First, if a pair agents from the centralized marketplace deviate to the decentralized market, they still have to go through the random search process. Thus, they still expect the same payoff they would get if they deviated unilaterally. Then, there is no profitable deviation in that direction. On the other hand, deviations to the mechanism also cannot be profitable. Either (i) agents report their valuations truthfully and the marketplace shuts down and agents get 0 utility or (ii) they lie, and get a lower payoff than the mechanism promises them due to incentive compatibility. In the latter case, their payoffs cannot be greater than the payoffs they get from the decentralized market; they were offered lower utilities to begin with and by lying about their valuations, they obtain even lower utilities.
	
\end{proof}

\subsubsection{The Mechanism That Induces The Simple Equilibrium}\label{simplemechproof}

\begin{proposition}\label{simplemech}
	In the simple equilibrium, the mechanism the designer offers has the following allocation and transfer rules:
	
	\begin{align*}
		q(\theta)&=\begin{cases}
			-1 &\text{ if }  \theta\leq \dfrac{p \mathbb{E}[x|x\in[\unt,\ovt]] +(2-p)\unt }{2}\\
			0 & \text{ if }  \dfrac{p \mathbb{E}[x|x\in[\unt,\ovt]] +(2-p)\unt }{2}\leq \theta\leq \dfrac{p \mathbb{E}[x|x\in[\unt,\ovt]] +(2-p)\ovt }{2}\\
			1 & \text{ if } \theta\geq \dfrac{p \mathbb{E}[x|x\in[\unt,\ovt]] +(2-p)\ovt }{2}\\
		\end{cases}\\
		t(\theta)&=\begin{cases}
			-\dfrac{p \mathbb{E}[x|x\in[\unt,\ovt]] +(2-p)\unt }{2} & \text{ if }  \theta\leq \dfrac{p \mathbb{E}[x|x\in[\unt,\ovt]] +(2-p)\unt }{2}\\
			0 & \text{ if } \dfrac{p \mathbb{E}[x|x\in[\unt,\ovt]] +(2-p)\unt }{2}\leq \theta\leq \dfrac{p \mathbb{E}[x|x\in[\unt,\ovt]] +(2-p)\ovt }{2}\\
			\dfrac{p \mathbb{E}[x|x\in[\unt,\ovt]] +(2-p)\ovt }{2} & \text{ if } \theta\geq \dfrac{p \mathbb{E}[x|x\in[\unt,\ovt]] +(2-p)\ovt }{2}\\
		\end{cases}
	\end{align*}

\end{proposition}

\begin{proof}[Proof of Proposition \ref{simplemech}]
	
	First, we compute the transfers for the agents who join the mechanism using the binding IR constraints for $\unt$ and $\ovt$.
	
	In the simple equilibrium we construct, for an agent with $\theta\in [0,\unt]$,
	
	\begin{align*}
		t(\theta)&=\theta q(\theta) - u^m(0) -\int_{0}^{\theta}q(x)dx\\
		&= \theta (-1) - u^m(0) - \theta (-1) \\
		&= -u^m(0) = -u^d(\unt) + \int_{0}^{\unt} q(x)dx = -u^d(\unt) + \int_{0}^{\unt} (-1)dx\\
		&=-\dfrac{p}{2} \left[ E[x|x\in [\unt,\ovt]] - \unt \right] - \unt\\
		&=-\dfrac{p}{2}  E[x|x\in [\unt,\ovt]] - \unt\dfrac{2-p}{2}\\
	\end{align*}
	
	Similarly, we can compute the transfer of agents with $\theta\in [\ovt,1]$,

	\begin{align*}
		t(\theta)&=\theta q(\theta) -u^m(\ovt) - \int_{\ovt}^{\theta}q(x)dx\\
		&= \theta (1) - u^m(\ovt) - 1(\theta - \ovt ) \\
		&= \ovt -u^m(\ovt) = \ovt -\dfrac{p}{2} \left[ \ovt - E[x|x\in [\unt,\ovt]] \right] \\
		&=  \dfrac{p}{2}  E[x|x\in [\unt,\ovt]] + \ovt\dfrac{2-p}{2}\\
	\end{align*}
	
	Knowing these, the designer can offer $-1$ allocation to all agents whose valuations are below $\frac{p \mathbb{E}[x|x\in[\unt,\ovt]] +(2-p)\unt }{2}$. Similarly, for agents with valuations above $\frac{p \mathbb{E}[x|x\in[\unt,\ovt]] +(2-p)\ovt }{2}$, $1$ unit of allocation can be offered. In between, they are not offered any trade. This allocation is clearly increasing. Moreover, accompanied by the transfers $-\frac{p \mathbb{E}[x|x\in[\unt,\ovt]] +(2-p)\unt }{2}$ for agents with negative allocations and $\frac{p \mathbb{E}[x|x\in[\unt,\ovt]] +(2-p)\ovt }{2}$ for agents with positive allocations, agents with values in $(\unt,\ovt)$ would strictly prefer the search market. To see this, note that these agents have utilities with slopes $-1$ until $u^m$ hits 0, then it is constant at 0 and then it has the slope $1$, after $\frac{p \mathbb{E}[x|x\in[\unt,\ovt]] +(2-p)\ovt }{2}$. Moreover, $u^m$ and $u^d$ are equal at $\unt$ and $\ovt$. Since the slope of $u^d$ is bounded between $-\dfrac{p}{2}$ and $\dfrac{p}{2}$, and $u^d$ is positive, $u^m$ and $u^d$ cannot cross each other at any point other than $\unt$ and $\ovt$. Thus, $u^d$ remains below $u^m$ for values in $(\unt,\ovt)$.

\end{proof}

\subsection{Simple Economics of Optimal Marketplaces}\label{app:simpecon}

Here is how Figure \ref{fig:simplecomplex} works. The aggregate compensations for buyers is equal to $ \frac{p}{2} q^*\times (\unt -\mathbb{E}_d) $ where $\mathbb{E}_d$ the expected valuation in the decentralized market, i.e., $\mathbb{E}_d=\mathbb{E}[\theta|\theta\in [\unt,\ovt]]$. Similarly, the aggregate compensations for sellers is equal to $ \frac{p}{2} q^*\times (\mathbb{E}_d - \unt) $. Since $\mathbb{E}_d$ is a point between $\unt$ and $\ovt$, this means total compensations will be equal to $\dfrac{p}{2}$ times the area of the rectangle between $\unt$ and $\ovt$ on the vertical axis and between $0$ and $q^*_C$; $(\ovt-\unt)\times q^*_C$.

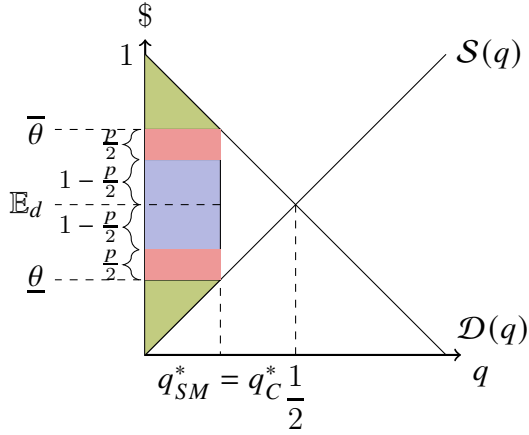
\begin{figure}[H]
	\centering
	\begin{tikzpicture}[scale=0.4]
		
		\draw[thick,<->] (0,10.5) node[above]{$\$$}--(0,0)--(10.5,0) node[below right]{$q$};
		
		\draw(0,0)--(10,10) node[right]{$\mathcal{S}(q)$};

		\draw(0,10)--(10,0) node[above right]{$\mathcal{D}(q)$};
		
		\draw[dashed](5,5)--(5,0) node [below] {$\dfrac{1}{2}$};

		\node[] at (0,10) [left] {$1$};

		\draw[dashed](2.5,7.5)--(2.5,0) node [below] {$q^*_{SM}=q^*_{C}$};
		
		\draw[dashed](2.5,7.5)--(-3,7.5) node [left] {$\ovt$};
		
		\draw[dashed](2.5,2.5)--(-3,2.5) node [left] {$\unt$};

		\draw[fill=solarizedViolet!50] (0,2.5) -- (2.5,2.5) -- (2.5,7.5) --(0,7.5);	
		
		\draw[fill=solarizedGreen!50] (0,10) -- (2.5,7.5) --(0,7.5);
		
		\draw[fill=solarizedGreen!50] (0,0) -- (2.5,2.5) --(0,2.5);
		
		\draw [name path=X] (0,7.5) to  (2.5,7.5) node[above left]{};
		
		\draw [name path=A] (0,6.5) to (2.5,6.5) node[above left]{};
		
		\tikzfillbetween[of=X and A]{solarizedRed!50};
		
		\draw [name path=Y] (0,2.5) to  (2.5,2.5) node[above left]{};
		
		\draw [name path=B] (0,3.5) to (2.5,3.5) node[above left]{};
		
		\tikzfillbetween[of=Y and B]{solarizedRed!50};
		
		\draw [decorate,decoration={brace,amplitude=6pt, mirror},xshift=-0.2cm,yshift=0pt]
		(0,7.5) -- (0,6.5) node [midway,left,xshift=-.1cm] {\footnotesize$\frac{p}{2}$}; 
		\draw [decorate,decoration={brace,amplitude=6pt},xshift=-0.2cm,yshift=0pt]
		(0,2.5) -- (0,3.5) node [midway,left,xshift=.-0.1cm] {\footnotesize$\frac{p}{2}$};

		\draw [decorate,decoration={brace,amplitude=6pt, mirror},xshift=-0.2cm,yshift=0pt]
		(0,6.5) -- (0,5) node [midway,left,xshift=-.1cm] {\footnotesize$1-\frac{p}{2}$}; 
		\draw [decorate,decoration={brace,amplitude=6pt, mirror},xshift=-0.2cm,yshift=0pt]
		(0,5) -- (0,3.5) node [midway,left,xshift=.-0.1cm] {\footnotesize$1-\frac{p}{2}$};

		\draw[dashed](2.5,5)--(-3,5) node [left] {$\mathbb{E}_d$};

	\end{tikzpicture}
\caption{The ratio of the compensations only depend on $p$ and not the distribution.}		
\label{fig:simplecomplex}
\end{figure}

\subsection{Efficiency of Coexistence}\label{app:efficiency}
\subsubsection{Under Uniform Distribution}\label{efficiencyproofuni}
\begin{proof}[Proof of Proposition \ref{efficiencypropuni}]
	
	Suppose everyone is in the search market. Then, the total welfare can be computed as follows:
	
	\begin{align*}
		\mathbb{E}[u^d(\theta)]&=p \int_0^1 [\theta \theta - (1-\theta)\theta]d\theta\\
		&=p \int_0^1 [2\theta^2 -\theta] d\theta\\
		&=p \left[ \dfrac{2\theta^3}{3} -\dfrac{\theta^2}{2} \right]_0^1 = \dfrac{p}{6}.
	\end{align*}
	
	Next, we compute the welfare created by the marketplace alone in the coexistence equilibrium. The welfare marketplace generates will be more than enough to exceed the total welfare of the pure search market, so we do not need to compute the welfare created by the search market in the coexistence.
	
	The profit function from the simple equilibrium under the uniform distribution is a constant times $\unt (\ovt-\unt)=\unt(1-2\unt)$ using the fact that the feasibility binds so that $\unt=1-\ovt$. This is maximized at $\unt=\frac{1}{4}$ and $\ovt=\frac{3}{4}$. Then, the welfare generated by the marketplace is
	
	\begin{equation*}
		\int_{0.75}^{1} \theta_b d\theta_b - \int_{0}^{0.25} \theta_s d\theta_s  =  \left[ \dfrac{\theta^2}{2} \right]_{0.75}^{1} - \left[ \dfrac{\theta^2}{2} \right]_{0}^{0.25} = \dfrac{3}{16}.
	\end{equation*}

The total welfare from the search market is $\dfrac{p}{6}\leq \dfrac{1}{6}$ for any matching function since the probability a meeting will be less than or equal to 1. Moreover, $\dfrac{3}{16}>\dfrac{1}{6}$. Thus, for any matching function, the coexistence equilibrium creates a welfare higher than the pure search market.
\end{proof}

\subsubsection{Under General Distribution}\label{efficiencyproof}
\begin{proof}[Proof of Proposition \ref{efficiencyprop}]
	
	For the pure search market, the total welfare created is given by
	
	\begin{align*}
		&\int_{0}^{1} \left[ pF(\theta)   \theta- p(1-F(\theta)) \theta  \right] f(\theta) d\theta \\
		=& p \int_{0}^{1} \theta \left[   2F(\theta) -1  \right] f(\theta) d\theta
	\end{align*}
	
	In the first line above, $pF(\theta)$ is the probability that the agent with the value $\theta$ meets with an agent with a value less than $\theta$, so she gets $\theta$ in the trade and $p(1-F(\theta))$ is the probability that she meets with an agent whose value is higher so she loses $\theta$. As with the uniform case, we ignore the transfers in the utilities as the transfers will cancel in the search market.
	
	In the coexistence equilibrium, the total welfare created in the search market is
	
	\begin{align*}
		&\int_{0}^{1} \left[ p\left[  \dfrac{F(\theta) - F(\unt)}{F(\ovt) - F(\unt)}  \right]  \theta - p\left[  \dfrac{F(\ovt)-F(\theta)}{F(\ovt) - F(\unt)}  \right] \theta \right] f(\theta) d\theta \\
		=& p \int_{\unt}^{\ovt} \theta \left[  \dfrac{2F(\theta) - F(\unt)-F(\ovt)}{F(\ovt) - F(\unt)}  \right] f(\theta) d\theta\\
		=& p \int_{\unt}^{\ovt} \theta \left[  \dfrac{2F(\theta) - 1}{F(\ovt) - F(\unt)}  \right] f(\theta) d\theta\\
	\end{align*}

	The welfare generated by the marketplace in the coexistence is
	
	\begin{align*}
	\int_{\ovt}^{1} x f(x)dx - \int_{0}^{\unt} x f(x)dx 
	\end{align*}

	\begin{align*}
		& p \int_{\unt}^{\ovt} \theta \left[  \dfrac{2F(\theta) - 1}{F(\ovt) - F(\unt)}  \right] f(\theta) d\theta +  \int_{\ovt}^{1} x f(x)dx - \int_{0}^{\unt} x f(x)dx  \geq p \int_{0}^{1} \theta \left[   2F(\theta) -1  \right] f(\theta) d\theta\\	
		\iff& p \int_{\unt}^{\ovt} \theta \left[  \dfrac{2F(\theta)  }{F(\ovt) - F(\unt)}  \right] f(\theta) d\theta + \int_{\ovt}^{1} x f(x)dx +  p \int_{0}^{1} \theta f(\theta) d\theta \\
		& \geq p \int_{\unt}^{\ovt} \left[  \dfrac{ \theta f(\theta) d\theta }{F(\ovt) - F(\unt)}  \right]  + p \int_{0}^{1} \theta \left[   2F(\theta)  \right] f(\theta) d\theta + \int_{0}^{\unt} x f(x)dx \\
		\iff& p \int_{\unt}^{\ovt} \theta \left[  \dfrac{2F(\theta)  }{F(\ovt) - F(\unt)}  \right] f(\theta) d\theta  +  (1+p) \int_{\ovt}^{1} \theta f(\theta) d\theta \\
		& \geq p \left[ \dfrac{2F(\unt)}{1-2F(\unt)} \right] \int_{\unt}^{\ovt}  \theta f(\theta) d\theta   + 2p \int_{0}^{1} \theta \left[   F(\theta)  \right] f(\theta) d\theta + (1-p) \int_{0}^{\unt} x f(x)dx \\		
	\end{align*}

In the first line above, either we have $$\int_{\unt}^{\ovt} \theta \left[  \dfrac{2F(\theta) - 1}{F(\ovt) - F(\unt)}  \right] f(\theta) d\theta \geq \int_{0}^{1} \theta \left[   2F(\theta) -1  \right] f(\theta) d\theta $$

in which case the inequality is satisfied for any $p$, since $\int_{\ovt}^{1} x f(x)dx - \int_{0}^{\unt} x f(x)dx \geq 0$ when $F(\unt)=1-F(\ovt)$ or

$$\int_{\unt}^{\ovt} \theta \left[  \dfrac{2F(\theta) - 1}{F(\ovt) - F(\unt)}  \right] f(\theta) d\theta <  \int_{0}^{1} \theta \left[   2F(\theta) -1  \right] f(\theta) d\theta. $$

Thus, to show that for any $p$, the coexistence is more efficient, it is enough to show it with $p=1$. 
	
\end{proof}

\subsection{Slope of Utilities from Search Market}\label{slope}

\begin{lemma}\label{slopelemma}
	Under any equilibrium, $-\dfrac{p}{2} \leq \dfrac{\partial u^d(\theta)}{\partial \theta} \leq \dfrac{p}{2}$ for each agent.
\end{lemma}

\begin{proof}
	
Suppose $\Theta^d$ is the set of agents who join the decentralized market, $\mu(\Theta^d)=\mathbb{P}[x\in \Theta^d ]$ their measure, and let $\theta\in Cov(\Theta^d)$. Then,

\begin{align*}
	u^d(\theta) &= \dfrac{p}{2}  \mathbb{P}[x>\theta| x\in \Theta^d ] \left[ \mathbb{E}[x| x>\theta, x\in \Theta^d ] - \theta \right]\\
	& +\dfrac{p}{2}  \mathbb{P}[x<\theta| x\in \Theta^d ] \left[ \theta- \mathbb{E}[x| x<\theta, x\in \Theta^d ] \right]\\
	&= \dfrac{p}{2}  \dfrac{\mathbb{P}[x>\theta, x\in \Theta^d ]}{\mu(\Theta^d)} \left[ \mathbb{E}[x| x>\theta, x\in \Theta^d ] - \theta \right]\\
	& +\dfrac{p}{2}  \dfrac{\mathbb{P}[x<\theta, x\in \Theta^d ]}{\mu(\Theta^d)} \left[ \theta- \mathbb{E}[x| x<\theta, x\in \Theta^d ] \right]\\
	&=\dfrac{p}{2 \mu(\Theta^d)} \left[ \int\limits_{\{x\in\Theta^d: x>\theta\}} x f(x)dx - \theta \mathbb{P}[x>\theta, x\in \Theta^d ] \right]\\
	&+\dfrac{p}{2\mu(\Theta^d)} \left[ \theta \mathbb{P}[x<\theta, x\in \Theta^d ] - \int\limits_{\{x\in\Theta^d: x<\theta\}} x f(x)dx \right]\\
	\dfrac{\partial u^d(\theta)}{\partial \theta} &= \dfrac{p}{2\mu(\Theta^d)} \left[ \mathbb{P}[x<\theta, x\in \Theta^d ] - \mathbb{P}[x>\theta, x\in \Theta^d ] \right]. 
\end{align*}

If $\theta\leq\theta'$  for each $\theta' \in \Theta^d$, then

\begin{align*}
	u^d(\theta) &= \dfrac{p}{2}  \mathbb{P}[x>\theta| x\in \Theta^d ] \left[ \mathbb{E}[x| x\in \Theta^d ] - \theta \right]\\
	&=\dfrac{p}{2} \left[ \int\limits_{\{x\in\Theta^d\}} \dfrac{x f(x)dx}{\mu(\Theta^d)} - \theta ] \right]\\
	\dfrac{\partial u^d(\theta)}{\partial \theta} &= -\dfrac{p}{2}.
\end{align*}

Finally, if $\theta\geq \theta'$  for each $\theta' \in \Theta^d$, then

\begin{align*}
	u^d(\theta) &= \dfrac{p}{2}  \mathbb{P}[x<\theta| x\in \Theta^d ] \left[ \theta - \mathbb{E}[x| x\in \Theta^d ]  \right]\\
	&=\dfrac{p}{2} \left[   \theta - \int\limits_{\{x\in\Theta^d\}} \dfrac{x f(x)dx}{\mu(\Theta^d)} ] \right]\\
	\dfrac{\partial u^d(\theta)}{\partial \theta} &= \dfrac{p}{2}.
\end{align*}
\end{proof}

\subsection{Unrestricted Mechanisms}\label{unrestricted}

\begin{proof}[Proof of Theorem \ref{unrestrictedthm}]
	We start with a simple observation: For the mechanism to make a positive profit, there has to be both agents who buy and sell at the marketplace. This means for a positive measure of agents, $u^m(\theta)\geq u^d(\theta)$ on both regions with $\dfrac{\partial u^m(\theta)}{\partial \theta}=1$ and $\dfrac{\partial u^m(\theta)}{\partial \theta}=-1$.
	
	We have shown in Appendix \ref{slope} that for an arbitrary segmentation of agents, the expected utility from search has a slope between $-0.5$ and $0.5$.
	
	Notice that if for an agent $\dfrac{\partial u^m(\theta)}{\partial \theta}=1$, then for each $\theta'>\theta$, $q(\theta')=1$ by the envelope condition and monotonicity of the allocation for an IC mechanism. Similarly, if $\dfrac{\partial u^m(\theta)}{\partial \theta}=-1$, then for each $\theta'<\theta$, $q(\theta')=-1$.
	
	Given this, if for $\theta$, $u^m(\theta)\geq u^d(\theta)$ and $\dfrac{\partial u^m(\theta)}{\partial \theta}=1$, then for each $\theta'>\theta$, $u^m(\theta')\geq u^d(\theta')$ and similarly for the sellers. Thus, let $\unt$ be the highest value such that $u^m(\theta)\geq u^d(\theta)$ and $q(\theta)=-1$ in the equilibrium. Similarly, let $\ovt$ be the lowest value such that $u^m(\theta)\geq u^d(\theta)$ and $q(\theta)=1$.
	
	There must be at least one type such that $u^m(\theta)=u^d(\theta)$. If not, either the utilities from search are above the utilities from the mechanism everywhere so that no one comes to the marketplace and the profit of the marketplace is zero or the utilities from the mechanism is strictly higher everywhere so everyone is in the mechanism and the mechanism can reduce the utilities until some IR constraint binds to strictly increase the profit.
	
	Next, we argue that it cannot be the case that $u^m$ and $u^d$ are only tangent at $\unt$ and $\ovt$, the utilities have to cross each other at these cutoffs: If they are only tangent but does not cross each other, then $\unt$ or $\ovt$ would have to be the point of a kink on $u^m$. Then, there are two cases: Either (i) every agent joins the mechanism or (ii) only one of the cutoffs is at a kink, and an interval of agents near the other cutoff join the search market. (i) In the former case, the outside option is zero for every agent so the mechanism does not compensate them. But then, they bilaterally deviate for a positive payoff. (ii) In the latter case, suppose there is a kink at $\ovt$. Then, if the search market is active, $u^d$ has to be increasing at $\ovt$ since everyone above it is in the mechanism so that an agent with the value $\ovt$ can only be a buyer in the search market. For there to be agents in the search market, $u^d$ should cross $u^m$ at a point $\theta$ such that $u^d$ is decreasing since it cannot cross $u^m$ on the part it is constant or has a slope of $1$ below $\ovt$. But if there is such a point, then $u^d$ would be increasing at $\theta$, since all agents in the search market will be below it as well, which shows this case is impossible as well. Thus, $u^m$ and $u^d$ cannot be tangent at $\unt $ and $\ovt$, they have to cross each other at these points.
	
	Then, due to the shape of the feasible utility functions ($u^m$ can have slopes $-1$, $0$, and $1$ in this order and $u^d$ is first decreasing and then increasing -with a potentially constant $0$ slope in the middle- with a slope that remains between $-\frac{1}{2}$ and $\frac{1}{2}$), either all agents with values in $[\unt, \ovt]$ join the search market or the flat part of the $u^m$ crosses $u^d$ twice again, in which case agents with values in $[\unt, a]$ and $[b,\ovt]$ join the search market for some $\unt<a<b<\ovt$ and agents with values in $a,b$ join the mechanism as well. Moreover, in the latter case, we need $F(a)-F(\unt)=F(\ovt)-F(b)$; otherwise the $u^d$ would be either strictly decreasing or strictly increasing for agents with values in $[a,b]$, in which case $u^d$ and $u^m$ would not cross at both $a$ and $b$, as this case requires. When $F(a)-F(\unt)=F(\ovt)-F(b)$, $u^d$ would be flat, as it can be seen from the slope we computed above. So, we can write the profit as follows where the case with $a=b$ corresponds to the situation where the flat part of $u^m$ does not cross $u^d$.
	
	\begin{align*}
		\Pi & = -F(\unt) u^d(\unt) - (1-F(\ovt)) u^d(\ovt) - (F(b)-F(a)) u^d(a) - \int_0^{\unt} C(x)f(x)dx +  \int_{\ovt}^1 V(x)f(x)dx
	\end{align*}
	
	Moreover, the constraints are $0\leq \unt \leq a\leq b\leq \ovt \leq 1$, $F(\unt)\geq 1-F(\ovt)$ and $F(a)-F(\unt)=F(\ovt)-F(b)$. 
	
	Let  $ \Theta^d = [\unt, a] \cup [b,\ovt]$. Then,
	
	\begin{align*}
		\Pi  = & -F(\unt) u^d(\unt) - (1-F(\ovt)) u^d(\ovt) - (F(b)-F(a)) u^d(a) - \int_0^{\unt} C(x)f(x)dx +  \int_{\ovt}^1 V(x)f(x)dx\\
		= & - \dfrac{pF(\unt)}{2} \left[ \int\limits_{\{x\in\Theta^d\}}   \dfrac{x f(x)dx}{\left[ F(\ovt)-F(b)+F(a)-F(\unt) \right]} - \unt  \right]\\
		& -  \dfrac{p(1-F(\ovt))}{2} \left[  \ovt - \int\limits_{\{x\in\Theta^d\}} \dfrac{x f(x)dx}{\left[ F(\ovt)-F(b)+F(a)-F(\unt) \right]} \right]\\
		& - \dfrac{p(F(b)-F(a))}{2\left[ F(\ovt)-F(b)+F(a)-F(\unt) \right]} \left[  \int\limits_{b}^{\ovt} x f(x)dx - \int\limits_{\unt}^{a} x f(x)dx \right] - \left[ - \underline{\theta}F(\underline{\theta}) -\overline{\theta}F(\overline{\theta}) +\overline{\theta} \right]\\
		= & \dfrac{p\left[  (1-F(\ovt) - F(\unt))\right] }{2} \left[  \int\limits_{\{x\in\Theta^d\}}  \dfrac{x f(x)dx}{\left[ F(\ovt)-F(b)+F(a)-F(\unt) \right]} \right]  + \dfrac{p(F(\unt )\unt  - (1-F(\ovt))\ovt)}{4}  \\
		& - \dfrac{p(F(b)-F(a))}{2\left[ F(\ovt)-F(b)+F(a)-F(\unt) \right]} \left[  \int\limits_{b}^{\ovt} x f(x)dx - \int\limits_{\unt}^{a} x f(x)dx \right] + \dfrac{1}{2}\left[ - \underline{\theta}F(\underline{\theta}) -\overline{\theta}F(\overline{\theta}) +\overline{\theta} \right]\\
		\dfrac{\partial \Pi}{\partial \unt} = & \dfrac{p\left[  (1-F(\ovt) - F(\unt))\right] }{2} \left[ \dfrac{-\unt f(\unt) \left[ F(\ovt)-F(b)+F(a)-F(\unt) \right] +f(\unt)  \int\limits_{\{x\in\Theta^d\}} xf(x)dx }{\left[ F(\ovt)-F(b)+F(a)-F(\unt) \right]^2} \right]\\
		& - \dfrac{pf(\unt)}{2} \left[  \int\limits_{\{x\in\Theta^d\}}  \dfrac{x f(x)dx}{\left[ F(\ovt)-F(b)+F(a)-F(\unt) \right]} \right] + \dfrac{p(\unt f(\unt) + F(\unt))}{2}\\
		& - \dfrac{p(F(b)-F(a))}{2\left[ F(\ovt)-F(b)+F(a)-F(\unt) \right]} \left[ \unt f(\unt) \right]\\
		& - \dfrac{pf(\unt)(F(b)-F(a))}{2\left[ F(\ovt)-F(b)+F(a)-F(\unt) \right]^2} \left[  \int\limits_{b}^{\ovt} x f(x)dx - \int\limits_{\unt}^{a} x f(x)dx \right] - (\unt f(\unt) + F(\unt))\\
		= &\dfrac{pf(\unt)\left[  (1-F(\ovt) - F(\unt))\right] }{2} \left[ \dfrac{   \mathbb{E}[x|x\in \Theta^d] - \unt }{\left[ F(\ovt)-F(b)+F(a)-F(\unt) \right]} \right]\\
		& - \dfrac{pf(\unt)}{2} \left[  \int\limits_{\{x\in\Theta^d\}}  \dfrac{x f(x)dx}{\left[ F(\ovt)-F(b)+F(a)-F(\unt) \right]} \right] - \dfrac{p(F(b)-F(a))}{2\left[ F(\ovt)-F(b)+F(a)-F(\unt) \right]} \left[ \unt f(\unt) \right]\\
		& - \dfrac{pf(\unt)(F(b)-F(a))}{2\left[ F(\ovt)-F(b)+F(a)-F(\unt) \right]^2} \left[  \int\limits_{b}^{\ovt} x f(x)dx - \int\limits_{\unt}^{a} x f(x)dx \right] - [\unt f(\unt) + F(\unt)]\left[ 1 - \dfrac{p}{2} \right]\\
	\end{align*}
	
	By noting that $(1-F(\ovt) - F(\unt))\leq 0$ by feasibility, each term in the above sum is negative and hence $\Pi$ is decreasing in $\unt$. We will use this to show that the feasibility binds.

	Next, we consider the Lagrangian problem to study the KKT conditions. Here, we will initially relax the problem by relaxing the equality constraint $F(a)-F(\unt)=F(\ovt)-F(b)$ to $F(a)+F(b)\geq F(\unt)+F(\ovt)$ but focus on solutions where it binds. From here, we are going to learn that the feasibility constraint must bind. We will use this to observe that $a=b$ should hold in the equilibrium, which will reduce the unrestricted equilibrium to a simple equilibrium.
	
	\begin{align*}
		\mathcal{L}(\unt, a, b, \ovt, \lambda) &= \Pi + \lambda_1 (F(\unt)+F(\ovt)-1) + \lambda_2 (1-\ovt) + \lambda_3 (\ovt -b) + \lambda_4 (b-a) + \lambda_5 (a-\unt) + \lambda_6 \unt\\
		+&\lambda_7(F(a)+F(b) - F(\unt)-F(\ovt))\\
		&\dfrac{\partial \mathcal{L}}{\partial \unt} = 	\dfrac{\partial \Pi}{\partial \unt} + \lambda_1 f(\unt) -\lambda_5 + \lambda_6 - \lambda_7f(\unt)=0\\
		&\dfrac{\partial \mathcal{L}}{\partial a} = 	\dfrac{\partial \Pi}{\partial a} -\lambda_4 + \lambda_5 + \lambda_7f(a)=0\\
		&\dfrac{\partial \mathcal{L}}{\partial b} = 	\dfrac{\partial \Pi}{\partial b} -\lambda_3 + \lambda_4 +  \lambda_7f(b)=0\\
		&\dfrac{\partial \mathcal{L}}{\partial \ovt} = 	\dfrac{\partial \Pi}{\partial \ovt} + \lambda_1 f(\ovt) -\lambda_3 + \lambda_4 -  \lambda_7f(\ovt)=0\\
		&\lambda_i\geq 0\\
		&\lambda_1 (F(\unt)+F(\ovt)-1)=0\\
		&\lambda_2 (1-\ovt)=0\\
		&\lambda_3 (\ovt -b)=0\\
		&\lambda_4 (b-a)=0\\
		&\lambda_5 (a-\unt)=0\\
		&\lambda_6 \unt=0\\
		&\lambda_7(F(a)+F(b) - F(\unt)-F(\ovt))=0.
	\end{align*}
	
	First, we note that for $\Pi>0$, we need $1>\ovt$. Moreover, for $1>\ovt$, we need $\unt>0$ by feasibility. Then, we have $\lambda_2=\lambda_6=0$ by complementary slackness conditions.
	
	Remember that for $\unt>0$, we have $\dfrac{\partial \Pi}{\partial \unt}<0$. Then, since $\lambda_6=0$ and $\lambda_5,\lambda_7\geq 0$, for $\dfrac{\partial \mathcal{L}}{\partial \unt} =0$, we need $\lambda_1>0$. By complementary slackness, this implies the feasibility constraint must bind.
	
	Then, the profit function becomes:
	
	\begin{align*}
		\Pi&=- \dfrac{p(F(b)-F(a))}{2\left[ F(\ovt)-F(b)+F(a)-F(\unt) \right]} \left[  \int\limits_{b}^{\ovt} x f(x)dx - \int\limits_{\unt}^{a} x f(x)dx \right] + \dfrac{2-p}{2}\left[ - \underline{\theta}F(\underline{\theta}) -\overline{\theta}F(\overline{\theta}) +\overline{\theta} \right].
	\end{align*}

	Next we are going to argue that in any solution to the above problem with a positive profit, we must have $a=b$. Suppose $(\unt, a,b,\ovt)$ maximizes $\Pi$ and  $b>a$. Remember that we reject any solution that does not satisfy $F(a)-F(\unt)=F(\ovt)-F(b)$, since this is an equilibrium requirement. Then, the differences of integrals in the above equation is nonnegative. Moreover, it is strictly positive if $\Pi>0$:
	
	Notice that for $\Pi>0$, we need $\ovt>\unt$. If $\ovt=\unt$, then it must be the case that $\ovt=a=b=\unt=F^{-1}(0.5)$ and then we can verify that $\Pi=0$. $\ovt>\unt$ implies $\ovt>b$ and $a>\unt$ because (i) we need $F(a)-F(\unt)=F(\ovt)-F(b)$ in the equilibrium and (ii) $\ovt=b>a=\unt$ cannot happen in the equilibrium as shown before stating the Lagrangian problem. But when we have $\ovt>b\geq a >\unt$ and $F(a)-F(\unt)=F(\ovt)-F(b)$, we have:
	\begin{align*}
		\left[  \int\limits_{b}^{\ovt} x f(x)dx - \int\limits_{\unt}^{a} x f(x)dx \right] &= \left[  (F(\ovt)-F(b))\int\limits_{b}^{\ovt} \dfrac{x f(x)dx}{(F(\ovt)-F(b))} - (F(a)-F(\unt))\int\limits_{\unt}^{a} \dfrac{x f(x)dx}{(F(a)-F(\unt))} \right]\\
		&=(F(\ovt)-F(b)) \mathbb{E}[x|x\in[b,\ovt]] - (F(a)-F(\unt)) \mathbb{E}[x|x\in[\unt,a]]\\
		&=(F(\ovt)-F(b)) \left[\mathbb{E}[x|x\in[b,\ovt]] - \mathbb{E}[x|x\in[\unt,a]]\right]>0.
	\end{align*}
	
	Then, we must have $a=b$, since this does not effect the virtual surplus but minimizes the cost. Moreover, it must be the case that $F(a)=F(b)=F^{-1}(\frac{1}{2})$ since we need $F(a)-F(\unt)=F(\ovt)-F(b)$ and the feasibility binds. Thus, we have reduced the unrestricted equilibrium to a simple equilibrium.
\end{proof}

\subsection{Uniqueness of the Equilibrium}\label{app:unique}

\begin{proof}[Proof of Proposition \ref{prop:unique}]
	
	\begin{enumerate}
		
		\item Suppose all agents join the centralized market. Then, everyone's outside option is 0, since no agent trades in the decentralized market. Thus, the marketplace must be operating the baseline mechanism (when it operated without the competition from the decentralized market) to maximize the profit. Then, agents who do not get to trade get 0 utility. However, any two agent can profitably deviate to the decentralized market: They almost surely have different valuations and as long as they have different valuations, they have a positive surplus to share. Thus, it is not a best response for them to stay in the marketplace. Moreover, there is a continuum of such pairs.
		
		\item Suppose all agents joined the decentralized market. We are going to show that there exist $a\leq \unt$ and $b\geq \ovt$ such that agents in $[0,a)\cup(b,1]$ strictly prefer joining the centralized marketplace instead.
		
		When all agents join the decentralized market,

		\begin{align*}
			u^d(0)&=\dfrac{p}{2}\mathbb{E}\left[ \theta | \theta\in [0,1] \right] \\
			u^d(1)&=\dfrac{p}{2}\left[ 1- \mathbb{E}\left[ \theta | \theta\in [0,1] \right]  \right]
		\end{align*}
		
		Next, we compute the utilities the centralized marketplace promises them.
		
		\begin{align*}
			u^m(0)&=\dfrac{p}{2}\left[\mathbb{E}\left[ \theta | \theta\in [\unt,\ovt] \right] - \unt \right] + \unt \\
			u^m(1)&=\dfrac{p}{2}\left[ \ovt- \mathbb{E}\left[ \theta | \theta\in [\unt,\ovt] \right]  \right] + 1-\ovt
		\end{align*}
		
		Now, we check whether these two types can profitably deviate from the decentralized market to the marketplace:
		
		\begin{align*}
			u^m(0) + u^m(1) &>   u^d(0) + u^d(1) \\	
			\iff \dfrac{p}{2}\left[ - \unt \right] + \unt + \dfrac{p}{2}\left[ \ovt \right] + 1-\ovt &>  \dfrac{p}{2}\left[ 1\right]\\
			\iff \unt + 1-\ovt &>  \dfrac{p}{2}\left[ 1 - \ovt + \unt \right]\\
			\iff 1 &> \dfrac{p}{2}
		\end{align*}
		
		Since both $u^m$ and $u^d$ are continuous functions, the above inequality show that there are indeed $a,b$ such that pairs of agents with values in $[0,a)$ and in $(b,1]$ strictly prefer the mechanism to the decentralized market. Again, there is a continuum of pairs that would benefit from deviating. Thus, all agents joining the decentralized market cannot be an equilibrium either.
		
		\item Earlier, I have shown that the only possible coexistence equilibrium is the simple equilibrium, and that a simple equilibrium exists. Now, I have further shown that there is no equilibrium without coexistence. Thus, the unique equilibrium is the simple equilibrium.
		
	\end{enumerate}

\end{proof}

\subsection{Double Auction}\label{app:doubleauctionmain}

To compute all agents' expected payoffs from the decentralized market, we need to know the optimal bids of agents whose valuations lie outside $(\unt,\ovt)$. For agents with values below $\theta\leq \unt$, if they join the decentralized market and get matched to someone, the best response is to bid $b(\unt)$ and for agents with values above $\ovt$, the best response is to bid $b(\ovt)$. I show this in two steps. First, I show that for any agent, the optimal bid must lie in $[b(\unt),b(\ovt)]$. Then, I show that for agents with values less than $\unt$ the optimal bid is $b(\unt)$ while for agents with values above $\ovt$, the optimal bid is $b(\ovt)$. This is stated in the next lemma and proved in the Appendix \ref{app:cornerbid}.

\begin{lemma}\label{lem:cornerbid}
	For each $\theta\in [0,\unt]$, $b(\theta)=b(\unt)$ and for each $\theta\in [\ovt,1]$, $b(\theta)=b(\ovt)$.
\end{lemma}

Agents' utilities from the decentralized market:

Suppose $\theta\in [\unt,\ovt]$. If the agent has the lower value, she has the lower bid, since the bidding function is monotone. Then, the agent gives up her endowment but gets paid. If the agent has the higher value, then, she has the higher bid so she gets the other's endowment and pays for it. Thus, expected payoff is

\begin{align*}
	u^{da}(\theta)&=p G(\theta)\left[ \theta - \int\limits_{\unt}^{\theta} \dfrac{\frac{1}{2}[b(x)+b(\theta)]g(x)dx}{G(\theta)} \right]
	+ p (1-G(\theta))\left[ \int\limits_{\theta}^{\ovt} \dfrac{\frac{1}{2}[b(x)+b(\theta)]g(x)dx}{1-G(\theta)} - \theta \right]\\
	&=p\theta[2G(\theta)-1] + \dfrac{p}{2}\int\limits_{\theta}^{\ovt}[b(x)+b(\theta)]g(x)dx - \dfrac{p}{2} \int\limits_{\unt}^{\theta} [b(x)+b(\theta)]g(x)dx\\
	&=p\theta[2G(\theta)-1] + \dfrac{p}{2}(1-2G(\theta))b(\theta) + \dfrac{p}{2}\int\limits_{\theta}^{\ovt}b(x)g(x)dx - \dfrac{p}{2} \int\limits_{\unt}^{\theta} b(x)g(x)dx
\end{align*}

If $\theta \leq \unt$, then, as above argument shows, the optimal bid is $b(\unt)$ and the expected payoff from the decentralized market is:

\begin{equation*}
	u^{da}(\theta)=p \left[ \int\limits_{\unt}^{\ovt} \frac{1}{2}[b(\unt)+b(x)]g(x)dx - \theta \right]
\end{equation*}

Similarly, agents with $\theta \geq \ovt$ bid $b(\ovt)$ and get the expected payoff:

\begin{equation*}
	u^{da}(\theta)=p \left[ \theta - \int\limits_{\unt}^{\ovt} \frac{1}{2}[b(\ovt)+b(x)]g(x)dx \right]
\end{equation*}

Next, we compute $u^{da}(\unt)$ and $u^{da}(\ovt)$ by using the formulas for $b(\cdot)$, as the IR constraints of $\unt$ and $\ovt$ will again play an important role in the equilibrium. The details can be found in the Appendix \ref{app:da-bindingir} but here is the end result:

\begin{align*}
	u^{da}(\unt) & =  \frac{p}{2}\left[ - \unt + 4\int\limits^{G^{-1}(\frac{1}{2})}_{\unt} \left[ G(x) - \frac{1}{2} \right]^2 dx + \int\limits_{\unt}^{\ovt} b(x)g(x)dx \right]\\
	u^{da}(\ovt) & =  \frac{p}{2}\left[  \ovt + 4 \int\limits_{G^{-1}(\frac{1}{2})}^{\ovt} \left[ G(x) - \frac{1}{2} \right]^2 dx - \int\limits_{\unt}^{\ovt} b(x)g(x)dx \right]
\end{align*}

In the simple equilibrium, we are going to make the IR constraints of $\unt$ and $\ovt$ bind, as otherwise decreasing the payment until they bind increases the profit. For these cutoffs to work, we need the slope of the utility from the decentralized market, $\frac{\partial u^{da}(\theta)}{\partial \theta}$ for types below $\unt$ to be greater than $-1$, since $-1$ is the allocation they will be offered in the mechanism and we want $u^{da}$ to be less than $u^m$ on this region. Moreover, we need the slope of $u^{da}$ to be greater than $-1$ around $\unt$ and the slope should be increasing (thus the utility function should be convex). Finally, we need the slope of the utility from the decentralized market to be less than $1$ for agents with values above $\ovt$.

\begin{lemma}\label{lem:da-slope}
	\begin{equation*}
		\dfrac{\partial u^{da}(\theta)}{\partial \theta} = \begin{cases}
			-p &\text{ if } \theta\leq \unt\\
			p (G(\theta)-1) &\text{ if } \unt\leq\theta\leq \ovt\\
			p &\text{ if } \theta\geq \ovt\\
		\end{cases}
	\end{equation*}
\end{lemma}

Proof can be found in the Appendix \ref{app:da-slope}.

Since this is between $-1$ and $1$, for agents with values in $[\unt,\ovt]$, the designer can indeed offer lower utilities to these agents. One way of doing this would be offering the allocation $-1$ for agents between $\unt$ and $G^{-1}(\frac{1}{2})$ and allocation $1$ for agents between $G^{-1}(\frac{1}{2})$ and $\ovt$. This would make sure these agents are offered utilities lower than their expected payoff from the decentralized market since $u^m(\unt)=u^{da}(\unt)$ and $u^m(\ovt)=u^{da}(\ovt)$. (This may offer utilities below zero for some agents. The designer may not be concerned about this, since these agents are not wanted anyway. However, if the designer wishes the mechanism to offer nonnegative utilities, this can be achieved by flattening the utility when it reaches zero, as the Proposition \ref{simplemech} does.)

Now, we look at the profit function. As before, it is equal to

\begin{align*}\hspace*{-0.5in}
	\Pi_{\unt,\ovt} & = \mathbb{P}[\theta\in[0,\unt]] \mathbb{E}[t(\theta)|\theta\in[0,\unt]] + \mathbb{P}[\theta\in[\ovt,1]] \mathbb{E}[t(\theta)|\theta\in[\ovt,1]] \\
	&= \int\limits_0^{\unt} t(\theta) f(\theta) d\theta + \int\limits_{\ovt}^1 t(\theta)f(\theta)d\theta\\
	&= -F(\unt)u^d(\unt) -(1-F(\ovt ))u^d(\ovt ) +  \int\limits_0^{\unt}    \mathcal{C}(x) q(x) f(x)dx + \int\limits_{\ovt }^1   \mathcal{V}(x) q(x) f(x)dx\\
	&= -F(\unt)u^d(\unt) -(1-F(\ovt ))u^d(\ovt ) -  \int\limits_0^{\unt}  \mathcal{C}(x) f(x)dx +  \int\limits_{\ovt }^1  \mathcal{V}(x) f(x)dx
\end{align*}

First two terms are the compensations for agents to join the centralized marketplace, while the last two terms are the total virtual surplus. We know expression for the total surplus from before because that part is unchanged:

\begin{align*}
	-   \int\limits_0^{\unt}  \mathcal{C}(x) f(x)dx + \int\limits_{\ovt}^1  \mathcal{V}(x) f(x)dx = - \unt F(\unt) + \ovt(1-F(\ovt)) .
\end{align*}

Next, we study the compensations, since now they will be different from what we had for the Nash bargaining.

\begin{align*}\hspace*{-0.5in}
	&F(\unt)u^d(\unt) + (1-F(\ovt ))u^d(\ovt )\\
	=&F(\unt) p \left[  \frac{1}{2}\left[ b(\unt) + \int\limits_{\unt}^{\ovt} b(x)g(x)dx\right] - \unt \right] +(1-F(\ovt))p \left[ \ovt - \frac{1}{2} \left[ b(\ovt) + \int\limits_{\unt}^{\ovt} b(x)g(x)dx \right]   \right]\\
	=&\dfrac{p}{2} \left[ F(\unt)   \left[ b(\unt) + \int\limits_{\unt}^{\ovt} b(x)g(x)dx -2\unt  \right]  +(1-F(\ovt))\left[ 2\ovt -  b(\ovt) - \int\limits_{\unt}^{\ovt} b(x)g(x)dx \right]  \right] \\
	=&p \left[ - \unt F(\unt) + \ovt(1-F(\ovt)) \right]\\
	+&\dfrac{p}{2} \left[ \left[ F(\unt) +F(\ovt) -1\right] \int\limits_{\unt}^{\ovt} b(x)g(x)dx +F(\unt )b(\unt)  -(1-F(\ovt))  b(\ovt)  \right] \\
\end{align*}

Remember that the optimal bidding strategy is given by

\begin{equation*}
	b(\theta)=\theta - \dfrac{\int\limits_{G^{-1}(\frac{1}{2})}^{\theta} [G(x)-\dfrac{1}{2}]^2 dx }{[G(\theta)-\dfrac{1}{2}]^2}
\end{equation*}

with $G(\theta)=\dfrac{F(x)-F(\unt)}{F(\ovt)-F(\unt)}$ on $[\unt,\ovt]$ and $G^{-1}(\frac{1}{2})=F^{-1}\left( \frac{F(\unt)+F(\ovt)}{2} \right)$. Although it is relatively easy to show that a simple equilibrium exists with arbitrary distributions, the general solution to the profit maximization problem is too complicated to provide some useful comparative statics. Hence, I focus on the uniform distribution, $U[0,1]$ from here on to be able to find a closed form solution to the problem above.

Using the uniform distribution, with some algebra (see Appendix \ref{app:da-unibids}), we can show that

\begin{align*}
	b(\theta) &= \dfrac{\ovt+\unt+4\theta}{6},\\
	\int\limits_{\unt}^{\ovt} b(x)g(x)dx &= \dfrac{\ovt+\unt}{2}.
\end{align*}

Now, we can plug these back into the expression for the profit to see that it is decreasing in $\unt$. Thus, the feasibility must bind, which means $\unt=1-\ovt$. Using this, we simplify the profit further and obtain the following simple expression for the profit. (The derivations can be followed in Appendix \ref{app:da-profit}.)

\begin{align*}
	\Pi_{\unt,\ovt} & =\dfrac{6-5p}{6} \left[\unt(\ovt-\unt) \right] = \dfrac{6-5p}{6} \Pi^M
\end{align*}

Clearly, this problem has an interior solution, which is the same as the solution of the problem of the marketplace when it operated on its own: The profits in two cases are equal up to a constant multiplier. Thus, almost everything we have seen under the Nash bargaining hold here with the uniform distribution, with the exception of the change ratio of the profit.

\subsubsection{Optimal Bids for Agents in the Marketplace}\label{app:cornerbid}
	
	\begin{proof}[Proof of Lemma \ref{lem:cornerbid}]
		Step 1: Bidding below $b(\unt)$ can never be optimal: For any bid $b\leq b(\unt)$, the agent sells her endowment with certainty but get paid less than she would get if she bid $b(\unt)$. Mathematically, the expected utility of an agent who bids $b\leq b(\unt)$ is given by:
		
		\begin{equation*}
			\dfrac{1}{2}\int\limits_{\unt}^{\ovt} [b+b(x)]g(x)dx - \theta
		\end{equation*}
		
		Thus, bids strictly below $b(\unt)$ cannot be optimal.
		
		Similarly, bids strictly above $b(\ovt)$ cannot be optimal either. In that case, the agent's expected payoff would be
		
		\begin{equation*}
			\theta -	\dfrac{1}{2}\int\limits_{\unt}^{\ovt} [b+b(x)]g(x)dx
		\end{equation*}
		
		Hence, for each $\theta\in [0,\unt]\cup [\ovt,1]$, the optimal bid must be the bid of some type joins the decentralized market, that is: $b(\theta)\in[b(\unt),b(\ovt)]$.
		
		Step 2: Let us first define the following notation: If an agent bids $b(\theta')$, then the expected price for selling is $p_s(\theta')=\frac{1}{2}\int\limits_{\unt}^{\theta'} \frac{[b(\theta')+b(x)]g(x)dx}{1-G(\theta')}$ and the expected price for buying is $p_b(\theta')=\frac{1}{2}\int\limits_{\theta'}^{\ovt} \frac{[b(\theta')+b(x)]g(x)dx}{G(\theta')}$.
		
		Since the best response of agent with value $\unt$ is $b(\unt)$, her expected payoff from this bid should be higher than any other $b(\theta')$ by revealed preference. Then,
		
		\begin{align*}
			p_s(\unt) - \unt &\geq (1-G(\theta')) \left[ p_s(\theta') - \unt \right] + G(\theta') \left[ \unt - p_b(\theta')  \right]\\
			&=\unt [2G(\theta')-1] + (1-G(\theta')) p_s(\theta') - G(\theta') p_b(\theta')\\
			\iff 	p_s(\unt)- (1-G(\theta')) p_s(\theta') + G(\theta') p_b(\theta')  &\geq 2\unt G(\theta')
		\end{align*}

		Suppose $\theta \leq \unt$. We want to show that bidding $b(\unt)$ gives a higher payoff than any other type's bid $b(\theta')$:
		
		\begin{align*}
			p_s(\unt) - \theta &\geq (1-G(\theta')) \left[ p_s(\theta') -\theta \right] + G(\theta') \left[ \unt - p_b(\theta')  \right]\\
			&=\theta [2G(\theta')-1] + (1-G(\theta')) p_s(\theta') - G(\theta') p_b(\theta')\\
			p_s(\unt)- (1-G(\theta')) p_s(\theta') + G(\theta') p_b(\theta')  &\geq 2\theta G(\theta')
		\end{align*}
		
		But this is true since $p_s(\unt)- (1-G(\theta')) p_s(\theta') + G(\theta') p_b(\theta')  \geq 2\unt G(\theta') \geq 2\theta G(\theta')$ where the first inequality follows from the revealed preference argument above and the second one follows from $\unt\geq \theta$ and $G(\theta')\geq 0$.

		Similarly, the best response of an agent with value $\unt$ is $b(\ovt)$. Thus,
		
		\begin{align*}
			\ovt - p_b(\ovt) &\geq (1-G(\theta')) \left[ p_s(\theta') -\ovt \right] + G(\theta') \left[ \ovt - p_b(\theta')  \right]\\
			&=\ovt [2G(\theta')-1] + (1-G(\theta')) p_s(\theta') - G(\theta') p_b(\theta')\\
			\iff 2\ovt (1-G(\theta'))   &\geq p_b(\ovt) + (1-G(\theta')) p_s(\theta') - G(\theta') p_b(\theta')
		\end{align*}

		Suppose $\theta \geq \ovt$. In this case, we want to show that bidding $b(\ovt)$ gives a higher payoff than any other type's bid $b(\theta')$:
		
		\begin{align*}
			\theta - p_b(\ovt) &\geq (1-G(\theta')) \left[ p_s(\theta') -\theta \right] + G(\theta') \left[ \theta - p_b(\theta')  \right]\\
			&=\theta [2G(\theta')-1] + (1-G(\theta')) p_s(\theta') - G(\theta') p_b(\theta')\\
			\iff 2\theta (1-G(\theta'))   &\geq p_b(\ovt) + (1-G(\theta')) p_s(\theta') - G(\theta') p_b(\theta')
		\end{align*}
		
		Again, this is true since $ 2\theta (1-G(\theta'))   \geq 2\ovt (1-G(\theta'))   \geq p_b(\ovt) + (1-G(\theta')) p_s(\theta') - G(\theta') p_b(\theta')$ where the first inequality again follows from the revealed preference argument above and the second one follows from $\theta \geq \ovt$ and $1-G(\theta')\geq 0$.
	\end{proof}

\subsubsection{Binding IR constraints}\label{app:da-bindingir}
	
	\begin{align*}
		u^{da}(\unt) & = p \left[ \int\limits_{\unt}^{\ovt} \frac{1}{2}[b(\unt)+b(x)]g(x)dx - \unt \right] \\
		& =	p \left[  \frac{1}{2} \int\limits_{\unt}^{\ovt} [b(\unt)+b(x)]g(x)dx - \unt \right]\\
		& =	p \left[  \frac{1}{2}\left[ b(\unt) + \int\limits_{\unt}^{\ovt} b(x)g(x)dx\right] - \unt \right]\\
		& =	p \left[  \frac{1}{2}\left[ \unt - \dfrac{\int\limits_{G^{-1}(\frac{1}{2})}^{\unt} \left[ G(x) - \frac{1}{2} \right]^2 dx}{\left[ G(\unt) - \frac{1}{2} \right]^2} + \int\limits_{\unt}^{\ovt} b(x)g(x)dx\right] - \unt \right]\\
		& =	 \frac{p}{2}\left[ - \unt + 4\int\limits^{G^{-1}(\frac{1}{2})}_{\unt} \left[ G(x) - \frac{1}{2} \right]^2 dx + \int\limits_{\unt}^{\ovt} b(x)g(x)dx \right]
	\end{align*}
	
	Similarly, we obtain,
	
	\begin{align*}
		u^{da}(\ovt) & = p \left[ \ovt - \int\limits_{\unt}^{\ovt} \frac{1}{2}[b(\ovt)+b(x)]g(x)dx  \right] \\
		& = p \left[ \ovt - \frac{1}{2} \left[ b(\ovt) + \int\limits_{\unt}^{\ovt} b(x)g(x)dx \right]   \right] \\
		& = p \left[ \ovt - \frac{1}{2} \left[ \ovt - \dfrac{\int\limits_{G^{-1}(\frac{1}{2})}^{\ovt} \left[ G(x) - \frac{1}{2} \right]^2 dx}{\left[ G(\ovt) - \frac{1}{2} \right]^2}  + \int\limits_{\unt}^{\ovt} b(x)g(x)dx \right]   \right] \\
		& =	 \frac{p}{2}\left[  \ovt + 4 \int\limits_{G^{-1}(\frac{1}{2})}^{\ovt} \left[ G(x) - \frac{1}{2} \right]^2 dx - \int\limits_{\unt}^{\ovt} b(x)g(x)dx \right]
	\end{align*}

\subsubsection{Slope of Utilities from the Double Auction}\label{app:da-slope}
	
	\begin{proof}[Proof of Lemma \ref{lem:da-slope}]
		It is easy to verify that for agents with values less than $\unt$, $\frac{\partial u^{da}(\theta)}{\partial \theta}=-p\geq -1$ and for agents with values above $\ovt$, $\frac{\partial u^{da}(\theta)}{\partial \theta}=p\leq 1$. Next we show that $\frac{\partial u^{da}(\theta)}{\partial \theta}$ is greater than $-1$ for $\unt$ and less than $1$ for $\ovt$.
		
		First, we need the derivative of the bidding function:
		
		\begin{align*}
			b(\theta)&=\theta - \dfrac{\int\limits_{G^{-1}(\frac{1}{2})}^{\theta} [G(x)-\dfrac{1}{2}]^2 dx }{[G(\theta)-\dfrac{1}{2}]^2}\\
			b'(\theta)&=1 - \dfrac{ [G(\theta)-\dfrac{1}{2}]^4 -2g(\theta)(G(\theta)-\dfrac{1}{2}) \int\limits_{G^{-1}(\frac{1}{2})}^{\theta} [G(x)-\dfrac{1}{2}]^2 dx }{[G(\theta)-\dfrac{1}{2}]^4}\\
			&=\dfrac{ 2g(\theta)(G(\theta)-\dfrac{1}{2}) \int\limits_{G^{-1}(\frac{1}{2})}^{\theta} [G(x)-\dfrac{1}{2}]^2 dx }{[G(\theta)-\dfrac{1}{2}]^4}\\
			&= 2g(\theta) \dfrac{  \int\limits_{G^{-1}(\frac{1}{2})}^{\theta} [G(x)-\dfrac{1}{2}]^2 dx }{[G(\theta)-\dfrac{1}{2}]^3}
		\end{align*}
		
		\begin{align*}
			u^{da}(\theta)&=p\theta[2G(\theta)-1] + \dfrac{p}{2}(1-2G(\theta))b(\theta) + \dfrac{p}{2}\int\limits_{\theta}^{\ovt}b(x)g(x)dx - \dfrac{p}{2} \int\limits_{\unt}^{\theta} b(x)g(x)dx
		\end{align*}
		
		\begin{align*}
			\dfrac{\partial u^{da}(\theta)}{\partial \theta} &= p\left[ (2G(\theta) -1) + 2\theta g(\theta) \right] + \dfrac{p}{2} \left[ -2g(\theta)b(\theta) + (1-2G(\theta)) b'(\theta) \right] -\dfrac{p}{2}2b(\theta)g(\theta)\\
			&= p (2G(\theta) -1) + 2 p g(\theta)(\theta - b(\theta))  + \dfrac{p}{2} \left[ (1-2G(\theta)) b'(\theta) \right]\\
			&= p\left[ G(\theta) - 1 + 2 g(\theta)(\theta - b(\theta)) + \dfrac{1}{2} \left[ (1-2G(\theta)) b'(\theta) \right] \right]   \\
			&= p\left[ G(\theta) - 1 + 2g(\theta)\left[ \theta - \theta + \dfrac{\int\limits_{G^{-1}(\frac{1}{2})}^{\theta} [G(x)-\dfrac{1}{2}]^2 dx }{[G(\theta)-\dfrac{1}{2}]^2} \right] + \dfrac{1}{2}  (1-2G(\theta)) b'(\theta)  \right]   \\
			&= p\left[ G(\theta) - 1 + 2 g(\theta) \dfrac{\int\limits_{G^{-1}(\frac{1}{2})}^{\theta} [G(x)-\dfrac{1}{2}]^2 dx }{[G(\theta)-\dfrac{1}{2}]^2}  + \dfrac{1}{2}  (1-2G(\theta)) \left[ 2 g(\theta) \dfrac{\int\limits_{G^{-1}(\frac{1}{2})}^{\theta} [G(x)-\dfrac{1}{2}]^2 dx }{[G(\theta)-\dfrac{1}{2}]^3} \right] \right]   \\
			&= p\left[ G(\theta) - 1 + 2 g(\theta) \dfrac{\int\limits_{G^{-1}(\frac{1}{2})}^{\theta} [G(x)-\dfrac{1}{2}]^2 dx }{[G(\theta)-\dfrac{1}{2}]^2}  -  2 g(\theta) \dfrac{\int\limits_{G^{-1}(\frac{1}{2})}^{\theta} [G(x)-\dfrac{1}{2}]^2 dx }{[G(\theta)-\dfrac{1}{2}]^2}  \right]   \\
			&= p\left[ G(\theta) - 1  \right]   \\
		\end{align*}
	\end{proof}

\subsubsection{Bids with Uniform Distribution}\label{app:da-unibids}

\begin{align*}
	b(\theta)&=\theta - \dfrac{\int\limits_{G^{-1}(\frac{1}{2})}^{\theta} [G(x)-\dfrac{1}{2}]^2 dx }{[G(\theta)-\dfrac{1}{2}]^2} &=& \theta - \dfrac{\int\limits_{\frac{\unt+\ovt}{2}}^{\theta} \left[\dfrac{x-\unt}{\ovt-\unt}-\dfrac{1}{2}\right]^2 dx }{\left[\dfrac{\theta-\unt}{\ovt-\unt}-\dfrac{1}{2}\right]^2}\\
	&= \theta - \dfrac{\int\limits_{\frac{\unt+\ovt}{2}}^{\theta} \left[\dfrac{2x-\unt-\ovt}{2(\ovt-\unt)}\right]^2 dx }{\left[\dfrac{2\theta-\unt-\ovt}{2(\ovt-\unt)}\right]^2} &=& \theta - \dfrac{\int\limits_{\frac{\unt+\ovt}{2}}^{\theta} \left[2x-\unt-\ovt\right]^2 dx }{\left[2\theta-\unt-\ovt\right]^2}\\	
	&= \theta - \dfrac{ \left[ \frac{1}{2\times 3} \left[2x-\unt-\ovt\right]^3 \right]_{\frac{\unt+\ovt}{2}}^{\theta} }{\left[2\theta-\unt-\ovt\right]^2} &=& \theta + \frac{1}{6} \dfrac{  \left[2\theta-\unt-\ovt\right]^3  }{\left[2\theta-\unt-\ovt\right]^2}\\	
	&= \theta - \frac{2\theta-\unt-\ovt }{6}   &=& \frac{4\theta+\unt+\ovt }{6}
\end{align*}

Next, we compute another expression from the profit function:

\begin{align*}
	&\int\limits_{\unt}^{\ovt} b(x)g(x)dx=\int\limits_{\unt}^{\ovt} \dfrac{4\theta+\unt+\ovt }{6} \dfrac{1}{\ovt -\unt}dx =\dfrac{1}{6(\ovt-\unt)} \int\limits_{\unt}^{\ovt} (4\theta+\unt+\ovt) dx\\
	=& \dfrac{1}{6(\ovt-\unt)} \left[ 2 \ovt^2 -2\unt^2 +(\ovt+\unt)(\ovt-\unt) \right] = \dfrac{1}{6(\ovt-\unt)} \left[ 3(\ovt+\unt)(\ovt-\unt) \right] = \dfrac{\ovt+\unt}{2}
\end{align*}

\subsubsection{Profit from Simple Equilibrium under Double Auction}\label{app:da-profit}

\begin{align*}\hspace*{-0.5in}
	\Pi_{\unt,\ovt} & =  -   \int\limits_0^{\unt}  \mathcal{C}(x) f(x)dx +  \int\limits_{\ovt }^1  \mathcal{V}(x) f(x)dx -F(\unt)u^d(\unt) -(1-F(\ovt ))u^d(\ovt ) \\
	&=  \left[ - \unt F(\unt) + \ovt(1-F(\ovt)) \right] - p \left[ - \unt F(\unt) + \ovt(1-F(\ovt)) \right]\\
	-&\dfrac{p}{2} \left[ \left[ F(\unt) +F(\ovt) -1\right] \int\limits_{\unt}^{\ovt} b(x)g(x)dx +F(\unt )b(\unt)  -(1-F(\ovt))  b(\ovt)  \right]\\
	&= (1-p) \left[ - \unt F(\unt) + \ovt(1-F(\ovt)) \right]\\
	-&\dfrac{p}{2} \left[ \left[ F(\unt) +F(\ovt) -1\right] \int\limits_{\unt}^{\ovt} b(x)g(x)dx +F(\unt )b(\unt)  -(1-F(\ovt))  b(\ovt)  \right]\\
	&= (1-p) \left[ - \unt^2 + \ovt(1-\ovt) \right]\\
	-&\dfrac{p}{2} \left[ \left[ \unt +\ovt -1\right] \dfrac{\ovt+\unt}{2} +\unt \dfrac{5\unt +\ovt}{6}  -(1-\ovt)  \dfrac{5\ovt +\unt}{6}  \right]
\end{align*}

The profit is decreasing in $\unt$:

\begin{align*}\hspace*{-0.5in}
	\dfrac{\partial \Pi_{\unt,\ovt}}{\partial \unt} & = (1-p) \left[ - 2\unt \right] -\dfrac{p}{2} \left[  \dfrac{\ovt+\unt}{2} + \dfrac{1}{2}\left[ \unt +\ovt -1\right]  +\dfrac{5\unt +\ovt}{6} + \unt \dfrac{5}{6}  -(1-\ovt)  \dfrac{1}{6}  \right]\\
	& = (1-p) \left[ - 2\unt \right] -\dfrac{p}{2} \left[ \dfrac{2}{3} (4\unt + 2 \ovt - 1)  \right]\leq 0\\
\end{align*}

Notice that the first summand is negative and inside the brackets of the second summand is positive since $\unt \geq 1-\ovt$ by feasibility. Thus, the profit is decreasing in $\unt$. Hence, the feasibility binds and we have $\unt=1-\ovt$, otherwise decreasing $\unt$ until the feasibility binds strictly increases the profit. Then, we have

\begin{align*}\hspace*{-0.5in}
	\Pi_{\unt,\ovt} & = (1-p) \left[ - \unt^2 + \ovt(1-\ovt) \right] - \dfrac{p}{2} \left[ \unt \dfrac{5\unt +\ovt}{6}  -(1-\ovt)  \dfrac{5\ovt +\unt}{6}  \right]\\
	& = (1-p) \left[ - \unt^2 + \ovt(1-\ovt) \right] - \dfrac{p}{2} \left[ \unt \dfrac{5\unt +\ovt - 5\ovt - \unt}{6}   \right]\\
	& = (1-p) \left[ - \unt^2 + (1-\unt)\unt \right] - \dfrac{p}{2} \left[ \unt \dfrac{5\unt +(1-\unt) - 5(1-\unt) - \unt}{6}   \right]\\
	& = (1-p) \left[\unt(1-2\unt) \right] - \dfrac{p}{2} \left[ \unt \dfrac{8\unt -4}{6}   \right] = (1-p) \left[\unt(1-2\unt) \right] + \dfrac{p}{6} \left[ \unt (1-2\unt)   \right]\\
	&\dfrac{6-5p}{6} \left[\unt(1-2\unt) \right] = \dfrac{6-5p}{6} \Pi^M
\end{align*}

\end{document}


\maketitle
	
	The results below has first been obtained in \cite{idem} for an environment with finitely many agents, divisible goods and arbitrary endowments. Here I restate them for the environment I study in the Coexistence of Centralized and Decentralized Markets, with the proofs adjusted accordingly.
	
	\section{Monagora Environment}
	
	I reproduce the setup and the initial statement of the mechanism design problem here for convenience.
	
	\begin{itemize}
		
		\item Good: There is a single, indivisible good in the market.
		
		\item Agents: There is a continuum of agents on $[0,1]$.
		
		\item Endowments: Each agent has $1$ unit of endowment of the good.
		
		\item Demands: Each agent demands up to 2 units of the good. Since the good is indivisible, this means, they can consume 0, 1, or 2 units, depending on whether they buy or sell, or neither buy nor sell.
		
		\item Valuations: Each agent has some valuation $\theta\in[0,1]$ for a unit of the good. The valuations are drawn from some distribution $F$ with support $[0,1]$. Agents' valuations are their private information.
		
		\item Marketplace: A mechanism designer wants to design a mechanism to maximize its profit. She knows the distribution of valuations, $F$.
		
	\end{itemize}

	By revelation principle, I focus on direct mechanisms. Moreover, as agents are symmetric other than their valuations, I focus on anonymous mechanisms, which is without loss. Then, the designer will choose a mechanism that allocates $q:\theta \rightarrow \mathbb{R}$ units of good to each agent with valuation $\theta$ and asks her to pay $t:\theta \rightarrow \mathbb{R}$. Hence, the net utility of the agent with the valuation $\theta$ from the monagorastic mechanism is $$ u(\theta) = \theta \min\{ 1, q(\theta) \}  - t(\theta). $$
	
	As agents have demands for two units, having more than 2 unit of the good is same as having 2 unit. As such, the expression for the utility above caps the maximum net trade that increases the utility at $1$, since the agent already has $1$ unit of endowment.
	
	The profit of the marketplace is the net payments. Thus, the designer seeks to maximize total payment, given the incentive compatibility, individual rationality, and feasibility constraints.
	
	\begin{equation*}
		\begin{array}{lllll}
			
			\displaystyle\max\limits_{(q, t)} \int\limits_{[0,1]} t(\theta) f(\theta) d\theta \\

			\text{s. t. }\\
			
			\text{ (IC) } &\theta  \min\{ 1, q(\theta)\}  - t(\theta)  &\geq \theta \min\{ 1, q(\theta') \}   - t(\theta')\\
			
			\text{ (IR) }& \theta  \min\{ 1, q(\theta)  \}  - t(\theta) &\geq 0\\  
			
			\text{ (Individual Feasibility) }& q(\theta) &\geq -1\\

			\text{ (Aggregate Feasability) }& \int\limits_{[0,1]} q(\theta) f(\theta) d\theta &\leq 0\\
			
		\end{array}
	\end{equation*}

	\subsection{Simplifying The Designer's Problem}
	
	We first develop a series of lemmata that help us state the maximization problem above as a concave program.
	
	\begin{lemma}[Monotonicity]
		Suppose $(q, t)$ is a direct, IC mechanism. Then,
		
		\begin{enumerate}
			\item If $q(\theta) < 1$ for some $\theta\in [0,1]$, then $q(\theta)$ is increasing at $(\theta)$.
			
			\item If $q(\theta) \geq 1$ for some $\theta\in [0,1]$, then $q(\theta') \geq 1$ for each $\theta'\geq \theta$.
		\end{enumerate}
		
	\end{lemma}

	The proof is standard, except for taking care of the capacities so it can be found in the Appendix \ref{lemma1}.

	The next lemma presents the derivative of the utility of an agent in an IC mechanism.

	\begin{lemma}[Envelope Condition]
		If $(q, t)$ is a direct, IC mechanism, then for each $\theta\in[0,1]$
		
		$$\dfrac{\partial u(\theta)}{\partial \theta} =
		\begin{cases}
			q(\theta) , &\text{ if } q(\theta) < 1,\\
			1, &\text{ otherwise. }\\
		\end{cases}$$
	\end{lemma}

	Again, the proof is similar to standard arguments and can be found in Appendix \ref{lemma2}.

	\textbf{Notation:} For any direct mechanism $(q, t)$, let
	
	$$q^*(\theta)=
	\begin{cases}
		q(\theta) , &\text{ if } q(\theta) < 1,\\
		1, &\text{ otherwise. }\\
	\end{cases}$$
	
	Note that for a direct, IC mechanism, $q^*(\theta)$ is also weakly increasing.

	The next lemma gives the representation of the utility of each type as the integral of the allocation rule, using the previous lemma.

	\begin{lemma}[Payoff Equivalence]\label{integrable}
		If $(q, t)$ is a direct, IC mechanism, then $$ u(\theta) = u(0) + \int^{\theta}_{0}  q^*(x)dx  ,$$ for each $\theta\in [0.1]$.
	\end{lemma}

	\begin{proof}
		Since $u(\theta)$ is convex $\theta$ on both regions where $ q(\theta)> 1$ and $ q(\theta) \leq 1$ separately, it is absolutely continuous in $\theta$. Then, it is the integral of its derivative.
	\end{proof}

	Next, we pin down the transfer rule in an IC mechanism.

	\begin{lemma}[Revenue Equivalence]
		If $(q, t)$ is a direct, IC mechanism, then
		
		\begin{equation*}
			\begin{array}{ll}
				t(\theta)&= -u(0) +  \theta  q^*(\theta)  - \int\limits^{\theta}_{0} q^*(x)dx , \\
			\end{array}
		\end{equation*}

		for each $\theta\in[0,1]$. 
		
	\end{lemma}

	\begin{proof}
		From the definition of $u(\theta)$ and the previous lemma.
	\end{proof}

	Now we show that the necessary conditions above for incentive compatibility of a mechanism are also sufficient to establish the incentive compatibility of a mechanism.

	\begin{proposition}
		
		Let $(q, t)$ be a direct mechanism. The mechanism is incentive compatible if and only if,
		
		\begin{enumerate}
			\item $q^*(\theta)$ is increasing at $\theta$;

			\item $ t(\theta)= -u(0) +  \theta  q^*(\theta)  - \int\limits^{\theta}_{0} q^*(x)dx $.
		\end{enumerate}

	\end{proposition}

	Proof can be found in Appendix \ref{prop1}.

	The next proposition provides the characterization of the IR mechanisms by establishing the types with the lowest utilities. The reason this is an issue in this model is that in an auction, the lowest allocation an agent could receive is 0. Hence, the utility is always increasing in agent's type, as can be seen from the envelope condition. Of course, this means the lowest type has the lowest utility. However, here, an agent with a relatively low type can be a seller, which means he would get a negative allocation. Therefore, the utility of the lowest type is not the lowest utility, which can again by seen from the envelope condition.

	\begin{proposition}
		Let $(q, t)$ be a direct IC mechanism. Then, it is IR if and only if,
		
		$$ 	\theta^* q^*(\theta^*) \geq t(\theta^*), $$

		where $\theta^*$ is defined as

		\begin{enumerate}
			
			\item $\theta^*=0$ if $q^*(0)\geq 0$,
			
			\item $\theta^*=1$ if $q^*(1)<0$,
			
			\item a solution to $q^*(\theta^*)=0$ if such a type exists.
		\end{enumerate}
		
	\end{proposition}

	\begin{proof}
		
		\textit{Case 1:} Suppose $q^*(0)\geq 0$. Then, by Lemma \ref{integrable}, incentive compatibility of a mechanism implies that the associated ex-post utilities $u(\theta)$ are increasing in $\theta$. Hence, if $u(0) \geq 0$, we have $u(\theta)\geq 0$ for each $\theta\in [0,1]$.

		\textit{Case 2:} Suppose $q^*(1)<0$. Then, by Lemma \ref{integrable}, $u(\theta)$ are decreasing and hence,  $u(1)$ is the lowest payoff. Hence, if it is nonnegative, all other types' payoffs are nonnegative as above.
		
		\textit{Case 3:} Suppose there exists $\theta^*$ such that $q^*(\theta^*)=0$. Then, by Lemma \ref{integrable}, $u(\theta)$ is decreasing up to $\theta^*$ and increasing after that point. Hence, type $\theta^*$ has the lowest payoff. So, if $u(\theta^*)\geq 0$, each type's IR condition must also hold.
		
	\end{proof}

	\begin{lemma}\label{irprofmax}
		If an IC and IR mechanism maximizes the expected revenue of the designer, then,
		
		$$ t(\theta^*) = \theta^* q(\theta^*) $$
		
		where $\theta^*$ is defined as

		\begin{enumerate}
			
			\item $\theta^*=0$ if $q^*(0)\geq 0$,
			
			\item $\theta^*=1$ if $q^*(1)<0$,
			
			\item the solution to $q^*(\theta^*)=0$ if such a type exists.
		\end{enumerate}
		
	\end{lemma}
	
	\begin{proof}
		
		The previous proposition shows that IC and IR mechanisms must have $\theta^* q(\theta^*)$ greater than $t(\theta^*) $. However, if $ 	\theta^* q(\theta^*) > t(\theta^*) $, then the seller can increase the expected revenue by increasing $t(0)$ and keeping the allocation rule the same. This would increase all types' payments and the revenue strictly, contradicting revenue maximization.
		
	\end{proof}

	Using the condition about $\theta^*$ from Lemma \ref{irprofmax} and the previous lemmata, we have

	\begin{equation*}
		\begin{array}{lll}
			\theta^* q^*(\theta^*) &= t(\theta^*) \\
			
			= & -u(0) + \theta^* q^*(\theta^*) - \int\limits^{\theta^*}_{0} q^*(x)dx\\
			
			\iff & u(0) = - \int\limits^{\theta^*}_{0} q^*(x)dx\\
			
			\iff & t(\theta)= \int\limits^{\theta^*}_{0} q^*(x)dx + \theta q^*(\theta)  - \int\limits^{\theta}_{0} q^*(x)dx
		\end{array}
	\end{equation*}

	Now we are ready to show that the allocation rule in a revenue-maximizing mechanism is not `wasteful'.
	
	\begin{proposition}
		Let $(q, t)$ be a direct mechanism that maximizes the revenue of the designer. Then, $q(\theta) \leq 1$ with probability 1 and the aggregate feasibility holds with equality: $\int\limits_{[0,1]} q(\theta) f(\theta) d\theta = 0$.
	\end{proposition}

	\begin{proof}
		
		First, suppose that in the optimal mechanism, there exists a set $\Theta\subset[0,1]$ with a positive measure such that for each $\theta\in \Theta$, $q(\theta) > 1$. Notice that decreasing the allocation to $1$ unit has no effect on the agent's payoff. Hence, it doesn't effect any IC or IR constraints.
		
		Next, let us examine the transfer rule in a direct, IC mechanism:
		
		$$ t(\theta)= \int^{\theta^*}_{0} q^*(x)dx  + \theta q^*(\theta) - \int^{\theta}_{0} q^*(x)dx. $$
		
		If we have $q(\theta) > 1$ for a positive measure of types, then we must have $q(\theta) < 0$ for a corresponding positive measure of types by the aggregate feasibility constraint. Hence, if we reduced $q(\theta) = 1$ for $\theta\in\Theta$, this wouldn't affect any constraints but instead strictly increase profit as it allows us to increase $q(\theta) < 0$ for a positive measure of types, contradicting the optimality of the mechanism.
		
		By the same argument, having $\int\limits_{[0,1]} q(\theta) f(\theta) d\theta < 0$ cannot be optimal: Either buying less from types or selling more to some types would increase their payments, strictly increasing the profit.
	\end{proof}

	Now, by fixing $t(\theta)$ to the characterization we have from above, we can restate the problem as follows.

	\begin{equation*}\label{problem}
		\begin{array}{llllll}
		
			\displaystyle\max_{q} &  \bigint\limits_{[0,1]} \Biggl[ \int\limits_{\{ y|q(y)\leq 0 \}} q(x)dx +  \left( \theta q(\theta)  - \int\limits^{\theta}_{0} q(x)dx \right) \Biggr]  f(\theta)   d\theta\\

			\text{s. t. }&\\

			&q(\theta) \text{ is increasing}\\
			
			&q(\theta) \geq -1\\
			
			&\int\limits_{[0,1]} q(\theta) f(\theta) d\theta = 0\\
			
		\end{array}
	\end{equation*}
	
	After some transformations\footnote{The details can be followed in Appendix \ref{transform}.}, the problem above can be rewritten as follows:

	\begin{equation*}
		\begin{array}{llllll}
		
			\displaystyle\max_{q}  &\left[ \int\limits_{[0,1]} q(\theta) \left[  \dfrac{  \mathbbm{1} \{ q(\theta)\leq 0 \}}{f(\theta)} + \left(\theta  - \dfrac{(1-F(\theta))}{f(\theta)} \right)  \right]f(\theta)   d\theta \right] \\

			\text{s. t. }&\\
			
			& q(\theta) \text{ is increasing}\\
			
			&  q(\theta) \geq -1\\
			
			&\int\limits_{[0,1]} q(\theta) f(\theta) d\theta = 0
			
		\end{array}
	\end{equation*}

	\bibliographystyle{apacite}
	\bibliography{marketplace}

	\appendix

	\section{Proof of Lemma 1}\label{lemma1}
	
	\begin{proof}

		Let $\theta, \theta'\in [0,1]$. Then, by incentive compatibility
		
		$$ \theta \min\{ 1, q(\theta) \}  - t(\theta) \geq \theta \min\{ 1, q(\theta') \} - t(\theta')$$
		
		and		
		$$ \theta' \min\{ 1, q(\theta) \}  - t(\theta) \leq \theta' \min\{ 1, q(\theta') \}  - t(\theta').$$
		
		Subtracting the second inequality from the first one leads to:
		
		$$ (\theta-\theta') \min\{ 1, q(\theta)  \} \geq  (\theta-\theta') \min\{ 1, q(\theta') \}$$ 
		
		Suppose $q(\theta) < 1$ and $\theta > \theta'$. Then, we have
		
		\begin{equation*}
			\begin{array}{lll}
				\min\{ 1, q(\theta) \} &\geq   \min\{ 1, q(\theta') \} \iff \\
				1 > q(\theta)  &\geq   \min\{ 1, q(\theta')  \} \iff \\
				q(\theta) &\geq q(\theta')\\
			\end{array}
		\end{equation*}

		Now suppose $q(\theta)  \geq 1$ and $\theta'\geq \theta$. Then,

		\begin{equation*}
			\begin{array}{lll}
				\min\{ 1, q(\theta')  \} &\geq   \min\{ 1, q(\theta) \} \iff \\
				\min\{ 1, q(\theta')  \}  &\geq   1 \iff \\
				q(\theta')  &\geq 1.\\
			\end{array}
		\end{equation*}
		
	\end{proof}

	\section{Proof of Lemma 2}\label{lemma2}
	
	\begin{proof}
		
		First, suppose $ q(\theta ) < 1$. Then, IC implies that for type $\theta$ agent:

		\begin{equation*}
			\begin{array}{ll}
				u(\theta) &= \max\limits_{\theta'\in [0,1]} \min\{ 1, q(\theta') \} \theta - t(\theta') \\
				&=  \max\limits_{\theta'\in [0,1]} q(\theta') \theta - t(\theta').
			\end{array}
		\end{equation*}

		Notice that the RHS is the maximum of affine functions of $\theta$, so $u(\theta)$ is convex in $\theta$ on this region. Hence, $u(\theta)$ is differentiable almost everywhere in $\theta$ on this region. For any $\theta$ at which it is differentiable, for $\delta>0$, IC implies that
		
		\begin{equation*}
			\begin{array}{lllll}
				&\lim\limits_{\delta \rightarrow 0} \dfrac{u(\theta + \delta) - u(\theta) }{\delta}\\				
				\geq & \lim\limits_{\delta \rightarrow 0} \dfrac{ (q(\theta) (\theta +\delta ) - t(\theta)) -( q(\theta) \theta - t(\theta)) }{\delta} = q(\theta).\\
				&\lim\limits_{\delta \rightarrow 0} \dfrac{u(\theta) - u(\theta- \delta) }{\delta}\\
				\leq & \lim\limits_{\delta \rightarrow 0} \dfrac{ (q(\theta) \theta   - t(\theta)) -( q(\theta) (\theta- \delta) - t(\theta)) }{\delta} = q(\theta).
			\end{array}
		\end{equation*}

		Then, two inequalities together imply that
		
		$$\dfrac{\partial u(\theta)}{\partial \theta} = q(\theta ).$$
		
		Now suppose $ q(\theta) \geq 1$. Then,
		
		$$ u(\theta)= \min\{ 1, q(\theta)  \}  \theta - t(\theta) = \theta - t(\theta). $$
		
		Notice that $t(\theta)$ must be constant in $\theta$ on the region with  $ q(\theta ) \geq 1$: Since agent's effective allocation is constant, otherwise, $i$ would simply choose the type with the least cost. Then, of course, 
		
		$$\dfrac{\partial u(\theta)}{\partial \theta} =1.$$
	\end{proof}

	\section{Proof of Proposition 1}\label{prop1}
	
	\begin{proof}
		
		We want to show that for each $\theta, \theta'\in [0,1]$, we have
		
		\begin{equation*}
			\begin{array}{llllll}
				
				& u(\theta) \geq  \theta \min\{ 1, q(\theta')  \}  - t(\theta') \\
				
				\iff & u(\theta) \geq  \theta \min\{ 1, q(\theta') \} +  \theta' \min\{ 1, q(\theta')  \} \\
				
				&-  \theta' \min\{ 1, q(\theta')  \} - t(\theta')  \\

				\iff &	u(\theta) \geq \theta \min\{ 1, q(\theta') \} - \theta' \min\{ 1, q(\theta')  \}\\
				
				& + u(\theta')   \\
				
				\iff &	u(\theta) - u(\theta') \geq ( \theta - \theta' ) \min\{ 1, q(\theta') \}   \\
				
				\iff &	\int_{\theta'}^{\theta} q^*( x ) dx \geq \int_{\theta'}^{\theta} q^*(\theta') dx
				
			\end{array}
		\end{equation*}
		
		Suppose $\theta> \theta'$. Since $q^*(\cdot)$ is increasing, $q^*(x)\geq q^*(\theta') $ for each $x\in \left[ \theta', \theta \right]$. Then, the last inequality above holds. Similar analysis holds for the case of $\theta < \theta'$.
		
	\end{proof}
	
	\section{Transformations of the Designer's Problem}\label{transform}
	
	We start with the problem in Equation \ref{problem} and make the following transformation:

	\begin{equation*}
		\begin{array}{llllll}
			& \int\limits_{[0,1]}  \int\limits^{\theta}_{0} q(x)dx f(\theta)   d\theta\\
			
			=& \int\limits_{[0,1]}  \int\limits^{\overline{\theta}}_{x} f(\theta)   d\theta q(x )  dx\\
			
			=& \int\limits_{[0,1]}  q(x) (1-F(x)) dx\\

			=& \int\limits_{[0,1]}   q(\theta) \left( \dfrac{(1-F(\theta))}{f(\theta)} \right)  f(\theta)   d\theta\\
		\end{array}
	\end{equation*}

	So, the second part of the objective function becomes:
	
	\begin{equation*}
		\begin{array}{llll}
			&   \int\limits_{[0,1]}   \theta q(\theta) f(\theta)  d\theta -    \int\limits_{[0,1]} q(\theta) \left(\dfrac{(1-F(\theta))}   {f(\theta)} \right)   f(\theta)   d\theta  \\
			
			=&  \int\limits_{[0,1]} \left(\theta q(\theta) -  q(\theta ) \dfrac{(1-F(\theta))}{f(\theta)} \right) f(\theta)   d\theta\\
			
			=&  \int\limits_{[0,1]} q(\theta) \left(\theta  - \dfrac{(1-F(\theta))}{f(\theta)} \right) f(\theta)   d\theta\\	
		\end{array}
	\end{equation*}
	
	Next we look at the first summand in the objective function above. Notice that inside is actually a constant, so it can be expressed as below:

	\begin{equation*}
		\begin{array}{llll}
			&\int\limits_{[0,1]} \left[ \int\limits_{\{ y|q(y)\leq 0 \}} q(x )dx \right] f(\theta)   d\theta\\
			
			= &   \int\limits_{\{ y|q(y)\leq 0 \}} q(x )dx\\
			
			= &  \int\limits_{[0,1]} q(x) \mathbbm{1} \{ q(x)\leq 0 \} dx\\

			= & \int\limits_{[0,1]} q(\theta) \dfrac{  \mathbbm{1} \{ q(\theta)\leq 0 \}}{f(\theta)} f(\theta)   d\theta\\
		\end{array}
	\end{equation*}
	
	Finally, the objective function can be written as:

	\begin{equation*}
		\begin{array}{llll}
			&  \int\limits_{[0,1]} q(\theta ) \dfrac{  \mathbbm{1} \{ q(\theta)\leq 0 \}}{f(\theta)} f(\theta)   d\theta  +  \int\limits_{[0,1]} q(\theta) \left(\theta  - \dfrac{(1-F(\theta))}{f(\theta)} \right) f(\theta)   d\theta\\
			
			& =	 \int\limits_{[0,1]} q(\theta ) \left[  \dfrac{  \mathbbm{1} \{ q(\theta)\leq 0 \}}{f(\theta)} + \left(\theta  - \dfrac{(1-F(\theta))}{f(\theta)} \right)  \right]f(\theta)   d\theta \\

		\end{array}
	\end{equation*}

	Hence, the revenue maximization problem can be expressed as
	
	\begin{equation*}
		\begin{array}{llllll}
		
			\displaystyle\max_{(q, t)}  & \int\limits_{\Theta} q(\theta ) \left[  \dfrac{  \mathbbm{1} \{ q(\theta)\leq 0 \}}{f(\theta)} + \left(\theta  - \dfrac{(1-F(\theta))}{f(\theta)} \right)  \right]f(\theta)   d\theta  \\

			\text{s. t. }&\\
			
			& q(\theta) \text{ is increasing in } \theta\\
			
			& q(\theta) \geq - 1\\
			
			& 0 \geq  \int\limits_{[0,1]} q(\theta) d\theta \\
			
		\end{array}
	\end{equation*}